%% file: main-writeup.tex
\documentclass[11pt]{article}

\usepackage[utf8]{inputenc}
\usepackage{amsmath,amssymb}
\usepackage{amsthm}
\usepackage{bbm}
\usepackage{bm}
\usepackage{latexsym}
\usepackage[noadjust]{cite}
\usepackage{graphicx}
\usepackage{hyperref}
\usepackage{xcolor}
\usepackage{dirtytalk}
\usepackage{mathtools}
\usepackage[margin=1in]{geometry}
\usepackage{thm-restate}
\usepackage{enumitem}
\usepackage{tikz-cd}
\usepackage[linesnumbered,ruled]{algorithm2e}
\usepackage{comment}
\usepackage{float}
\theoremstyle{plain}
\newtheorem{theorem}{Theorem}[section]
\newtheorem{lemma}[theorem]{Lemma}

\newtheorem{proposition}[theorem]{Proposition}

\newtheorem{corollary}[theorem]{Corollary}

\newtheorem{remark}[theorem]{Remark}
\newtheorem{definition}[theorem]{Definition}



\usepackage[nameinlink]{cleveref}

\newcommand{\D}{\mathcal{D}}

\newcommand{\ep}{\varepsilon}

\newcommand{\eps}{\varepsilon}

\title{Quality Control Algorithms for Pattern Counting}
\author{Cassandra Marcussen\thanks{School of Engineering and Applied Sciences, Harvard University, Cambridge, Massachusetts, USA. Email: cmarcussen@g.harvard.edu. This material is based upon work supported by the Air Force Office of Scientific Research under award number FA9550-23-F-0014 (NDSEG Fellowship). Supported in part by NSF Award 2152413 and a Simons Investigator Award to Madhu Sudan.} \\ \and Ronitt Rubinfeld\thanks{Computer Science and Artificial Intelligence Laboratory, MIT, Cambridge, Massachusetts, USA. Email: ronitt@csail.mit.edu. Supported by the NSF TRIPODS program (award DMS-2022448) and CCF-2310818.} \\ \and Madhu Sudan\thanks{School of Engineering and Applied Sciences, Harvard University, Cambridge, Massachusetts, USA. Email: madhu@cs.harvard.edu. Supported in part by a Simons Investigator Award, NSF Award CCF 2152413, and AFOSR award FA9550-25-1-0112.}}
\date{}

\begin{document}
\pagenumbering{gobble}
\maketitle

\begin{abstract}

In recent work, Marcussen, Rubinfeld, and Sudan introduced the notion of quality control problems, which aim to capture the task of determining if a given input is truly random (i.e., from a specified distribution). Formally, their goal is to accept typical inputs from the specified distribution while rejecting every input whose value of a specified statistic is far from the distributional baseline. This formalism captures a commonly applied empirical practice of using some collection of specified statistics as a proxy for the quality of randomness. Empirical algorithms, however, have not exploited the asymmetry in the definition of quality control problems, which require soundness guarantees in the worst-case while only seeking average-case completeness. Their work abstracted a problem definition emphasizing this asymmetry and used it to give efficient quality control algorithms for assessing the randomness of graphs. 

In this work, we introduce and study quality control problems over sequences, where the goal is to distinguish a sequence of i.i.d. characters from sequences where some specified pattern appears too often (or too infrequently) as a subsequence. We consider this problem in both the finite-alphabet setting and for real-valued sequences. We refer to the former setting as the pattern counting problem. In the latter case, the natural notion of a pattern is to consider the relative ordering of the characters in the subsequence, and we refer to this as the permutation pattern counting problem. Algorithms to approximately count (permutation) patterns of length $k$ in a worst-case sequence of length $n$ can provably require exponential in $k$ queries into the sequence (assuming random access to the sequence). In contrast, we show that by taking advantage of the asymmetry in the definition of quality control, we give algorithms that run in poly$(k)$ time (with dependence on $k$ at most $k^8$ up to logarithmic factors in $k$) to solve these problems. We also prove that any quality control algorithm (over some natural distributions) requires superlinear queries in $k$. 

Our algorithms build on and introduce new notions of quasirandomness of sequences and permutations, a field that aims to identify interesting deterministic criteria for ``certifying'' randomness of objects and prove equivalences between such criteria. Our work adds to this field notions suitable for probabilistic analysis, including exponentially strong concentration of pattern occurrences. Our work also defines a notion analogous to Thomason's notion of graph jumbledness for permutations and develops tools that allow algorithmic applications in our setting.
\end{abstract}

\newpage 

\tableofcontents

\newpage

\pagenumbering{arabic}
\setcounter{page}{1}
\section{Introduction}

 Determining whether a corpus of data looks random is a central goal throughout statistics, average-case algorithms, and property testing. In many of these settings, an empirically adopted strategy is to decide whether the dataset \textit{possesses some property} held by random instances (for example, has a statistic matching random instances), not actually whether the data was indeed generated randomly. Thus the focus turns to designing efficient algorithms to test the property. Recent work of Marcussen, Rubinfeld, and Sudan \cite{qualitycontrol}, however, points out that it suffices to solve a weaker problem than that of testing the property. Namely, while the testing algorithm should reject data that do not have the property, it does not need to accept \textit{every} input satisfying the property. It suffices to accept \textit{most} inputs satisfying the property.

 Specifically, the formalization of Marcussen, Rubinfeld, and Sudan \cite{qualitycontrol} is as follows: Given a statistic $\rho$ and a distribution $\D$ for which $\rho \approx 1$ for instances from $\D$ with high probability, the goal of a $(\D, \rho)$-quality control algorithm is to accept all but a small fraction of instances drawn from $\D$ (completeness) and reject any input such that $\rho \not \approx 1$ (soundness). Importantly, there is a gap between completeness and soundness, which allows the algorithm to reject a small fraction of instances that are drawn from $\D$, even if $\rho \approx 1$ for such instances. In \cite{qualitycontrol}, Marcussen, Rubinfeld, and Sudan focused on the setting where $\rho$ captures the counts of small motifs (such as $ k$-cliques) in graphs drawn from the Erd\H{o}s--R\'enyi random graph distribution. They considered the sublinear-time complexity of solving these problems, proving that $(\D, \rho)$-quality control can be provably superpolynomially more efficient than computing $\rho$ on worst-case instances.

 We study quality control for patterns over sequences also in the setting of sublinear algorithms. This is a natural and previously unexplored context for quality control. It captures the setting where an individual has access to a fixed, ordered sequence of scalar-valued data points, and wants to certify whether the data possesses some property held by i.i.d. data points. While we cannot verify whether a fixed set of data points is exactly i.i.d., via the quality control framework we can analyze whether the data points possess some statistical properties held by i.i.d. sequences. A reasonable choice is the counts of small subsequences --- either in the form of relative orderings or exact values, depending on whether the distribution is non-atomic\footnote{A non-atomic distribution is a probability distribution where no single specific outcome has a nonzero probability mass. For example, the uniform distribution on $[0, 1]$ is non-atomic.} or over a finite alphabet. We call these \textit{permutation patterns} for non-atomic distributions and \textit{patterns} for sequences over finite alphabets.

 We set up the corresponding quality control problem as follows. First, for a real-valued sequence $\sigma:[n] \to \mathbb{R}$, permutation pattern $\tau: [k] \to [k]$, and a non-atomic distribution $p$ over $\mathbb{R}$, if $\sigma \sim p^{\otimes n}$ (i.e., drawn from the product distribution over $n$ elements from $p$), then the number of copies of $\tau$ in $\sigma$\footnote{Here, a subset $S=\{i_1,\ldots,i_k\} \subseteq [n]$ with $i_1 < i_2 < \cdots < i_k$, is supposed to contain a copy of $\tau$ if for every $1 \leq j < \ell \leq k$ we have $\sigma(i_j) < \sigma(i_\ell) \Leftrightarrow \tau(j) < \tau(\ell)$.} is close to $\binom{n}{k} \frac{1}{k!}$ with high probability. Therefore, let $\rho(\sigma)$ be the number of copies of $\tau$ in $\sigma$, divided by $\binom{n}{k} \frac{1}{k!}$. Given these definitions, $(\D, \rho)$-quality control is now well-defined. We can set up a similar problem in the case where $\sigma:[n] \to [m]$ is a sequence over a finite domain and we count subsequences $P$ relative to a distribution $p$ over $[m]$.

 These problems are interesting because the (permutation) pattern count statistics are both natural but expensive to even approximate. For worst-case sequences, we can show that approximating each of the two statistics can require $\exp(k)$ queries (see \Cref{thm:worstcasepermutation} and \Cref{thm:worst-case-subsequences}), where $k$ is the length of the pattern or subsequence we approximately count. For permutation patterns, even detecting a single occurrence is NP-hard \cite{DBLP:journals/ipl/BoseBL98}. 
 
 We circumvent these exponential bounds (on the query complexity as well as the runtime) in the quality control framework. In the continuous case, given a permutation pattern $\tau$ of length $k$, we give a $\text{poly}(k)$-time quality control algorithm in the sublinear-time model distinguishing between typical sequences drawn from $p^{\otimes n}$ and sequences that have an atypical number of copies of $\tau$. Similarly, in the discrete case, for subsequences and distributions $p$ over $[m]$ with minimum probability $p_{\min}$, quality control algorithms require only $\text{poly}(mk/p_{\min})$ queries and time. 
 
 This is the first setting demonstrating that quality control can turn a problem requiring exponential time in the worst case into one solvable in fully-polynomial time. For quality control over sequences, this parameter $k$ is of primary concern, and achieving a polynomial dependence on it requires new techniques. Our query complexity bounds are obtained by adopting a notion of exponentially robust quasirandomness from \cite{qualitycontrol} and applying it to patterns in sequences. Our algorithms are nonadaptive. Achieving a $\text{poly}(k)$ runtime requires us to develop a strengthened notion of quasirandomness for permutations, which we call \textit{permutation jumbledness}. As opposed to existing notions of permutation quasirandomness \cite{cooper2006permutation, kral2013quasirandom, sixpatterns}, which are asymptotic, permutation jumbledness controls the count of a pattern at a multiplicative-error scale that is needed for polynomial-size samples. This strengthens permutation quasirandomness in a similar way to how Thomason's notion of graph jumbledness \cite{thomason1987pseudo} strengthens asymptotic notions of graph quasirandomness \cite{chung1989quasi}. 
 
 We now move to defining quality control, stating our results, and giving the ideas behind the results.

\subsection{Definitions and our results}

We begin by formally defining the query access model that our algorithms will have.

\begin{definition}[Query access model]
    On input $n \in \mathbb{N}$, an algorithm $\mathcal{A}$ has query access to a length-$n$ sequence $x$ of the following form: On a query $i \in [n]$, the algorithm receives the response $x_i$, the value at the $i$th index of $x$. Such query access is denoted by $\mathcal{A}^x(n)$.
\end{definition}

We now define quality control algorithms of \cite{qualitycontrol}. Quality control is defined with respect to a family of quality functions and a family of distributions.

\begin{definition}[Quality control algorithms \cite{qualitycontrol}]\label{def:quality control}
    Let $\mathcal{X} = \{\mathcal{X}_n\}_{n \in \mathbb{N}}$ be a family of input domains, and let
    $\rho = \{\rho_n\}_{n \in \mathbb{N}}$ be a family of \emph{quality functions}
    $\rho_n : \mathcal{X}_n \to \mathbb{R}_{\geq 0}$. Fix a constant $\ep > 0$ and a
    family of distributions $\mathcal{D} = \{\mathcal{D}_n\}_{n \in \mathbb{N}}$, where each
    $\mathcal{D}_n$ is a distribution over $\mathcal{X}_n$ satisfying
    \[
        \Pr_{x \sim \mathcal{D}_n}\big[\, |\rho_n(x) - 1| \leq \ep \,\big] \geq 1 - o(1).
    \]
    A randomized algorithm $\mathcal{A}$ that, on input $n \in \mathbb{N}$ and given query
    access to an element $x \in \mathcal{X}_n$, outputs
    $\mathcal{A}^{x}(n) \in \{\textsc{Accept}, \textsc{Reject}\}$ is a
    \emph{$(\mathcal{D}, \rho, \ep)$-quality control algorithm} iff:
    \begin{enumerate}
        \item (\textbf{Completeness}) For all sufficiently large $n \in \mathbb{N}$, with probability $1 - o(1)$ over
        the draw $x \sim \mathcal{D}_n$ and the internal randomness of $\mathcal{A}$,
        $\mathcal{A}^{x}(n) = \textsc{Accept}$.
        \item (\textbf{Soundness}) For every $n \in \mathbb{N}$ and every $x \in \mathcal{X}_n$
        with $|\rho_n(x) - 1| > \ep$, with probability at least $2/3$ over the internal
        randomness of $\mathcal{A}$, $\mathcal{A}^{x}(n) = \textsc{Reject}$.
    \end{enumerate}
    The \emph{$(\mathcal{D}, \rho)$-quality control problem} is that of designing a
    $(\mathcal{D}, \rho, \ep)$-quality control algorithm for every $\ep > 0$. $(\mathcal{D}, \rho)$-quality control is
    \emph{solvable in query complexity $q(n)$ and time complexity $t(n)$} if, for every
    $\ep > 0$, there is a $(\mathcal{D}, \rho, \ep)$-quality control algorithm that on
    input $n$ makes $O_{\ep}(q(n))$ queries to $x$ and runs in $O_{\ep}(t(n))$ time.
\end{definition}

For a family of quality functions $\rho = \{\rho_n\}_n$ and a fixed $\ep>0$, we call a family of distributions $\mathcal{D} = \{\mathcal{D}_n \}$ $(\rho,\ep)$-good when, for all sufficiently large $n \in \mathbb{N}$, $\mathbb{P}_{x \sim \mathcal{D}_n}[|\rho_n(x) - 1| \leq \ep] \geq 1 - o(1)$. (We call such a family of distributions ``good'' when $\rho$ and $\ep$ are clear from context.) Additionally, as in \cite{qualitycontrol}, throughout the paper, $o(1)$ in the runtime of an algorithm means a quantity that approaches $0$ along every sequence of parameter settings for which the stated runtime bound approaches infinity. Equivalently, for every $\delta>0$, the completeness error can be made at most $\delta$ by increasing the runtime by a factor that may depend on $\delta$. The $o(1)$ notation does not necessarily refer to the limit $n\to\infty$ with all other parameters fixed.

One of our main objects of study is the count of permutation patterns in real-valued sequences. We define permutation patterns and their counts formally below.

\begin{definition}[Permutation Pattern]
    Given a real-valued sequence $f: [n] \to \mathbb{R}$ and permutation (``permutation pattern'') $\tau: [k] \to [k]$, a {\em copy} of $\tau$ in $f$ is any subset of $k$ indices $i_1 < i_2 < \cdots < i_k$ such that, for $a, b \in [k]$, $f(i_a) < f(i_b)$ iff $\tau(a) < \tau(b)$.
    
    Let the count of $\tau$ in length-$n$ $f$ be $C_{\tau, n}(f)$. Define the {\em permutation pattern count quality} $Q_{\tau, n}$ of a length-$n$ sequence to be $$Q_{\tau, n}(f) = \frac{C_{\tau, n}(f)}{\binom{n}{k} \cdot \frac{1}{k!}}.$$
\end{definition}

For i.i.d. sequences, the count of any length-$k$ permutation pattern $\tau$ is around $\binom{n}{k} \cdot \frac{1}{k!}$ so $Q_{\tau, n} \approx 1$. To determine whether a worst-case sequence has $\approx \binom{n}{k} \cdot \frac{1}{k!}$ copies of a permutation pattern, any algorithm must make $\exp(k)$ queries (see \Cref{thm:worstcasepermutation}). On the other hand, we prove that quality control needs only $\text{poly}(k)$ queries!

\begin{theorem}[Permutation Pattern Counting]\label{thm:any-pattern-intro}
    Let $p$ be a non-atomic distribution over $\mathbb{R}$, $\tau$ be a length-$k$ permutation pattern, and $\ep > 0$. Define $\D = \{p^{\otimes n}\}_n$ and $Q_{\tau} = \{Q_{\tau, n}\}_n$. Then, $(\D, Q_\tau)$-quality control is achievable in $\text{poly}(k/\ep)$ queries and time. Additionally, for every length-$k$ permutation pattern $\tau$, the $(\D, Q_\tau)$-quality control problem requires $k^{1.5 - \delta}$ queries for each $\delta > 0$.
\end{theorem}

See \Cref{thm:qc-ub} and \Cref{thm:any-pattern-lb-in-section-1} for the statements of the upper and lower bounds, and \Cref{thm:any-pattern-lb-in-section-2} for the stronger lower bound for the increasing pattern.

When $p$ is a distribution over a finite domain, many indices of a sequence from $p$ can possess the same value, and thus it makes more sense to study subsequences (which we call ``patterns'') instead of permutation patterns. We define the count of patterns below.

\begin{definition}[Pattern]
    Given a sequence $f: [n] \to [m]$ and sequence $P: [k] \to [m]$, a {\em  copy} of $P$ in $f$ is any subset of $k$ indices $i_1 < i_2 < \cdots < i_k$ such that, for each $a \in [k]$, $f(i_a) = P(a)$.
    
    Let the count of $P$ in length-$n$ $f$ be $C_{P, n}(f)$. Define the {\em pattern count quality} $R_{P, n}$ of a length-$n$ sequence with respect to the uniform distribution over $[m]$ to be $$R_{P, n}(f) = \frac{C_{P, n}(f)}{\binom{n}{k} \cdot \frac{1}{m^k}}.$$
\end{definition}

Let $\mathcal{U}_m$ be the uniform distribution over $[m]$ and $P$ be a length-$k$ pattern. A sequence $x \sim \mathcal{U}_m^{\otimes n}$ has approximately $\binom{n}{k} \cdot \frac{1}{m^k}$ copies of $P$ with high probability, so $R_{P, n}(x) \approx 1$. Determining whether a worst-case sequence has $\approx \binom{n}{k} \cdot \frac{1}{m^k}$ copies of $P$ can require $\Omega(\ep m^k)$ queries for any algorithm (see \Cref{thm:worst-case-subsequences}). On the other hand, quality control for this problem requires only $\text{poly}(mk)$ queries:

\begin{theorem}[Pattern Counting]\label{thm:any-subsequence-intro}
    Let $\mathcal{U}_m$ be the uniform distribution over $[m]$, $P$ be a length-$k$ sequence, and $\ep > 0$. Define $\D = \{\mathcal{U}_m^{\otimes n}\}_n$ and $R_P = \{R_{P, n}\}_n$. Then, $(\D, R_P)$-quality control is achievable in $\text{poly}(mk/\ep)$ queries and time. Additionally, there exists a length-$k$ sequence $P$ such that $(\D, R_P)$-quality control requires $\Omega(k^2 m/\ep^2)$ queries.
\end{theorem}

\begin{remark} \begin{enumerate}
    \item \Cref{thm:any-subsequence-intro} can be extended to hold more generally for sequences sampled from distribution $p^{\otimes n}$ for arbitrary distributions $p$ over $[m]$. If the distribution has minimum probability $p_{\min}$ over any element, then the corresponding quality control problem for pattern counts has a $\text{poly}(k/(\ep p_{\min}))$ query complexity and time. See \Cref{thm:subsequences} and \Cref{thm:qc-lb-sequences-general1} for the formal statements of the upper and lower bounds.
    \item We remark that in both \Cref{thm:any-pattern-intro} and \Cref{thm:any-subsequence-intro} the algorithms given are uniform in that, given the distribution $p$, there is a single algorithm that takes the pattern as input and solves the quality control problem $(\{p^{\otimes n}\}_n, Q_\tau)$ and similarly for $(\{p^{\otimes n}\}_n, R_P)$.
\end{enumerate}
    
\end{remark}

\subsection{Proof overviews}

We now highlight the challenges and the main ideas that contribute to our results.

\subsubsection{Challenges}

We focus on the case of permutation patterns, as the challenges will be analogous for subsequences. Recall that our goal is to design an algorithm making $\text{poly}(k)$ queries and running in $\text{poly}(k)$ time that accepts most random sequences while rejecting sequences whose count of a length-$k$ permutation pattern $\tau$ is not $\approx \binom{n}{k} \frac{1}{k!}$. Roughly, our plan is to sample $\text{poly}(k)$ locations in the sequence, query them, and then decide whether to accept based on the responses. It is easy to see that the expected fraction of subsequences of our sample that take on a desired pattern is the same as in the global sequence.  However, there are two main barriers to obtaining a $\text{poly}(k)$ algorithm. We isolate each of them, as our main two ingredients will handle each one.

\paragraph{The variance barrier (queries).} 

We rule out natural sample-based strategies: First, a tester that samples size-$k$ subsets and checks whether each forms a copy of $\tau$ must draw $\exp(k)$ subsets before it even \textit{finds} a copy, since only a $1/k!$-fraction of $k$-subsets form a copy of $\tau$ in the accept case. 
Second, a tester that picks a random subsequence and uses its count of $\tau$ as an estimate of the global count inherits a variance that is so large that the subsequence must have length $\exp(k)$ for the estimate to concentrate. 
We prove indirectly that such testers cannot work by ruling out all worst-case algorithms: we show that determining whether a worst-case length-$n$ sequence has $\approx \binom{n}{k} \frac{1}{k!}$ copies of a length-$k$ permutation pattern $\tau$ requires $\exp(k)$ queries! (see \Cref{sec:worst-case-permutation}) 

We show
that it is possible to overcome this barrier by counting the number of occurrences of patterns other than just $\tau$, some possibly of other lengths, as we expand on later.

\paragraph{The runtime barrier (time).} Even with the conclusion above that it suffices to count the number of occurrences of some $\text{poly}(k)$ pattern $\tau'$ in $\text{poly}(k)$-sized subsequence, we are left with an algorithmic challenge. No method is known that counts length-$\text{poly}(k)$ patterns in a length-$\text{poly}(k)$ sequence in less than the naive $k^{\Omega(k)}$ time. Indeed, for permutation patterns, even distinguishing if the subsequence has zero copies of $\tau$ or a positive number is NP-hard \cite{DBLP:journals/ipl/BoseBL98}. Thus, a $\text{poly}(k)$ query bound alone does not suffice. We need to separately find a way to read the count of $\tau$ from the sample without ever enumerating patterns. Such a process should again use the asymmetry in the definition of quality control, since otherwise we would be solving an NP-complete problem in polynomial time.

\subsubsection{Permutation pattern upper bound}

 We develop two ingredients that each handle one of the barriers described above. First, we introduce \textit{exponentially robust quasirandomness}, which handles the query complexity by letting us certify the global count of $\tau$ from a random subsequence of length only $N = \text{poly}(k)$. This notion is inspired by \cite{qualitycontrol}, and our work shows how this concept may be broadly applicable. 
 Second, we define \textit{permutation jumbledness} to address the runtime. 
 Permutation jumbledness is an $O(N)$-time computable proxy for the pattern counts on a subsequence, which allows us to avoid explicit counting of patterns. This notion is inspired by work in quasirandomness of permutations \cite{cooper2005quasirandom}, 
 though our work needs significantly
 stronger quantitative bounds than in previous definitions --- and so we need to craft the definition carefully and analyze its properties from scratch.

\paragraph{Query complexity via exponentially robust quasirandomness.} To circumvent the variance barrier, we seek a deterministic property that is both inherited by random sequences and strong enough to imply that the count of $\tau$ on most $\text{poly}(k)$-length subsequences reflects the global count.

An idealized first candidate is the following: \textit{every} subsequence of the sequence $\sigma$ has a count of $\tau$ that matches what is expected from an i.i.d. sequence. This is exactly what we would need from an inductive argument on pattern lengths: control of counts at length $\ell$ should lead to control at length $\ell + 1$. It also identifies the right high-level quantity to track. However, crucially, this is not testable: there are exponentially many subsequences, and no small sample certifies all of them.

Thus, we relax this candidate to a high-probability, testable version, which we call \textit{exponentially robust quasirandomness} (see \Cref{def:exp-rob-qr}). A length-$n$ sequence $\sigma$ is exponentially robustly quasirandom with respect to a length-$k$ permutation pattern $\tau$ if, for each $s \approx \text{poly}(k)$, all but an $\exp(-\alpha)$ fraction of length-$s$ subsequences have a count of $\tau$ that matches what is expected from an i.i.d. sequence: $\approx \binom{s}{k} \frac{1}{k!}$. Here, the parameter $\alpha$ controls the exponentially small failure probability. Note that a random sequence is exponentially robustly quasirandom with high probability. Furthermore, if we could establish that $\sigma$ is exponentially robustly quasirandom, then we can accept $\sigma$. However, testing exponentially robust quasirandomness seems even harder than testing counts. But it turns out the strength of the definition allows an inductive test with $\text{poly}(k)$ query complexity. 
Roughly, we will aim to verify that 
$\sigma$ is exponentially robustly quasirandom with respect to every sub-permutation pattern $\tau'$ in $\tau$. For such a sequence $\sigma$, we apply a concentration bound (\Cref{lem:concentration-lemma}) to show that, with high probability over $S$, $\sigma|_S$ has the same fraction of occurrences of $\tau$ as the global fraction of occurrences of $\tau$. Thus, a simple test to verify exponentially robust quasirandomness with respect to $\tau$ given exponentially robust quasirandomness with respect to $\tau'$ is to count the number of occurrences of $\tau$ in $\sigma|_S$.

Crucially, however, this is only a reduction in the query complexity. Not only is counting the number of occurrences of $\tau$ in $\sigma|_S$ hard, but we even have the additional complication that we want to count the number of occurrences of $\tau'$ for exponentially many $\tau'$. The former issue is more challenging, and we explain our solution in the next paragraph. The latter can be dealt with by the same solution technique in the case of permutation patterns, but that approach does not work out in the case of (regular value-based) patterns; see \Cref{sec:pattern-overview}. This still leaves us with the former challenge, which we turn to next.

\paragraph{An efficiently checkable proxy: permutation jumbledness.} As asserted earlier, it is not feasible to count (or even approximate), for a worst-case choice of $\sigma|_S$, the number of occurrences of $\tau$ in $\sigma|_S$. In this section, we propose a proxy, i.e., a property $J$ (for jumbledness) such that most sequences satisfy $J$, while every sequence $\sigma|_S$ that satisfies $J$ has roughly the right count of $\tau$. Critically, it is algorithmically easy to determine if a sequence $\sigma|_S$ satisfies $J$. We describe the property $J$ next. 

Define $J$ as follows. View a permutation $\sigma$ over $[N]$ as the point set $\{(i, \sigma(i))\}_{i \in [N]} \subseteq [N] \times [N]$.  Tile $[N] \times [N]$ into a $(N/m) \times (N/m)$ ``coarse grid'' of $m \times m$ boxes, for $m = \Theta(N/k^2)$. For a random permutation, each box holds about $m^2/N$ points, with fluctuations of order $m/\sqrt{N}$, the standard deviation of the corresponding hypergeometric count. We call $\sigma$ $(m, \beta)$-jumbled if every box is within $\beta$ standard deviations of its mean. That is, $||\sigma(I_i) \cap I_j| - m^2/N| \leq \beta m/\sqrt{N}$ for the intervals $I_1, \dots, I_{N/m}$. In words, $\sigma$ spreads its points across the coarsened grid as evenly as a random permutation would.

Our notion is inspired by Cooper's notion of permutation quasirandomness \cite{cooper2005quasirandom}. Cooper's quasirandomness notion says that, for any intervals $I_1, I_2 \subseteq [N]$, $I_1 \times I_2$ should have $|I_1||I_2|/N$ points, up to a $o(N)$ deviation. Our definition is incomparable. On the one hand, it fixes a small number of non-overlapping boxes ($O(N)$ for a length-$N$ permutation, whereas Cooper's view would consider $\Omega(N^4)$ different boxes) and only requires the right count for these. On the other hand, we insist on a much smaller deviation.

Jumbledness turns out to be suitable for our needs. First, it is \textit{checkable in $O(N)$ time}, via binning the $N$ points into boxes and comparing each of the $(N/m)^2 \leq N$ box counts to $m^2/N$. Next, jumbledness \textit{holds for random permutations}, by hypergeometric tail bounds and a union bound over the $(N/m)^2$ boxes. Our main technical contribution here is to show that  it also \textit{certifies the counts}: for a suitable $m, \beta$ and $N = \text{poly}(k)$, $(m, \beta)$-jumbledness implies that the count of every subpattern of length $\leq k$ is within $(1 \pm \ep)$ of its value under i.i.d. sequences. We show this in \Cref{lem:soundness}. Our key insight is that the dominant contribution to the count comes from copies spread across $k$ distinct rows and $k$ distinct columns, while copies confined to fewer rows or columns contribute a  lower-order term. For this dominant term, jumbledness pins down each box's count tightly enough to control this contribution. We note that if we had worked with Cooper's notion of quasirandomness we would not be able to establish the third property and in particular the contribution from the dominant term would be harder to pin down.

\paragraph{Final tester.} Combining the two ingredients gives us a simple algorithm: sample a subsequence of $N = \text{poly}(k)$ indices, form the induced permutation, and test whether it is $(m, \beta)$-jumbled. Accept iff it is. Completeness is immediate, as the restriction of a random permutation to a random index set is again a random permutation, and thus jumbled with high probability. For soundness, suppose $Q_\tau(\sigma)$ is far from $1$. By the exponentially robust quasirandomness argument, there is a first length $\ell$ at which the length-$\ell$ subpattern count on the sample deviates from its i.i.d. value. But jumbledness would force all such counts to be correct, so the sample cannot be jumbled and the tester rejects. Thus, exponentially robust quasirandomness guarantees that a $\text{poly}(k)$-size sample retains the signal, and jumbledness extracts that signal in linear time. This circumvents both the variance barrier and the $\exp(k)$-time counting barrier.

\subsubsection{Permutation pattern lower bound}

Ideally, one would hope to bring the query complexity close to $k$, but our upper bound is only polynomial in $k$. Here we show that at least some superlinear query complexity is necessary: we obtain a nearly quadratic lower bound for the increasing pattern (\Cref{thm:any-pattern-lb-in-section-2}) and a nearly $k^{1.5}$ lower bound for every permutation pattern (\Cref{thm:any-pattern-lb-in-section-1}). Let us begin with the lower bound for the increasing pattern, before seeing how the construction can be extended to work for general patterns.

We consider the problem of distinguishing between $x \sim \mathcal{D}_1$ and $y \sim \mathcal{D}_2$, where $\mathcal{D}_1$ is i.i.d. and $\mathcal{D}_2$ has many more copies of the length-$k$ increasing pattern. For our choice of $\mathcal{D}_1$ and $\mathcal{D}_2$, we show that no deterministic algorithm making $o(k^{2 - \delta})$ queries (for any constant $\delta \in (0, 1)$) can distinguish $\mathcal{D}_1$ from $\mathcal{D}_2$. We elaborate on the two distributions below and then explain the idea behind the query complexity lower bound.

Since we are proving a lower bound for $(\{p^{\otimes n}\}_n, Q_\tau)$-quality control, where $\tau$ is the increasing pattern, we set $\mathcal{D}_1 = p^{\otimes n}$, which is the case the tester should accept.

We now describe $\mathcal{D}_2$. The goal is to make $\mathcal{D}_2$ look like $\mathcal{D}_1$ for the most part, so we design $\mathcal{D}_2$ to be a ``noisy'' copy of $\mathcal{D}_1$. Specifically, to sample $y \sim \mathcal{D}_2$, we first sample $y \sim \mathcal{D}_1$ and then modify the sequence randomly to create many copies of $\tau$ while still making the modifications hard to find. We define $\D_2$ as follows.

\paragraph{The distribution $\mathcal{D}_2$ (illustrative single-block case).}
    Fix a block length $s$ and a step size $\delta > 0$ small enough that all planted values lie in $(0,1)$ and occupy the intended rank range. A sequence $y \in \mathbb{R}^n$ is drawn from $\mathcal{D}_2$ as follows:
    \begin{enumerate}
        \item Sample $z \sim \mathcal{D}_1 = \text{Unif}(0,1)^{\otimes n}$.
        \item Sample an index $i \sim \text{Unif}([n - s + 1])$.
        \item For each $j \in \{0, 1, \dots, s - 1\}$, set $y[i + j] = \frac{i}{n} + j \delta$.
        \item For every position $\ell \notin \{i, i + 1, \dots, i + s - 1\}$, set $y[\ell] = z[\ell]$.
    \end{enumerate}

 $\mathcal{D}_2$ looks like an i.i.d. sequence except on one randomly located block of length $s$, inside which we plant a short (very slowly) increasing sequence. These planted values are chosen to play the role of the middle-ranked elements of the increasing pattern. We prove that sequences $y \sim \mathcal{D}_2$ have many copies of the length-$k$ increasing pattern. 
 
 The reason this gives a query lower bound is that the only way to reliably distinguish $\mathcal{D}_1$ from $\mathcal{D}_2$ is to find the planted block. Outside this block, the sequence is distributed exactly as under $\mathcal{D}_1$. Since the block has length only $s$ and is placed randomly, an algorithm needs $\Omega(n/s)$ queries to hit it with constant probability. We take $s$ as small as possible subject to $\mathcal{D}_2$ still having enough copies of $\tau$ (with high probability) to force rejection. For the increasing pattern this yields $s = \Theta\!\left(n/k^{2 - 2/t}\right)$ for an arbitrarily large constant $t$, and the $\Omega(k^{2 - \delta})$ bound of \Cref{thm:any-pattern-lb-in-section-2}.

A similar idea can be extended to general patterns $\tau$. See \Cref{sec:qc-lb-all}. The main change is that the $t$ planted middle-ranked elements no longer sit in a single contiguous block. The planted values form up to $t$ separate short runs rather than one long run. Consequently, the remaining $k-t$ ranks split into some number $b_{\tau,t}$ of nonempty intervals (\Cref{def:b-P}), and this quantity controls how many copies of $\tau$ the planting creates, and thus how large $s$ may be taken. The increasing pattern is the best case, with $b_{\tau,t}=2$, giving an exponent approaching $2$ as $t$ tends to infinity. In the worst case, $b_{\tau,t}=t+1$, giving an exponent approaching $\frac{3}{2}$ and resulting in the bound of \Cref{thm:any-pattern-lb-in-section-1}.

\subsubsection{Pattern upper bound}\label{sec:pattern-overview}

Now we turn to the setting of quality control of patterns over finite alphabets. Here the input $x$ and the pattern $P$ are sequences from $[m]^*$ and we want to distinguish between strings drawn from $p^{\otimes n}$, for some distribution $p$ over $[m]$, and those in which the number of occurrences of $P$ as subsequences is deviating significantly from the expected number for strings from $p^{\otimes n}$. While the template for our quality control algorithm is similar to the setting of permutation patterns, the details end up deviating significantly. We first describe the similarity before going into the differences.

As in the case of permutation patterns, we reduce quality control on the larger pattern to quality control of smaller subpattern counts within a sequence of length $s = \text{poly}(k/p_{\min})$, where $k$ is the length of the pattern $P$ and $p_{\min}$ is the probability of the minimum-probability element in $[m]$ under $p$.

Also as in the case of permutation patterns, we first define a stronger problem where we are required to reject even more sequences, though still a $o(1)$ fraction of sequences drawn from $p^{\otimes n}$. We then show that this stronger problem can be solved in polynomial time; and then show that it suffices to solve this stronger problem on a randomly chosen subsample of $s = \text{poly}(k/p_{\min})$ locations of the input $x$, where $k$ is the length of the pattern $P$.

The deviations from the permutation pattern setting start with the actual definition of the stronger problem. Here, we do not count the number of occurrences of all subpatterns $P'$ of $P$ in $x$ and instead count only one aggregate sum, for every $\ell \in [k]$, of the total number of occurrences of all subpatterns $P'$ of $P$ in the input $x$. This choice is affected by two considerations: (1) This is a problem for which we can get a worst-case polynomial-time algorithm using dynamic programming to compute the quality control function. (While this is stronger than what we need, it makes for a cleaner and more efficient result.) (2) While the information retrieved in this smaller problem is much less than that obtained in the stronger version used in the permutation pattern setting, we show that it suffices to carry out the exponentially robustly quasirandomness methodology, and thus sufficient to show that getting these counts on poly$(k)$ sized subsamples of $x$ suffice to solve the quality control problem on $x$. We open up our inductive proof hinted in the previous paragraph and notice that it suffices to know that, for every $\ell \in [k]$, the sum of occurrences of $\tau'$ of length $\ell$ in $\sigma|_S$ is roughly as in a random sequence for all but an exponentially small fraction of $S$. Thus, exponentially robust quasirandomness is now defined with respect to $\tau$ and $\ell$ (\Cref{def:exp-rob-qr}) and the induction step asserts that exponentially robust quasirandomness for level $\ell$ yields a simple test of exponentially robust quasirandomness at level $\ell+1$ (see the proof of \Cref{lem:exp-qr-main}).

\subsubsection{Pattern lower bound}

We prove an $\Omega(k^2/p_{\min})$ query lower bound, contrasting with our
$\widetilde{O}(k^5/p_{\min}^2)$ upper bound (\Cref{thm:subsequences}). To construct a lower bound for $(\{p^{\otimes n}\}_n, R_P)$-quality control, draw sequences $x \sim p^{\otimes n}$ and $y \sim q^{\otimes n}$ for some distribution $q \neq p$ over $[m]$ with many copies of some length-$k$ pattern $P$. We study the query complexity of distinguishing between these two cases.

First, let $i_{\min}$ be the minimum-probability element in $[m]$ under $p$. The pattern we consider is the constant pattern $P = i_{\min}^{\otimes k}$. Let $p_{\min} := p(i_{\min})$ be the probability of $i_{\min}$. Under $q$, let the probability weight of $i_{\min}$ be $q_{\min}$, where we choose $q_{\min} \gg p_{\min}$. Conditioned on selecting $[m] \setminus \{i_{\min}\}$, let $p$ and $q$ be identically distributed. Increasing the weight of $i_{\min}$ from $p_{\min}$ to $q_{\min}$ makes this count a constant factor larger under $q$, forcing rejection. We will choose $q_{\min}$ so that $q_{\min}/p_{\min} = 1 + \Theta(1/k)$, i.e., a relative perturbation of order $1/k$. This is the quantity that ultimately produces the $k^2$ in the bound.

To show that $p$ and $q$ are difficult to distinguish, we give a coupling of $x \sim p^{\otimes n}$ and $y \sim q^{\otimes n}$ as follows. For each position, sample $u \sim \text{Unif}(0, 1)$ together with an element $w$ from the conditional distribution of $p$ on $[m] \setminus \{i_{\min}\}$. In the $p^{\otimes n}$ sequence, place $i_{\min}$ if $u \leq p_{\min}$ and $w$ otherwise; in the $q^{\otimes n}$ sequence, place $i_{\min}$ if $u \leq q_{\min}$ and $w$ otherwise. Thus the two sequences agree except at positions where $p_{\min} < u \leq q_{\min}$, at which $q^{\otimes n}$ receives an extra copy of $i_{\min}$.

The copies of $i_{\min}$ appear according to Bern($p_{\min}$) under $p$ and Bern($q_{\min}$) under $q$. Thus, distinguishing $p$ from $q$ allows us to distinguish Bern($p_{\min}$) from Bern($q_{\min}$), so we need at least as many queries as needed to distinguish between these Bernoulli variables. This will yield the lower bound.

\subsubsection{Worst-case lower bounds} 
Our lower bounds for analyzing the count of permutation patterns and subsequences in worst-case sequences both rely on the same principle, and for now we discuss permutation patterns. Suppose we want to create a lower bound on detecting whether there are $(1 \pm \ep) \binom{n}{k} \frac{1}{k!}$ copies of the increasing pattern of length $k$, or if there are not. Let $Q$ be the increasing pattern of length $k-1$. Embed $Q$ into the worst-case sequence we are creating by dividing up the sequence into $k-1$ blocks and adding corresponding values to create many copies of $Q$. Such a sequence contains no copies of $\tau$. Then, we show that we can modify the sequence slightly and in a random fashion in order to create many copies of $\tau$. For permutation patterns, only a $1/\exp(k)$-fraction of indices need to be modified, and for subsequences, only a $1/m^k$ fraction. This yields the stated worst-case lower bounds.

\subsection{Related work}

\paragraph{Quality control:} As noted earlier our work builds on the work of Marcussen, Rubinfeld, and Sudan \cite{qualitycontrol}, who defined the notion of quality control algorithms and problems. They study quality control of randomness in graphs using motif counts as the quality function. Their notion relates to other previously studied notions in average-case complexity, property testing, and learning. We refer the reader to their work for a detailed description and comparisons. Our main contribution is to apply their definition to a setting other than graphs --- specifically to sequences. While many of the steps in our algorithms and analysis are inspired by their work, we note that each step requires careful definition and analysis. We also note that the quantitative benefit of working with quality control is clearer in this work --- understanding the benefits in their results requires a nuanced understanding of the growth of two parameters, whereas in our setting the difference is starkly clear: we go from $\exp(k)$ lower bounds on the worst case problem, to a $\text{poly}(k)$ upper bound in the quality control version.

\paragraph{Counting permutation patterns and subsequences:} There is a long line of research on both algorithms and hardness of exact counting \cite{DBLP:conf/isaac/AlbertAAH01, DBLP:journals/siamdm/AhalR08, DBLP:journals/moc/Kuszmaul18, DBLP:conf/isaac/0001G20, DBLP:conf/iwpec/JelinekOP21, DBLP:journals/algorithmica/BerendsohnKM21} and detection \cite{DBLP:journals/ipl/BoseBL98, fox2013stanley, DBLP:conf/soda/GuillemotM14, DBLP:conf/sosa/GawrychowskiR22} of patterns in sequences. Approximate counting of permutations has been more recently studied. Particularly, \cite{DBLP:conf/soda/Ben-EliezerMS26} obtains a deterministic near-linear time $(1 + \ep)$-approximation algorithm when $k \leq 5$. Recent work \cite{opler2026inapproximability} also establishes strong conditional inapproximability results for approximate counting of permutation patterns when $k$ grows. Our work instead studies the quality control setting for growing $k$.

Understanding whether the count of permutation patterns or subsequences matches the random case relates closely to notions of quasirandomness, described below.

\paragraph{Permutation quasirandomness:} 

Our notion of permutation jumbledness is inspired by the theory of quasirandomness of permutations. Cooper~\cite{cooper2006permutation} introduced quasirandom permutations and showed equivalences between several deterministic properties possessed by random permutations.  Since then, there has been continued and recent interest in minimal properties that imply permutation quasirandomness \cite{cooper2005quasirandom, hoppen2013limits, kral2013quasirandom, glebov2015finitely, chan2020characterization, kurevcka2022lower, kral2024forcing, sixpatterns}.

The literature on quasirandom permutations focuses on the asymptotic setting: they typically require pattern densities to converge to their random values as the permutation size tends to infinity. This yields results that are not at the right level of accuracy for our setting (but are in the same spirit). Thus we introduce the notion of ``permutation jumbledness.'' The jumbledness condition is efficiently testable and suffices to certify the relevant pattern counts on polynomial-size samples.

\subsection{Open problems}

The main open problem that remains is pinpointing the exact dependence on the parameters ($k, \ep, m, p_{\min}$) for both permutation patterns and subsequences. Our goal in this paper was to obtain fully-polynomial bounds for quality control of these problems, and now that this is achieved, the exact complexity becomes interesting as well.

Additionally, we believe that our notion of permutation jumbledness could be more widely useful in the study of quasirandom permutations and sequences. Exploring applications of permutation jumbledness is an interesting future direction.

\subsection{Paper organization}

The rest of the paper is organized as follows. In \Cref{sec:preliminaries}, we give preliminary propositions regarding bounded-difference concentration inequalities and properties of random permutations and sequences. In \Cref{sec:permutation-patterns}, we study quality-control upper and lower bounds for permutation patterns. We prove the $\text{poly}(k)$ upper bound (\Cref{thm:qc-ub}) in \Cref{sec:quality-control-permutation-ub} and the query lower bounds (\Cref{thm:any-pattern-lb-in-section-1,thm:any-pattern-lb-in-section-2}) in \Cref{sec:quality-control-permutation-lb}. In \Cref{sec:patterns}, we study quality-control upper and lower bounds for patterns (where values now matter instead of relative orderings), proving the $\text{poly}(mk)$ upper bound  (\Cref{thm:subsequences}) and lower bound (\Cref{thm:qc-lb-sequences-general1}). In \Cref{sec:worst-case-lower-bounds}, we prove $\exp(k)$ worst-case lower bounds for approximating the count of permutation patterns and patterns. In \Cref{sec:appendix}, we prove the propositions given in the preliminary section (\Cref{sec:preliminaries}).

\section{Preliminaries}\label{sec:preliminaries}

In this section, we state and prove various concentration inequalities and properties of random permutations and sequences. For the inequalities and properties that are not readily obtained in the existing literature, the proofs are included in \Cref{appendix:prelims}.

We begin with a few remarks and definitions we will rely on throughout our proofs.

\subsection{Remarks, definitions, and notation}

We remark about our treatment of floors and ceilings.

\begin{remark}
    Throughout this paper, we will ignore floors and ceilings in statements and proofs, as these will only contribute negligibly to the computations.
\end{remark}

We now define subsequences and subpatterns.

\begin{definition}\label{def:subsequence}
    Consider a length-$n$ sequence $\sigma$. $\tau$ of length $\ell$ is a subsequence of $\sigma$ (denoted $\tau \subseteq \sigma$) if there exists indices $i_1 < i_2 < \dots < i_{\ell} \subseteq [n]$ such that $\sigma_{i_j} = \tau_j$ for all $j \in [\ell]$.
\end{definition}

\begin{definition}\label{def:subpattern}
    If $\sigma$ is a pattern (being counted in a larger sequence), then any subsequence $\tau$ of $\sigma$ is called a \textit{subpattern} of $\sigma$, denoted $\tau \subseteq \sigma$. If $\sigma$ is a permutation pattern and $I=\{i_1<\cdots<i_\ell\}\subseteq[|\sigma|]$, then the subpattern induced by $I$ is the unique permutation $\operatorname{pat}_I(\sigma):[\ell]\to[\ell]$ such that
    $$\operatorname{pat}_I(\sigma)(a)<\operatorname{pat}_I(\sigma)(b) \Longleftrightarrow \sigma(i_a)<\sigma(i_b).$$
    We write $\tau\subseteq\sigma$ if $\tau=\operatorname{pat}_I(\sigma)$ for some $I$. When summing over subpatterns, different index sets $I$ are counted with multiplicity.
\end{definition}

\paragraph{Notation.} We use $p$ to refer to a probability distribution over a 1-dimensional space. We let $p^{\otimes n}$ represent the product distribution over $n$ entries distributed according to $p$. We let $\D = \{p^{\otimes n}\}_n$ be the family of product distributions over $n$ entries distributed according to $p$, for all $n \in \mathbb{N}$.

\subsection{Concentration inequalities}

We rely on a concentration inequality for functions that have a ``bounded difference'' property on a high-probability subset of the domain. We prove the following subset analogue of Proposition 2 of \cite{combes2024extension}, which is a high-probability bounded differences inequality.

\begin{proposition}[Subset analogue of Proposition 2 of \cite{combes2024extension}]\label{thm:combes}
    Let $f: \binom{[n]}{t} \to \mathbb{R}$ be a bounded, non-negative function, and $\mathcal{Y} \subseteq \binom{[n]}{t}$. Suppose the following ``bounded differences property'' holds for $f$ and $\mathcal{Y}$:

    For all $A, B \in \mathcal{Y}$,
    $$\left|f(A) - f (B)\right| \leq \frac{c}{2} |A \triangle B|$$
    for some constant $c$.

    Consider any $\delta > 0$. Let $S$ be a uniformly random element of $\binom{[n]}{t}$. Let $$q:= 1 - \mathbb{P}\left( S \in \mathcal{Y}\right),$$ and suppose $q \leq \frac{\delta}{2 \max(f)}.$

    Then, the following holds:
    $$\mathbb{P}\left( \left| f(S) - \mathbb{E}\left( f(S)\right)\right| \geq \delta\right) \leq q +  2 \exp\left(- \frac{\delta^2}{8 t c^2} \right).$$
\end{proposition}

The proof is in \Cref{appendix:prelims}.

We next give a Binomial concentration inequality.

\begin{proposition}\label{prop:binomial}
    Let $p \in [0, 1]$ be a constant. Then
    $$\mathbb{P}\left(\text{Binom}(N, p) = \lfloor Np \rfloor \right) \geq \frac{C_p}{\sqrt{N}}$$
    for some constant $C_p > 0$ depending only on $p$.
\end{proposition}

\begin{proof}
    If $p \in \{0,1\}$, the probability is $1$. Otherwise,
    \begin{equation}\label{eq:binomial}
       \mathbb{P}\left(\text{Binom}(N, p) = \lfloor Np \rfloor \right) = \binom{N}{\lfloor Np\rfloor} p^{\lfloor Np\rfloor} (1 - p)^{N-\lfloor Np\rfloor}.
    \end{equation}
    Applying Stirling's approximation, there is some constant $C_p > 0$ depending only on $p$ such that \Cref{eq:binomial} is at least $\frac{C_p}{\sqrt{N}}$.
\end{proof}

\subsection{Properties of random permutations and sequences}

\begin{proposition}\label{prop:random-permutation-P-copies}
    Consider a random permutation $\sigma: [n] \to [n]$ and any length-$k$ permutation pattern $\tau$. If $n-k\geq 2k^3$, then $\sigma$ has $(1 \pm \ep) \binom{n}{k} \frac{1}{k!}$ copies of $\tau$ with probability at least $1-\frac{2k^3}{(n-k)\ep^2}$. In particular, for fixed $k$, $\tau$, and $\ep$, this probability is $1-o_n(1)$ as $n\to\infty$.
\end{proposition}

See \Cref{appendix:prelims} for a proof.

We state a similar result for sequences.

\begin{proposition}\label{prop:min-weight-sequence-P-copies}
    Let $p$ be a distribution over $[m]$ such that the minimum probability of any element is $p_{\min}$. Consider any length-$k$ sequence $P = P_1 P_2 \dots P_k$ over $[m]$, and let $p_{i}$ be the weight that $p$ assigns $P_i$ for all $i \in [k]$. Let $\sigma \sim p^{\otimes n}$. For $n = \Omega\left(\frac{k^3}{ \ep^2 p_{\min}} \right)$, $\sigma$ has $(1 \pm \ep) \binom{n}{k} \prod_{i \in [k]} p_i$ copies of $P$ with probability at least $1 - \frac{1}{k}$. Furthermore, for fixed $k$, $P$, $p$, and $\ep$, this probability is $1-o_n(1)$ as $n \to \infty$.
\end{proposition}

See \Cref{appendix:prelims} for the proof.

We observe the following corollary for the uniform distribution.

\begin{corollary}\label{cor:random-sequence-P-copies}
    Let $\sigma \sim (Unif([m]))^{\otimes n}$ and $P$ be a length-$k$ sequence over $[m]$. For $n = \Omega\left(\frac{m k^3}{ \ep^2} \right)$, $\sigma$ has $(1 \pm \ep) \binom{n}{k} \frac{1}{m^k}$ copies of $P$ with probability at least $1 - \frac{1}{k}$.
\end{corollary}

\section{Permutation patterns}\label{sec:permutation-patterns}

In this section, we study quality control of sequences with respect to the count of permutation patterns. Our main upper bound, \Cref{thm:qc-ub}, is stated at the start of \Cref{sec:quality-control-permutation-ub} and proved at its end, after we develop its two ingredients. The matching lower bounds, \Cref{thm:any-pattern-lb-in-section-1,thm:any-pattern-lb-in-section-2}, are proved in \Cref{sec:quality-control-permutation-lb}.

For all $n \in \mathbb{N}$, let the count of $\tau$ in a length-$n$ sequence $f$ be $C_{\tau, n}(f)$. Define the permutation pattern count quality $Q_{\tau, n}$ to be $$Q_{\tau, n}(f) = \frac{C_{\tau, n}(f)}{\binom{n}{k} \cdot \frac{1}{k!}}.$$ For a distribution $p$, let $p^{\otimes n}$ be the product distribution over $n$ i.i.d. copies of $p$. 

Let $\D = \{p^{\otimes n}\}_n$ and $Q_{\tau} = \{Q_{\tau, n}\}_n$. We are interested in $(\D, Q_{\tau})$-quality control.

\subsection{Quality control algorithm}\label{sec:quality-control-permutation-ub}

We now turn to quality control algorithms, proving the following theorem.

\begin{theorem}\label{thm:qc-ub}
    Let $\tau$ be a length-$k$ permutation pattern, $\ep > 0$, and $p$ be a non-atomic distribution over $\mathbb{R}$. Define $\D = \{p^{\otimes n}\}_n$. Then, $(\D, Q_\tau)$-quality control is achievable in $O\left(\frac{k^8}{\ep^5}\log\left(\frac{k}{\ep}\right)\right)$ queries and time. 
\end{theorem}

We stress the lack of dependence on $n$ in the theorem statement, as well as the polynomial dependence on $k$. The lack of dependence on $n$ is relatively straightforward to see. One may sample length-$k$ subsequences roughly $\exp(k)$ times and estimate the frequency of the pattern $\tau$. However, the dependence on $k$ is not, a priori, clear, and achieving $\text{poly}(k)$ queries and time will require new ingredients suitable for the quality control framework.  

We first define a stronger target quality control problem, which is denoted $(\D, S_{\tau})$-quality control. 

\begin{definition}\label{def:S-tau}
    For each $N$, let $S_{\tau, N} = 1 + \max_{\tau' \subseteq \tau} (|1 - Q_{\tau', N}|)$.
    Let $S_\tau = \{S_{\tau, N}\}_N$.
\end{definition}

$S_{\tau, N} \approx 1$ certifies that the frequencies of \textit{all} subpatterns of $\tau$ are occurring with the right frequency in a given instance. Clearly, since $S_{\tau} \approx 1$ implies $Q_{\tau} \approx 1$, this is a harder problem to solve. However, the ability to reject instances that violate $Q_{\tau'} \approx 1$ for some subpattern $\tau' \subseteq \tau$ gives us room to avoid some instances that may otherwise offer obstacles to testing whether $Q_{\tau} \approx 1$ with small query complexity. Indeed, in the rest of this section, we give two ingredients that allow us to prove a strengthening of \Cref{thm:qc-ub}, which even solves the $(\D, S_{\tau})$-quality control problem with the desired $\text{poly}(k)$-query and time complexity. 

To get this solution for the $(\D, S_{\tau})$-quality control problem, we combine two independent ingredients. First, we show that $(\D, S_{\tau})$-quality control can be solved in linear time (and linear query complexity) on all sequences of length at least $\text{poly}(k)$, using a notion that we develop called ``permutation jumbledness.'' Permutation jumbledness turns out to be an even stronger property than requiring $S_{\tau} \approx 1$, but it is satisfied by random sequences and is computationally easy to verify. This yields the claimed time-efficient algorithm for the $(\D, S_{\tau})$-quality control problem (see \Cref{thm:jumbled-pattern} in \Cref{sec:jumbled}).

However, this algorithm does not circumvent the need to look at $\min(n, \exp(k))$-length subsequences in order to check $S_{\tau} \approx 1$. To get to this result, in \Cref{sec:exp-rob-qr}, we show a $\text{poly}(k)$-size sample preserves the counts of length-$k$ permutation patterns (\Cref{lem:exp-qr-main}), via a notion called exponentially robust quasirandomness. \Cref{sec:permutation-alg} combines these ingredients to prove \Cref{thm:qc-ub}. The resulting algorithm even has a succinct description, given as follows. On a length-$n$ permutation:

\begin{enumerate}
    \item Pick a random subsequence of length $N = \text{poly}(k)$.
    \item Represent a permutation $\sigma$ over $[N]$ as a point set $\{(i, \sigma(i))\}_{i \in [N]} \subseteq [N] \times [N]$. Tile $[N] \times [N]$ into a $(N/m) \times (N/m)$ ``coarse grid'' of $m \times m$ boxes. 
    \item Verify if the following property (called ``permutation jumbledness'') holds: each box holds about $m^2/N$ points, with fluctuations of order $m/\sqrt{N}$.
    \item If the property holds, output \textsc{Accept}. Else output \textsc{Reject}.
\end{enumerate}

\subsubsection{Analysis: permutation jumbledness}\label{sec:jumbled}

\paragraph{Flow of this section.} This section states and proves \Cref{thm:jumbled-pattern}, the first of two ingredients contributing to \Cref{thm:qc-ub}. We first define permutation jumbledness (\Cref{def:jumbled}), a geometric regularity condition on the point set of the permutation. We then establish that it is checkable in $O(N)$ queries and time (\Cref{lem:completeness}), satisfied by random permutations with high probability (\Cref{lem:completeness}), and certifies that all subpattern counts are close to their expectations for random permutations (\Cref{lem:soundness}).

Together, these give \Cref{thm:jumbled-pattern}, an $O(N)$-time quality control algorithm with the following guarantees. On a permutation pattern $\tau$, the quality control algorithm accepts length-$N$ i.i.d. sequences with high probability. The algorithm rejects any sequence, with high probability over the algorithm's randomness, for which there exists a subpattern $\tau' \subseteq \tau$ (as in \Cref{def:subpattern}) whose count is not $\approx \binom{N}{|\tau'|} \frac{1}{(|\tau'|)!}$ in the sequence.

We now give \Cref{thm:jumbled-pattern}.

\begin{theorem}\label{thm:jumbled-pattern}
    Let $\tau$ be a length-$k$ permutation pattern and let $p$ be a non-atomic distribution over $\mathbb{R}$. Let $\D = \{p^{\otimes N}\}_N$. Then $(\D, S_\tau)$-quality control is achievable in $O(N)$ queries and $O(N)$ time, as long as $N \gtrsim \frac{k^8}{\ep^5}\log \left( \frac{k}{\ep}\right)$.
\end{theorem}

This theorem is a partial step towards proving \Cref{thm:qc-ub}. Observe that it is a quality control algorithm with a similar guarantee to that statement, with a stronger quality function. However, it runs in linear time in the length of the sequence, which is not good enough for our claimed final result. As we will see in \Cref{sec:exp-rob-qr}, however, we will run the algorithm of \Cref{thm:jumbled-pattern} on a $\text{poly}(k)$-length subsequence of the original sequence we want to perform quality control on.

That is, this theorem performs quality control on the count of \textit{all subpatterns of $\tau$}. As we reduce to counting these subpatterns over a sequence of length $N = \text{poly}(k)$, on such a smaller sequence we need to consider quality control for this stronger property.

\paragraph{Defining permutation jumbledness.} We begin by defining ``permutation jumbledness.'' 

\begin{definition}[Permutation jumbledness]\label{def:jumbled}
   Consider a permutation $\sigma$ over $[N]$. Specifically, let $t$ be a positive integer and let $m=N/t$. Let $\ell_1,\ldots,\ell_{t} \in \{\lfloor m\rfloor, \lceil m\rceil\}$ such that $\sum_{j=1}^t \ell_j = N$ and $\ell_j \geq \ell_{j+1}$. Now let $I_j = [1+ \sum_{r=1}^{j-1} \ell_r, \sum_{r=1}^{j} \ell_r]$.   We say that a permutation $\sigma$ is $(m, \beta)$-jumbled if for all $i, j \in [t]$, 
   $$\left|\left|\sigma(I_i) \cap I_j \right| - \frac{|I_i||I_j|}{N}\right| \leq \beta \frac{m}{\sqrt{N}}.$$
\end{definition}

\begin{figure}[h]
    \centering
    \input{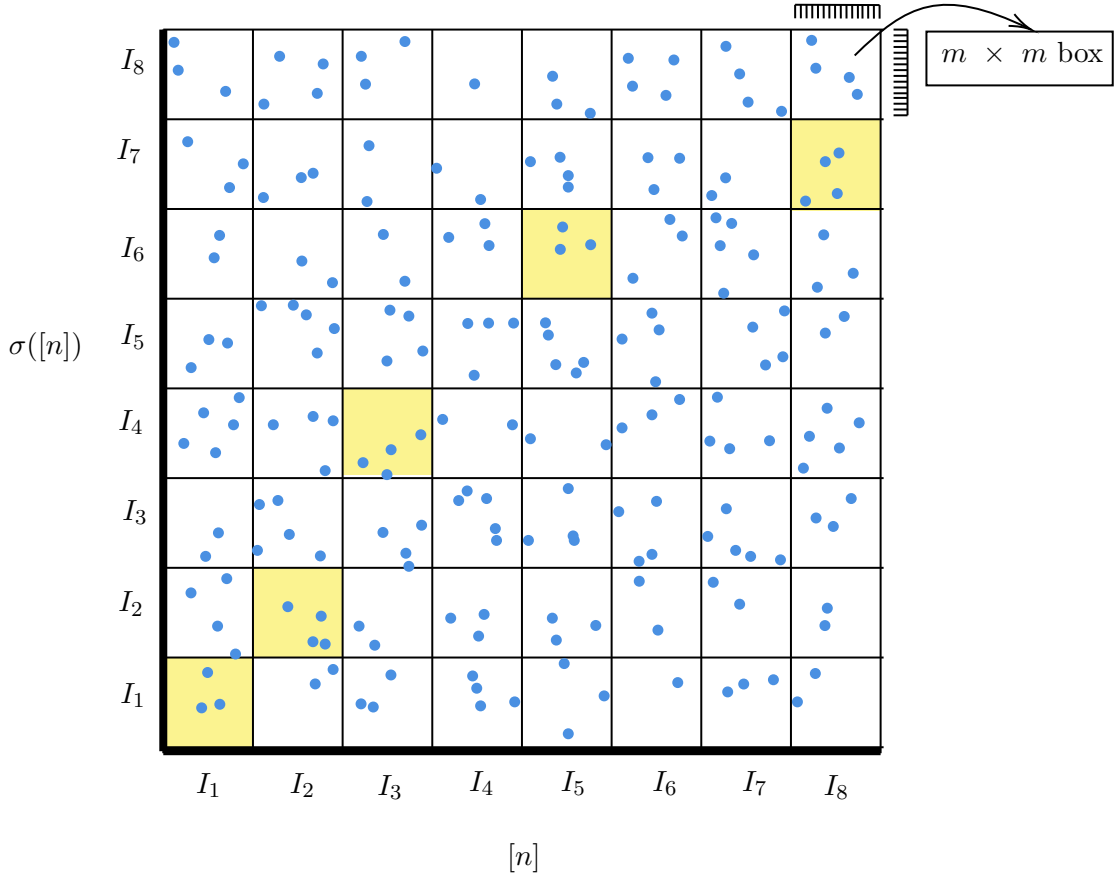}
    \caption{A random permutation $\sigma: [N]\to[N]$ is jumbled (\Cref{def:jumbled}) with high probability, meaning in a partition of $[N] \times [N]$ into $m$-by-$m$ squares, there are roughly $m^2/N$ points $(i, \sigma(i))$. Suppose we want to count the occurrences of the length-$k$ increasing pattern. As this figure portrays, look at any sequence of $k$ increasing blocks. Most occurrences of the pattern will occur in these sequences of increasing blocks, and an analogous statement can be made for general length-$k$ permutation patterns.}
    \label{fig:jumbled-figure}
\end{figure}

Intuitively, permutation jumbledness can be pictured as follows: View a permutation $\sigma:[N]\to[N]$ as the collection of points of the form $(i, \sigma(i)) \in [N] \times [N]$. Jumbledness considers a coarsening of $[N] \times [N]$ into $(N/m)^2$ blocks each of side length roughly $m$. In a random permutation $\sigma$, one would expect $m^2/N$ points to fall in each block. Jumbledness requires this to hold for each block (up to some small error). The fluctuation in the number of points is around $m/\sqrt{N}$, which is also what is expected from random permutations. Thus, permutation jumbledness characterizes the property that a random permutation sends elements $i$ to seemingly random elements $\sigma(i)$, and it captures both the expectation and fluctuations that a random permutation has in the counts of elements in the coarse grid.

\paragraph{Use of permutation jumbledness.} We prove that, on a length-$N$ sequence for $N = \text{poly}(k)$, permutation jumbledness can be checked in $O(N)$ queries and time. We also prove that permutation jumbledness implies all patterns of length $k$ are in $(1 \pm \ep) \binom{N}{k} \cdot \frac{1}{k!}$. This yields \Cref{thm:jumbled-pattern}.

The algorithm (\Cref{figure:jumbled-pattern}) of \Cref{thm:jumbled-pattern} will be used as a sub-routine in our final algorithm. We will reduce performing quality control of a length-$n$ sequence with respect to a length-$k$ pattern count to checking the pattern's count on a subsequence of length $N = \text{poly}(k)$, then utilize permutation jumbledness to check the pattern's count.

\paragraph{Proving \Cref{thm:jumbled-pattern}.} To prove \Cref{thm:jumbled-pattern}, we construct an algorithm that checks permutation jumbledness for some appropriate $m$ and $\beta$.

\begin{figure}[h]
\centering
\fbox{
\parbox{\dimexpr0.95\linewidth-2\fboxsep-2\fboxrule\relax}{
\textbf{Algorithm Jumbled-Count:} \\
On inputs: A length-$N$ permutation $\sigma$, a length-$k$ pattern $\tau$, parameter $\ep \in (0, 1)$, $N \gtrsim \frac{k^8}{\ep^5}\log \left( \frac{k}{\ep}\right)$
\begin{enumerate}
    \item Let $t = \left\lceil \frac{k^2}{a \ep} \right\rceil$ for a sufficiently small absolute constant $a$, and let $m = N/t$ be the typical interval length (so $m = \Theta(\ep N / k^2)$). Let
    $$\beta = b \sqrt{\frac{k^2}{\ep}\left( d \log k + \log \left( \frac{N}{m}\right) + \log N\right)}$$
    for a sufficiently large absolute constant $b$ and constant $d$.
    \item Partition $[N]$ into $t$ equitable intervals, so that each is of length $\lfloor N/t \rfloor$ or $\lceil N/t \rceil$. Call these intervals $I_1, I_2, \dots, I_t$.
    \item For each pair of intervals $I_i, I_j$, store a counter $\textsc{Count}_{ij}$, initialized to $0$.
    \item For each index $x \in [N]$, find the pair $(i, j)$ such that $(x, \sigma(x)) \in I_i \times I_j$ and increment $\textsc{Count}_{ij}$ by $1$.
    \item For each rectangle $I_i \times I_j$ in the resulting $t \times t$ grid over $[N] \times [N]$:
    \begin{enumerate}
        \item (Observe $\textsc{Count}_{ij} = \left|\sigma(I_i) \cap I_j \right|$.)
        \item If $$\left|\textsc{Count}_{ij} - \frac{|I_i||I_j|}{N}\right| > \beta \frac{m}{\sqrt{N}},$$
        output \textsc{Reject}. Else, proceed.
    \end{enumerate}
    \item Output \textsc{Accept}.
\end{enumerate}
}
}
\caption{Quality control algorithm for verifying whether a permutation pattern $\tau$'s count in $\sigma$ is in $(1 \pm \ep) \binom{N}{k}/k!$. This algorithm checks whether a permutation is jumbled (\Cref{def:jumbled}). By \Cref{lem:soundness}, if the permutation is jumbled, it has the correct count. Additionally, a random permutation is jumbled with high probability by \Cref{lem:completeness}.}
\label{figure:jumbled-pattern}
\end{figure}

For completeness (accepting random permutations), we first show that a random permutation is jumbled with high probability.

\begin{lemma}\label{lem:completeness}
    A random permutation $\sigma$ over $[N]$ is $(m, \beta)$-jumbled with probability at least $1 - O\left(\frac{1}{k^dN}\right)$, for $m \geq a \ep N/k^2$ and $\beta \geq b \sqrt{\frac{k^2}{\ep} \left( d \log k + \log \left( \frac{N}{m}\right) + \log N\right)}$ for sufficiently small absolute constant $a$ and sufficiently large absolute constant $b$.
\end{lemma}

\begin{proof}
    We first observe that, for a random permutation $\sigma$ and for any interval $I_a$ of length $m' \in \{\lfloor m\rfloor,\lceil m\rceil\}$, $\sigma(I_a)$ is a uniform random subset of $[N]$ of size $m'$. This follows since, under a random permutation $\sigma$, each set of size $m'$ equals $\sigma(I_a)$ with equal probability. Thus, for every $J\in\binom{[N]}{m'}$, the probability that $J=\sigma(I_a)$ is $1/\binom{N}{m'}$.

    Since $\sigma(I_i)$ is a uniformly random $|I_i|$-subset of $[N]$, $S_{i,j}=|\sigma(I_i)\cap I_j|$ has a hypergeometric distribution with expectation
    $$\mu_{i,j}:=\frac{|I_i||I_j|}{N}.$$
    A Chernoff bound for the hypergeometric distribution gives, for every $0\leq u\leq \mu_{i,j}$,
    $$\mathbb{P}\left[|S_{i,j}-\mu_{i,j}|\geq u\right]\leq 2\exp\left(-\frac{u^2}{3\mu_{i,j}}\right).$$
    Taking $u=\beta m/\sqrt{N}$, and using $|I_i|,|I_j|=m\pm O(1)$, we obtain, for the parameter range used below and after adjusting the absolute constant,
    $$\mathbb{P}\left[\left|\left|\sigma(I_i) \cap I_j \right| - \frac{|I_i||I_j|}{N}\right| \geq \beta \frac{m}{\sqrt{N}}\right] \leq 2 \exp\left( - c\beta^2\right)$$
    for an absolute constant $c>0$.
    Now take a union bound over all pairs. Since $t=N/m$,
    $$\mathbb{P}_{\sigma}\left[ \sigma \text{ is not $(m, \beta)$-jumbled}\right]$$
    $$\leq \mathbb{P}_{\sigma}\left[ \exists i, j \in [t] ~ : ~ \left|\left|\sigma(I_i) \cap I_j \right| - \frac{|I_i||I_j|}{N}\right| \geq \beta \frac{m}{\sqrt{N}}\right]$$
    $$\leq 2 \left(\frac{N}{m} \right)^2 \exp\left( - c\beta^2\right).$$
    This is at most $1/(k^dN)$ whenever
    $$c\beta^2\geq d\log k+2\log\left(\frac{N}{m}\right)+\log N+\log 2.$$
    The stated choice
    $$\beta \geq b \sqrt{\frac{k^2}{\ep} \left( d \log k + \log \left( \frac{N}{m}\right) + \log N\right)}$$
    satisfies this inequality for a sufficiently large absolute constant $b$. In particular, the failure probability is $o(1)$ as the runtime $O(N)$ tends to infinity.
\end{proof}

Next we show that jumbled permutations have $(1 \pm \ep) \binom{n}{k}/k!$ copies of length-$k$ permutations. This is the main lemma that allows us to argue soundness (rejecting permutations for which the count of permutation pattern $\tau$ is not in $(1 \pm \ep) \binom{n}{k}/k!$).

\begin{lemma}\label{lem:soundness}
    Suppose that a permutation $\sigma$ over $[N]$ is $(m, \beta)$-jumbled for $m = a \ep N/k^2$  and $\beta = b \sqrt{\frac{k^2}{\ep} \left( d \log k + \log \left( \frac{N}{m}\right) + \log N\right)}$ and for sufficiently small absolute constant $a$ and sufficiently large absolute constant $b$. Suppose $N \geq C \frac{k^8}{\ep^5}\left( d\log k + \log \left( \frac{k}{\ep}\right)\right)$ for a sufficiently large absolute constant $C$ (which depends on $a$ and $b$). Then, the number of copies of $\tau$ in $\sigma$ is in $(1 \pm \ep) \binom{N}{k}/k!$.
\end{lemma}

\begin{proof}
    We count the number of copies of $\tau$ in $\sigma$ by considering various cases for where the copies can lie relative to the grid considered by $(m, \beta)$-jumbledness.

    We consider all choices of how many rows and columns the copy of $\tau$ may be across. Let $r, c \in \{1, 2, \dots, k\}$. Let $C_\tau(r, c)$ be the number of copies of $\tau$ across any choices of $r$ rows and $c$ columns.

    We begin by bounding $C_\tau(k, k)$. By the definition of jumbledness, 

    $$C_\tau(k, k) \in \binom{N/m}{k}^2 \left( \frac{m^2}{N} \pm \beta \frac{m}{\sqrt{N}}\right)^k = \binom{N/m}{k}^2 \left( \frac{m^2}{N}\right)^k \left( 1 \pm \beta \frac{\sqrt{N}}{m}\right)^k.$$
    This is in:
    $$\in \binom{N/m}{k}^2 \left( \frac{m^2}{N}\right)^k \left( 1 \pm 2k \beta \frac{\sqrt{N}}{m}\right)$$
    when $$k \beta \frac{\sqrt{N}}{m} \leq 1,$$
    using that $(1-x)^k \geq 1 - kx$ and $(1+x)^k \leq 1 + 2kx$ for $x = \beta\frac{\sqrt{N}}{m}$ with $kx \leq 1$.
    Plugging in $\beta$, this tells us that we need
    $$k b \sqrt{\frac{k^2}{\ep} \left( d \log k + \log \left( \frac{N}{m}\right) + \log N\right)} \frac{\sqrt{N}}{m} \leq 1.$$
    We have set $m, N$ so that this holds.

    Next, we prove that $C_\tau(r, c)$ decays as $r$ and $c$ decrease. This will tell us that $C_\tau(k, k)$ is the maximum and that all $C_\tau(r< k, c< k)$ will contribute negligibly to the total count of $\tau$ in the permutation. Concretely, we prove
    \begin{equation}\label{eq:decreasing}
        C_\tau(r, c) \leq 3 \cdot \left(\frac{2mk^2}{N}\right)^{(k-r) + (k-c)} \cdot C_\tau(k, k).
    \end{equation}

    To show this, first we write out $C_\tau(r, c)$ via our jumbledness assumption: $$C_\tau(r, c) \leq \binom{N/m}{r}\binom{N/m}{c} \binom{k-1}{r-1} \binom{k-1}{c-1} \left( \frac{m^2}{N} + \beta \frac{m}{\sqrt{N}}\right)^k $$ $$= \binom{N/m}{r}\binom{N/m}{c} \binom{k-1}{r-1} \binom{k-1}{c-1} \left( \frac{m^2}{N}\right)^k \left( 1 + \beta \frac{\sqrt{N}}{m}\right)^k.$$

    We compare this to $C_\tau(k, k)$:
    $$\frac{C_\tau(r, c)}{C_\tau(k, k)} \leq \frac{\binom{N/m}{r}\binom{N/m}{c} \binom{k-1}{r-1} \binom{k-1}{c-1} \left( \frac{m^2}{N}\right)^k \left( 1 + \beta \frac{\sqrt{N}}{m}\right)^k}{\binom{N/m}{k}^2 \left( \frac{m^2}{N}\right)^k \left( 1 - \beta \frac{\sqrt{N}}{m}\right)^k}$$
    $$\leq \left(\frac{k^2}{N/m - k + 1} \right)^{(k-r) + (k-c)}\left(\frac{1 + \beta \frac{\sqrt{N}}{m}}{1 - \beta \frac{\sqrt{N}}{m}} \right)^k.$$
    When $N/m \geq 2k$ (for the first term in the product) and $\beta \sqrt{N}/m \leq 1/(3k)$ (for the second term), this is:
    $$\leq \left( \frac{2mk^2}{N}\right)^{(k-r) + (k-c)} \cdot 3,$$
    where we used $\left(\frac{1 + \beta\sqrt{N}/m}{1 - \beta\sqrt{N}/m}\right)^k \leq (1 + 3\beta\sqrt{N}/m)^k \leq (1 + 1/k)^k \leq e \leq 3$.

    \Cref{eq:decreasing} then implies that
    $$\sum_{(r, c) \neq (k, k)} C_\tau(r, c) \leq 3 C_\tau(k, k)\left( \frac{ 1}{(1 - \frac{2mk^2}{N})^2} - 1\right) \leq 3 \cdot 3 \frac{2mk^2}{N} C_\tau(k, k) = \frac{18mk^2}{N}C_\tau(k, k).$$
    When $mk^2/N \leq \delta/36$, this is at most
    $$\leq \delta C_\tau(k, k).$$

    Putting all of these steps together, we find that the total number of copies of $\tau$ in $\sigma$ is:
    $$\sum_{r = 1}^k \sum_{c = 1}^k C_\tau(r, c) \in [C_\tau(k, k), (1 + \delta) C_\tau(k, k)]$$
    $$\in \binom{N/m}{k}^2 \left( \frac{m^2}{N}\right)^k \left( 1 \pm 2k \beta \frac{\sqrt{N}}{m}\right) (1 + \delta).$$

    As long as $2k \beta \sqrt{N}/m \leq \gamma/3$ and $\delta \leq \gamma/3$, this is 
    $$\in (1 \pm \gamma) \binom{N/m}{k}^2 \left( \frac{m^2}{N}\right)^k.$$
    For $\gamma = \ep/2$ and $m \leq a \ep N/k^2$ for some $a$, we now prove that this is in:
    $$\in (1 \pm \ep) \binom{N}{k} \frac{1}{k!}.$$
    We use that $\binom{x}{y} = \frac{x^y}{y!} \prod_{i = 0}^{y - 1}\left( 1 - \frac{i}{x}\right)$:
    $$ (1 \pm \gamma) \binom{N/m}{k}^2 \left( \frac{m^2}{N}\right)^k = (1 \pm \gamma) \frac{(N/m)^{2k}}{(k!)^2} \prod_{i = 0}^{k - 1}\left( 1 - \frac{i m}{N}\right)^2 \left( \frac{m^2}{N}\right)^k$$
    $$= (1 \pm \gamma) \frac{N^k}{(k!)^2} \prod_{i = 0}^{k - 1}\left( 1 - \frac{i m}{N}\right)^2.$$
    When $m \leq a \frac{\ep N}{2 k^2}$ for small enough $a$ and $\gamma = \ep/2$, this is
    $$\in (1 \pm \ep) \prod_{i = 1}^{k-1} \left(1 - \frac{i}{N} \right) = (1 \pm \ep) \binom{N}{k} \frac{1}{k!}.$$

    To see this, observe:
    $$1 \geq \prod_{i = 0}^{k - 1}\left( 1 - \frac{i m}{N}\right)^2 \geq \left( 1 - \frac{km}{N}\right)^{2k}\geq 1 - \frac{2k^2 m}{N} \geq 1 - 2a \ep.$$
    A similar calculation works for $\prod_{i = 1}^{k-1} \left(1 - \frac{i}{N} \right)$, and thus for small enough $a$,
    $$ = (1 \pm O(a \ep)) \prod_{i = 1}^{k-1} \left(1 - \frac{i}{N} \right) \in (1 \pm \ep) \prod_{i = 1}^{k-1} \left(1 - \frac{i}{N} \right) = (1 \pm \ep) \binom{N}{k} \frac{1}{k!}.$$
\end{proof}

Given the lemmas we have now proven, we are ready to prove \Cref{thm:jumbled-pattern}.

\begin{proof}[Proof of \Cref{thm:jumbled-pattern}]
    We now prove the correctness and runtime of the algorithm of \Cref{figure:jumbled-pattern}.
    \\\\
    \noindent\textbf{Correctness.} For this algorithm to be a quality control algorithm, it must accept random permutations with probability $1 - o(1)$ (``completeness''). It must reject each permutation with the number of copies of the pattern not in $(1 \pm \ep) \binom{N}{k} \frac{1}{k!}$ with probability at least $2/3$ (``soundness''). By \Cref{lem:completeness}, the completeness failure probability is $O(1/(k^dN))=o(1)$ in the runtime $O(N)$. For soundness, we can prove something stronger: whenever there is a subpattern $\tau' \subseteq \tau$ whose number of copies in $\sigma$ is \textit{not} in $(1 \pm \ep) \binom{N}{|\tau'|}/|\tau'|!$, the algorithm \textit{always} outputs Reject. Indeed, apply \Cref{lem:soundness} with $r=|\tau'|\leq k$; the parameters chosen for $k$ satisfy the required bounds for every such $r$. Thus $\sigma$ is \textit{not} $(m, \beta)$-jumbled for $\beta = b\sqrt{\frac{k^2}{\ep}\left(d\log k + \log(N/m) + \log N\right)}$ and $m = a \ep N/k^2$ for some constants $a, b$. By \Cref{def:jumbled} of permutation jumbledness, this means that there exists some intervals $I_i, I_j$ (in the $N/m \times N/m$ grid of $[N] \times [N]$) such that 
    \begin{equation}\label{eq:bad-case}
        \left|\left|\sigma(I_i) \cap I_j \right| - \frac{|I_i||I_j|}{N}\right| > \beta \frac{m}{\sqrt{N}}.
    \end{equation}
    Algorithm \textsc{Jumbled-Count} stores $\left|\sigma(I_i) \cap I_j \right|$ via $\textsc{Count}_{ij}$. The algorithm will detect \Cref{eq:bad-case} and output Reject.

    We now justify our choice of parameters. In the proof, we needed the following parameter relationships:
    \begin{enumerate}
        \item $m = a \ep N/k^2$ 
        \item $\beta \geq b \sqrt{\frac{k^2}{\ep} \left( d \log k + \log \left( \frac{N}{m}\right) + \log N\right)} $ for some constant $b$, and constant $d$ of our choice
        \item $2k \beta \frac{\sqrt{N}}{m} \leq 1$
        \item $2k \beta \sqrt{N}/m \leq \ep/6$
        \item $2 m k^2/N \leq \ep/144$
        \item $N/m \geq 2k$.
    \end{enumerate}
    Solving for these, we can choose $m = \Theta( \ep N/k^2)$ and $\beta = b \sqrt{\frac{k^2}{\ep} \left( d \log k + \log \left( \frac{N}{m}\right) + \log N\right)}  = \Theta(\sqrt{\log (k/\ep)}) k/\sqrt{\ep}$ for the stated polynomial choice of $N$. 
    Therefore

$$N \gtrsim \frac{k^8}{\ep^5} \log \left( \frac{k}{\ep}\right)$$
    suffices.
    \\\\
    \noindent \textbf{Runtime.} The runtime of \textsc{Jumbled-Count} is $O(N + (N/m)^2)$. Since $m \geq \sqrt{N}$, this is $O(N)$.
\end{proof}

\subsubsection{Analysis: exponentially robust quasirandomness}\label{sec:exp-rob-qr}

In this section, we will prove and state \Cref{lem:exp-qr-main}. Now that we understand how to verify counts of patterns on length $N = \text{poly}(k)$ subsequences, we need to argue that it suffices to look at counts on a length $N = \text{poly}(k)$ subsequence in the first place. What we are looking for is a statement of the following. If the global count of a length-$k$ pattern is not $\approx \binom{n}{k} \frac{1}{k!}$, then the count on the subsequence of length $N = \text{poly}(k)$ will reflect this. On the other hand, if the global count \textit{is} $\approx \binom{n}{k} \frac{1}{k!}$, the count on the subsequence of length $N = \text{poly}(k)$ will also be $\approx \binom{N}{k} \frac{1}{k!}$.

We begin with notation. 

\begin{definition}
    For a permutation $\sigma$ and a length-$k$ permutation pattern $\tau$, and $\ell \in [k]$, define $\textsc{Count}_{\ell}(\tau, \sigma)$ be the count in $\sigma$ of all permutation patterns of length $\ell$ that are subpatterns of $\tau$, counted with multiplicity:

    $$\textsc{Count}_{\ell}(\tau, \sigma) = \sum_{\substack{Q \subseteq \tau\\|Q| = \ell}} \sum_{\substack{R \subseteq \sigma\\|R| = \ell}} 1(R \equiv Q),$$
    where $R \equiv Q$ means that $R$ forms a copy of permutation pattern $Q$.
\end{definition}

For a subset of indices $S \subseteq [n]$ and a length-$n$ permutation $\sigma$, let $\sigma|_S$ be the permutation on the indices $S$ induced by $\sigma$.

In \Cref{prop:random-permutation-aggregate-copies} of \Cref{appendix:aggregate-counts}, we state and prove the following, which captures the property that the count of a pattern on a random length $N = \text{poly}(k)$ subsequence gives sufficient information about the global count of the pattern in the sequence.

\begin{lemma}\label{lem:exp-qr-main}
    Suppose $\sigma$ is a length-$n$ permutation, $\tau$ is a length-$k$ pattern, and $\ep \in (0, 1)$. Sample a set $S \subseteq [n]$ of $s = \widetilde{O}(k^7/\ep^2)$ indices. Suppose there exists some $\ell \in [k]$ such that $\textsc{Count}_{\ell}(\tau, \sigma) \not\in (1 \pm \beta_\ell) \binom{k}{\ell}\binom{n}{\ell} \frac{1}{\ell!}$ (and consider the first such $\ell$). Then, with probability at least $1 - \exp(-\alpha_{\ell})$, $\textsc{Count}_{\ell}(\tau, \sigma|_S) \not\in (1 \pm 3\beta_{\ell}/4) \binom{k}{\ell}\binom{|S|}{\ell} \frac{1}{\ell!}$.
\end{lemma}

Here and throughout this section, the sample size hidden by $\widetilde{O}(k^7/\ep^2)$ includes a multiplicative factor of $\log(2+k/\ep)$.

We observe the following corollary.

\begin{corollary}\label{cor:qc-samples}
    Let $\tau$ be a length-$k$ permutation pattern, $\ep > 0$, and $p$ be a non-atomic distribution over $\mathbb{R}$. Define $\D = \{p^{\otimes n}\}_n$. Then, $(\D, S_\tau)$-quality control is achievable in $\widetilde{O}(k^7/\ep^2)$ queries.
\end{corollary}

We prove the corollary below.

\paragraph{The algorithm.} The quality control algorithm performs the following procedure. First, sample a subsequence on $s := \widetilde{O}(k^7/\ep^2)$ indices. Then, for each $\ell \in 2, 3, \dots, k$, evaluate whether $\textsc{Count}_{\ell}(\tau, \sigma|_S) \not\in (1 \pm 3\beta_\ell/4) \binom{k}{\ell}\binom{|S|}{\ell} \frac{1}{\ell!}$ in $O(s^k)$ time.

\begin{proof}[Proof of \Cref{cor:qc-samples} from \Cref{lem:exp-qr-main}]
    We argue completeness and soundness. While we are handling real-valued sequences drawn from $p^{\otimes n}$, since $p$ is non-atomic, replacing queries by their ranked values yields a uniformly random permutation. We thus argue about permutations.

    \paragraph{Completeness.} We prove that a random permutation is accepted with probability $1 - o(1)$. By \Cref{prop:random-permutation-aggregate-copies}, the probability that one of the tests fails is at most $O\left(\frac{k^4}{s\ep^2}\right)=O\left(\frac{1}{k^3\log(2+k/\ep)}\right)=o(1)$ in the runtime of the algorithm. Thus, for all $\ell = 2, 3, \dots, k$, $\textsc{Count}_{\ell}(\tau, \sigma|_S) \in (1 \pm 3\beta_\ell/4) \binom{k}{\ell}\binom{|S|}{\ell} \frac{1}{\ell!}$, and $\sigma$ is accepted.

    \paragraph{Soundness.} Suppose that $\sigma$ is a permutation such that the count of length-$k$ pattern $\tau$ in $\sigma$ is not in $(1 \pm \ep) \binom{n}{k} \cdot \frac{1}{k!}$. Then, there is a minimum $\ell \in \{2, 3, \dots, k\}$ such that $\textsc{Count}_{\ell}(\tau, \sigma) \not\in (1 \pm \beta_\ell) \binom{k}{\ell}\binom{n}{\ell} \frac{1}{\ell!}$. Then, by \Cref{lem:exp-qr-main}, on a random subset $S$ of size $N \geq s := \widetilde{O}(k^7/\ep^2)$, with probability at least $1-\exp(-\alpha_\ell)$, $\textsc{Count}_{\ell}(\tau, \sigma|_S) \not\in (1 \pm 3\beta_\ell/4) \binom{k}{\ell}\binom{|S|}{\ell} \frac{1}{\ell!}$. In this case, the quality control algorithm rejects.
\end{proof}

We follow the framework introduced by \cite{qualitycontrol} of exponentially robust quasirandomness to prove \Cref{lem:exp-qr-main}, adapted for our setting of parameters and objects. Exponentially robust quasirandomness refers to properties that capture quasirandomness of objects -- which means deterministic properties possessed by random instances. Importantly, exponentially robust quasirandomness is a relaxation of typical quasirandomness, as it requires a property to hold with all but exponentially small probability (say, over choices of portions of the object) instead of with probability 1.

We define a suitable notion of exponentially robust quasirandomness for permutations in our context.

\begin{definition}[Exponentially robust quasirandomness of permutations]\label{def:exp-rob-qr}
    For a length-$k$ permutation pattern $\tau$ and any $\ell \in [k]$, a permutation $\sigma$ is $(\alpha, \beta, s_0, \ell)$-exponentially robustly quasirandom with respect to $\tau$ if, for all $s \geq s_0$,
    $$\mathbb{P}_{S : |S| = s}\left(\textsc{Count}_{\ell}(\tau, \sigma|_S) \not \in (1 \pm \beta) \frac{\binom{k}{\ell}\binom{s}{\ell}}{\ell!} \right) \leq \exp(- \alpha).$$

\end{definition}

Note that when $\ell=k$, $(\alpha, \beta, s_0, \ell)$-exponentially robustly quasirandomness with respect to $\tau$ implies that $\sigma$ has roughly the right count of $\tau$ in all but an exponentially small fraction of subsequences $S$ of size $s$. When $\ell <k$ the definition is a bit more complex and may roughly be thought of as asserting that for every subpattern $\tau'$ of length $\ell$ of $\tau$,  $\sigma$ is $(\alpha, \beta, s_0, \ell)$-exponentially robustly quasirandom with respect to $\tau'$. (Of course, this is not formally true and the actual condition is a weaker one.) 

We prove \Cref{lem:exp-qr-main} by induction on the pattern length $\ell$, where
each step of the induction is made possible by the ``concentration lemma'' of \Cref{lem:concentration-lemma}.

The \textit{concentration lemma} is the
high-probability step: it states that if $\sigma$ is exponentially robustly
quasirandom for the count of length-$\ell$ subpatterns of $\tau$, then a small
sampled sub-permutation of $\sigma$ has its count of length-$(\ell+1)$
subpatterns tightly concentrated around the expectation, with all but
exponentially small failure probability. This lets us detect from the sample
alone whether the length-$(\ell+1)$ count matches the random-permutation case.

We then chain these steps together with \textit{induction}. Concentration at level
$\ell$ establishes exponentially robust quasirandomness at level $\ell$ (which implies that its count is close to the expectation whp). This is exactly the hypothesis the
concentration lemma needs at level $\ell+1$. Starting from the trivial base case
$\ell = 1$ and iterating up to $\ell = k$ propagates control from single elements
all the way to the full pattern $\tau$.

\begin{definition}\label{def:parameters-exprobustqr}
Set $\alpha_1 = \infty$ and $\beta_1 = 0$.

For $2 \leq \ell \leq k$, define
$$
    \beta_\ell := \frac{\ep}{4}\cdot 2^{(k-\ell)/k}.
$$
(Observe that $\beta_\ell=\Theta(\ep)$ for every $\ell\in\{2,\dots,k\}$, and in particular
$\beta_k \leq \ep$.)

For $2 \leq \ell \leq k$, define
$$
    \alpha_\ell := (k-\ell)\log 3 + k\log k + \log(3k/\ep).
$$

Let
$$
    s = \widetilde{O}(k^7/\ep^2).
$$
\end{definition}

\begin{lemma}[Concentration lemma]\label{lem:concentration-lemma}
    Suppose a permutation $\sigma: [n] \to [n]$ is $(\alpha_{\ell}, 2\beta_{\ell}, s, \ell)$-exponentially robustly quasirandom with respect to $\tau$. Then, for $\alpha_{\ell+1}, \beta_{\ell+1}, s$ as in \Cref{def:parameters-exprobustqr},

    $$\mathbb{P}_{S, |S| = s}\left(\left|\textsc{Count}_{\ell+1}(\tau, \sigma|_S)  - \mathbb{E}\left(\textsc{Count}_{\ell+1}(\tau, \sigma|_S) \right)\right|\geq  \frac{\beta_{\ell+1}}{4} \frac{\binom{k}{\ell+1}\binom{s}{\ell+1}}{(\ell+1)!} \right) \leq \exp(-\alpha_{\ell + 1}).$$
\end{lemma}

\begin{proof}[Proof of \Cref{lem:concentration-lemma}]
    We want to upper bound:
    \begin{equation}\label{eq:ub-2}
       \mathbb{P}_{S, |S| = s}\left(\left|\textsc{Count}_{\ell+1}(\tau, \sigma|_S)  - \mathbb{E}\left(\textsc{Count}_{\ell+1}(\tau, \sigma|_S) \right)\right|\geq  \frac{\beta_{\ell+1}}{4} \frac{\binom{k}{\ell+1}\binom{s}{\ell+1}}{(\ell+1)!} \right).
    \end{equation}

    We want to apply the high-probability-bounded-differences inequality given in \Cref{thm:combes} to $\textsc{Count}_{\ell+1}(\tau, \sigma|_S)$.
    \\\\
    \textit{Step 1: Verifying the conditions of \Cref{thm:combes}.}
    We first need to verify that the conditions of \Cref{thm:combes} hold in our context. The exponentially robust quasirandom assumption tells us that: 
    \begin{equation}\label{eq:ub-1}
      \mathbb{P}_{S : |S| = s}\left(\textsc{Count}_{\ell}(\tau, \sigma|_S) \not \in (1 \pm 2\beta_{\ell}) \frac{\binom{k}{\ell}\binom{s}{\ell}}{\ell!} \right) \leq \exp(- \alpha_{\ell}).  
    \end{equation}

    In applying \Cref{thm:combes}, let $\mathcal{Y}$ be the set consisting of sets $S$ of $s$ indices such that $\textsc{Count}_{\ell}(\tau, \sigma|_S) \in (1 \pm 2\beta_{\ell}) \frac{\binom{k}{\ell}\binom{s}{\ell}}{\ell!}$. Let $q := 1 - \mathbb{P}(S \in \mathcal{Y})$. By definition of exponentially robust quasirandomness, $q \leq \exp(-\alpha_{\ell})$. For \Cref{thm:combes} to apply, we need to verify that
    $$q \leq \frac{\delta}{2 \max(f)},$$
    where $\delta = \frac{\beta_{\ell+1}}{4}\frac{\binom{k}{\ell+1}\binom{s}{\ell+1}}{(\ell+1)!}$ (by choice) and $f = \textsc{Count}_{\ell+1}(\tau, \sigma|_S)$.

    Since $\max(f) \leq \binom{k}{\ell+1} \binom{s}{\ell+1}$,
    it suffices to verify that 
    $$\exp(-\alpha_{\ell}) \leq \frac{\delta}{2 \binom{k}{\ell+1} \binom{s}{\ell+1}}.$$
    This is satisfied by our choice of $\alpha_{\ell}$ from \Cref{def:parameters-exprobustqr}.

    For use later on in the proof, we also note that $q \leq \exp(-\alpha_{\ell}) = \frac{1}{3}\exp(-\alpha_{\ell + 1})$, since $\alpha_{\ell} = \alpha_{\ell + 1} + \log 3$.
    \\\\
    \textit{Step 2: Bounded difference analysis.} Every length-$(\ell+1)$ subpattern occurrence that is added or eliminated when one sampled index is replaced contains that index. Deleting it produces a length-$\ell$ subpattern occurrence. Each length-$\ell$ subpattern of $\tau$ can be extended using at most $k-\ell$ additional positions of $\tau$, and we account for occurrences containing either the removed or the added sampled index.
        
    By the assumption of exponentially robust quasirandomness, in the ``high probability'' case, this number of subpatterns of length $\ell$ is in $(1 \pm 2\beta_{\ell}) \frac{\binom{k}{\ell}\binom{s}{\ell}}{\ell!}$. Therefore, the bounded difference is upper-bounded by: 
       $\frac{2(k-\ell)\binom{k}{\ell}\binom{s}{\ell}}{\ell!}$.
    \\\\
    \textit{Applying the bounds.}
    Applying \Cref{thm:combes} and the analysis above, we find that, for $\delta = \frac{\beta_{\ell+1}}{4}\frac{\binom{k}{\ell+1}\binom{s}{\ell+1}}{(\ell+1)!}$ and $c = 2(k-\ell)\binom{s}{\ell} \binom{k}{\ell}/\ell!$, \Cref{eq:ub-2} is:
    $$\leq q + 2\exp\left(- \frac{\delta^2}{8s c^2} \right) \leq \frac{1}{3}\exp(-\alpha_{\ell + 1}) + 2 \exp\left(- \Omega\left(\beta_{\ell+1}^2\frac{(s - \ell)^{2}}{s (\ell + 1)^{6}}\right)\right),$$
    where the first term used $q \leq \exp(-\alpha_{\ell}) = \frac{1}{3}\exp(-\alpha_{\ell + 1})$.

    It therefore suffices to bound the second term by $\frac{2}{3}\exp(-\alpha_{\ell + 1})$, for which it suffices that the concentration exponent is at least $\alpha_{\ell}$, i.e.,
    $$\beta_{\ell+1}^2\frac{(s - \ell)^2}{s (\ell + 1)^6} \geq (k - \ell) \log 3 + k \log k + \log(3k/\ep).$$ This is satisfied for $s = \widetilde{\Omega}(k^7/\ep^2)$.
    
    Thus, overall, \Cref{eq:ub-2} is upper-bounded by $\frac{1}{3}\exp(-\alpha_{\ell + 1}) + \frac{2}{3}\exp(-\alpha_{\ell + 1}) = \exp(-\alpha_{\ell  + 1})$, yielding \Cref{lem:concentration-lemma}.
\end{proof}
We also note the following fact about a sufficient condition for being exponentially robustly quasirandom.

\begin{lemma}\label{lem:ok-global-counts-exp-qr}
    Suppose that a permutation $\sigma$ satisfies the property that, for some $\ell \leq k$ and all $\ell' \leq \ell$, the global count $\textsc{Count}_{\ell'}(\tau, \sigma)$ is $\in (1 \pm \beta_{\ell'}) \binom{k}{\ell'} \binom{n}{\ell'} \frac{1}{(\ell')!}$. Then, $\sigma$ is $(\alpha_{\ell}, 2\beta_{\ell}, s, \ell)$-exponentially robustly quasirandom with respect to $\tau$.
\end{lemma}

\begin{proof}
We prove this statement via induction using \Cref{lem:concentration-lemma}, following the proof idea of Claim 3.14 of \cite{qualitycontrol}. \\\\
\textit{Base case.} Let $\ell = 1$. Any permutation is $(\infty, 0, s, 1)$-exponentially robustly quasirandom with respect to $\tau$ for any pattern $\tau$, because exponentially robust quasirandomness for $\ell = 1$ concerns the number of individual values chosen in subsets, which is always $s$.
\\\\
\textit{Inductive hypothesis.} Suppose that $\sigma$ is $(\alpha_{\ell - 1}, 2\beta_{\ell - 1}, s, \ell - 1)$-exponentially robustly quasirandom with respect to $\tau$.
\\\\
\textit{Inductive argument.} Assume the inductive hypothesis holds. Since $\sigma$ is $(\alpha_{\ell-1}, 2\beta_{\ell-1}, s, \ell - 1)$-exponentially robustly quasirandom with respect to $\tau$, by \Cref{lem:concentration-lemma},
$$\mathbb{P}_{S, |S| = s}\left(\left|\textsc{Count}_{\ell}(\tau, \sigma|_S)  - \mathbb{E}\left(\textsc{Count}_{\ell}(\tau, \sigma|_S) \right)\right|\geq  \frac{\beta_{\ell}}{4} \frac{\binom{k}{\ell}\binom{s}{\ell}}{\ell!} \right) \leq \exp(-\alpha_{\ell}).$$
Since the global count $\textsc{Count}_{\ell}(\tau, \sigma)$ is in $(1 \pm \beta_{\ell}) \binom{k}{\ell} \binom{n}{\ell} \frac{1}{(\ell)!}$,
$$\mathbb{E}\left(\textsc{Count}_{\ell}(\tau, \sigma|_S) \right) \in (1 \pm \beta_{\ell}) \frac{\binom{k}{\ell}\binom{s}{\ell}}{\ell!}.$$
This implies that
$$\mathbb{P}_{S : |S| = s}\left(\textsc{Count}_{\ell}(\tau, \sigma|_S) \not \in (1 \pm 2\beta_{\ell}) \frac{\binom{k}{\ell}\binom{s}{\ell}}{\ell!} \right) \leq \exp(- \alpha_{\ell}),$$
because $\beta_{\ell} + \beta_{\ell}/4 \leq 2\beta_{\ell}$.

Thus, by definition of exponentially robustly quasirandom (\Cref{def:exp-rob-qr}), $\sigma$ is also $(\alpha_{\ell}, 2\beta_{\ell}, s, \ell)$-exponentially robustly quasirandom with respect to pattern $\tau$.
\end{proof}

We are now ready to prove \Cref{lem:exp-qr-main}.

\begin{proof}[Proof of \Cref{lem:exp-qr-main}]
  Assume there exists some $\ell \in [k]$ such that $\textsc{Count}_{\ell}(\tau, \sigma) \not\in (1 \pm \beta_{\ell}) \binom{k}{\ell}\binom{n}{\ell} \frac{1}{\ell!}$ (and consider the first such $\ell$). Since we are considering the first such $\ell$, for all $\ell' < \ell$, $\textsc{Count}_{\ell'}(\tau, \sigma) \in (1 \pm \beta_{\ell'}) \binom{k}{\ell'}\binom{n}{\ell'} \frac{1}{(\ell')!}$. Thus, by \Cref{lem:ok-global-counts-exp-qr}, $\sigma$ is $(\alpha_{\ell - 1}, 2\beta_{\ell - 1}, s, \ell - 1)$-exponentially robustly quasirandom with respect to $\tau$. 

  Because $\sigma$ is $(\alpha_{\ell - 1}, 2\beta_{\ell - 1}, s, \ell - 1)$-exponentially robustly quasirandom with respect to $\tau$, by \Cref{lem:concentration-lemma}, with probability at least $1 - \exp(-\alpha_\ell)$,
  $$\left|\textsc{Count}_{\ell}(\tau, \sigma|_S)  - \mathbb{E}\left(\textsc{Count}_{\ell}(\tau, \sigma|_S) \right)\right| < \frac{\beta_{\ell}}{4} \frac{\binom{k}{\ell}\binom{s}{\ell}}{\ell!}.$$
  Since
  $$\mathbb{E}\left(\textsc{Count}_{\ell}(\tau, \sigma|_S) \right) = \textsc{Count}_{\ell}(\tau, \sigma)\frac{\binom{s}{\ell}}{\binom{n}{\ell}},$$
  and $\textsc{Count}_{\ell}(\tau, \sigma) \not\in (1 \pm \beta_{\ell}) \binom{k}{\ell}\binom{n}{\ell} \frac{1}{\ell!}$, it follows that, with probability at least $1 - \exp(-\alpha_\ell)$,
  $$\textsc{Count}_{\ell}(\tau, \sigma|_S) \not\in (1 \pm 3\beta_{\ell}/4) \binom{k}{\ell}\binom{|S|}{\ell} \frac{1}{\ell!},$$
  because $\beta_\ell - \beta_\ell/4 = 3\beta_\ell/4$.
\end{proof}

\subsubsection{The algorithm}\label{sec:permutation-alg}

Recall the definition of $S_{\tau}$ from \Cref{def:S-tau}: for each $N$, let $S_{\tau, N} = 1 + \max_{\tau' \subseteq \tau} (|1 - Q_{\tau', N}|)$. Let $S_\tau = \{S_{\tau, N}\}_N$.

In this section, we prove the following.

\begin{lemma}\label{lem:s-tau-qc}
    Let $\tau$ be a length-$k$ permutation pattern, $\ep > 0$, and $p$ be a non-atomic distribution over $\mathbb{R}$. Define $\D = \{p^{\otimes n}\}_n$. Then, $(\D, S_\tau)$-quality control is achievable in $O\left(\frac{k^8}{\ep^5}\log\left(\frac{k}{\ep}\right)\right)$ queries and time. 
\end{lemma}

We will later observe \Cref{thm:qc-ub} as a corollary of \Cref{lem:s-tau-qc}.

Given our analysis of jumbledness and exponentially robust quasirandomness, we are ready to present the algorithm for $(\D, S_\tau)$-quality control. These two ingredients together allow us to construct an algorithm that is rather simple: sample a subsequence of the indices, and verify jumbledness on this subsequence. Exponentially robust quasirandomness allowed us to argue that this subsequence can be very small and still possess information about whether the global count of the pattern (and subpatterns) all align with random permutations. Jumbledness then allows us to extract this information from the small subsequence.

\begin{figure}[H]
\centering
\fbox{
\parbox{\dimexpr0.95\linewidth-2\fboxsep-2\fboxrule\relax}{
\textbf{Algorithm Permutation-Pattern-Count:} \\
On inputs: A length-$n$ real-valued sequence $\sigma$ with distinct entries, a length-$k$ permutation pattern $\tau$, parameter $\ep \in (0, 1)$, $n \geq N$
\begin{enumerate}
    \item Sample $N = C\frac{k^8}{\ep^5}\log \left(\frac{k}{\ep}\right)$ indices; replace the sampled values by their relative ranks and call the resulting induced permutation $\pi$.
    \item Run \textsc{Jumbled-Count} on inputs: $\pi$, $\tau$, $\ep/8$, and $N$.
    \item Accept iff the call to \textsc{Jumbled-Count} accepts; reject otherwise.
\end{enumerate}
}
}
\caption{Algorithm for quality control of permutation pattern counts in a permutation. Exponentially robust quasirandomness is used to prove that sampling $N$ indices suffices to preserve relative pattern counts as compared to the global count. Permutation jumbledness is used to verify that, in the permutation on the sampled indices, the counts of patterns are as expected from a random permutation.}
\label{figure:permutation-alg}
\end{figure}

We now prove \Cref{lem:s-tau-qc}.

\begin{proof}[Proof of \Cref{lem:s-tau-qc}]

Although the quality-control problem is defined for real-valued sequences drawn from $p^{\otimes n}$, since $p$ is non-atomic, the queried values are distinct almost surely. Replacing each queried value by its rank yields a uniformly random permutation, so we treat the input as a permutation $\sigma$ throughout.
\\\\
\textbf{Completeness.} We first prove that a random permutation is accepted with probability $1 - o(1)$. First, if $\sigma$ is a random permutation over $[n]$ and $S \subseteq [n]$, then $\sigma|_S$ is a random permutation over $S$. Thus, by \Cref{lem:completeness}, when \textsc{Jumbled-Count} is called on $\sigma|_S$, \textsc{Accept} is output except with probability $O(1/(k^dN))=o(1)$ in the runtime of the algorithm.
\\\\
\textbf{Soundness.} Suppose that $S_{\tau,n}(\sigma) \notin 1\pm\ep$. Then there exists a subpattern $\tau'\subseteq\tau$ of length $r\leq k$ such that $Q_{\tau',n}(\sigma)\notin1\pm\ep$. Apply the preceding argument to $\tau'$: there is a minimum $\ell \in \{2,3,\dots,r\}$ such that $\textsc{Count}_{\ell}(\tau', \sigma) \not\in (1 \pm \beta_\ell) \binom{r}{\ell}\binom{n}{\ell} \frac{1}{\ell!}$. Then, by \Cref{lem:exp-qr-main}, on a random subset $S$ of size $N \geq s := \widetilde{O}(k^7/\ep^2)$, with probability at least $1-\exp(-\alpha_\ell)$, $\textsc{Count}_{\ell}(\tau', \sigma|_S) \not\in (1 \pm 3\beta_\ell/4) \binom{r}{\ell}\binom{|S|}{\ell} \frac{1}{\ell!}$. Since $\beta_\ell \geq \ep/4$, a deviation of $3\beta_\ell/4 \geq 3\ep/16 > \ep/8$ in the summed count $\textsc{Count}_\ell$ forces at least one length-$\ell$ subpattern of $\tau'$ to deviate by more than $\ep/8$. By \Cref{thm:jumbled-pattern}, \textsc{Jumbled-Count}, run with parameter $\ep/8$, therefore rejects. The sample size and jumbledness parameters chosen using $k$ dominate those required for $r\leq k$, and jumbledness certifies all permutation patterns of length at most $k$ simultaneously.
\end{proof}

We now prove \Cref{thm:qc-ub} from \Cref{lem:s-tau-qc}.

\begin{proof}[Proof of \Cref{thm:qc-ub} from \Cref{lem:s-tau-qc}]
    Suppose $A$ is a $(\D, S_{\tau})$-quality control algorithm. We argue that it is also a $(\D, Q_{\tau})$-quality control algorithm. For completeness, $A$ accepts $x \sim \D$ with probability $1 - o(1)$. For soundness, by definition of $Q_{\tau, n}$, $Q_{\tau, n}(f) \not \in 1 \pm \ep$ implies $S_{\tau, n}(f) \not \in 1 \pm \ep$. Thus, $A$ rejects on $f$ with $Q_{\tau, n}(f) \not \in 1 \pm \ep$. We conclude that $A$ is a $(\D, Q_{\tau})$-quality control algorithm.
\end{proof}

\subsection{Quality control lower bound}\label{sec:quality-control-permutation-lb}

We turn to proving the lower bound of \Cref{thm:any-pattern-intro}, restated below.

\begin{theorem}\label{thm:any-pattern-lb-in-section-1}
    For every non-atomic distribution $p$ over $\mathbb{R}$, letting $\D = \{p^{\otimes n}\}_n$, and any length-$k$ permutation pattern $\tau$, any algorithm for $(\D, Q_\tau)$-quality control requires $\Omega(k^{1.5 -\delta})$ queries, for any constant $\delta \in (0, 1)$.
\end{theorem}

\begin{theorem}\label{thm:any-pattern-lb-in-section-2}
    Let $p$ be a non-atomic distribution over $\mathbb{R}$, let $\D = \{p^{\otimes n}\}_n$, and let $\tau$ be the increasing pattern of length $k$. Any algorithm for $(\D, Q_\tau)$-quality control requires $\Omega(k^{2 -\delta})$ queries, for any constant $\delta \in (0, 1)$.
\end{theorem}

Before giving the construction, we define a quantity $b_\tau$ of a permutation pattern $\tau$ that will play a role in our query complexity lower bound.

\begin{definition}\label{def:b-P}
    Let $\tau:[k]\to[k]$ be a length-$k$ permutation pattern and let $t \geq 1$ be a constant. Consider any rank $r \in [k-t+1]$ such that $\{r, r+1, \dots, r+t-1\} \subseteq [k/4, 3k/4]$. For such an $r$, let
    $$ B_r = \left\{ \tau^{-1}(r), \tau^{-1}(r+1), \dots, \tau^{-1}(r+t-1) \right\}
    \subseteq [k]$$
    be the set of positions whose ranks under $\tau$ form the consecutive run
    $r, r+1, \dots, r+t-1$. Deleting the $t$ positions of $B_r$ from the line $[k]$ leaves at most $t+1$ maximal intervals of consecutive positions, which specifically are the (up to) $t-1$ gaps between consecutive positions in $B_r$, together with the two end intervals. Let $\nu(\tau, t, r)$ be the number of these intervals that are nonempty.
        
    Define
    $$ b_{\tau, t} = 
        \min_{\substack{r \in [k-t+1] :\\
        \{r,r+1,\dots,r+t-1\}\subseteq[k/4,\,3k/4]}}
        \nu(\tau, t, r), $$
    the least number of nonempty intervals into which these $t$ positions partition the remaining $k - t$ positions, over all choices of the consecutive planted ranks lying
    in the middle half of $[k]$.
\end{definition}

\begin{remark}
    For $k > t$ we always have $1 \le b_{\tau,t} \le t+1$. For the increasing (and, symmetrically,
    the decreasing) pattern, the planted blocks are consecutive in sequence order, so only the two end intervals are nonempty and $b_{\tau,t} = 2$. This gives the exponent
    $2 - \tfrac{4}{2t} = 2 - \tfrac{2}{t} \to 2$ of \Cref{thm:any-pattern-lb-in-section-2}. In the worst case, the
    middle-ranked blocks are spread throughout the sequence, every interval is nonempty, and
    $b_{\tau,t} = t+1$, giving the exponent $2 - \frac{t+3}{2t} = \frac{3}{2} - \frac{3}{2t} \to \frac{3}{2}$ of
    \Cref{thm:any-pattern-lb-in-section-1}.
\end{remark}

When the permutation pattern $\tau$ and constant $t$ are clear from context, we refer to $b_{\tau, t}$ as $b$. \Cref{thm:any-pattern-lb-in-section-1} and \Cref{thm:any-pattern-lb-in-section-2} both follow from the following result we prove.

\begin{lemma}\label{lem:any-pattern-lb-in-section-3}
    For all non-atomic distributions $p$ over $\mathbb{R}$, letting $\D = \{p^{\otimes n}\}_n$, and for any length-$k$ permutation pattern $\tau$, any algorithm for $(\D, Q_\tau)$-quality control requires $\Omega\left(k^{2 - \frac{b_{\tau, t} + 2}{2t}}\right)$ queries, for any constant $t \geq 1$.
\end{lemma}

We begin by arguing about the case where $p = \mathrm{Unif}(0, 1)$ and then describe why the same argument extends from $\mathrm{Unif}(0, 1)$ to any non-atomic distribution $p$.

For the case of $p = \mathrm{Unif}(0, 1)$, we prove the theorem via the following construction of a sequence that is difficult to distinguish from the uniform distribution with fewer than $\Omega\left(k^{2 - \frac{b_{\tau, t} + 2}{2t}}\right)$ queries.

\paragraph{Construction.}

Let $t$ be an arbitrarily large constant. Let $s = \Theta\left(\frac{n}{k^{2 - \frac{b+2}{2t}}}\right)$.

\begin{enumerate}
    \item Divide a length-$n$ sequence into $k$ blocks. For each block $i \in [k]$, look at the rank of $i$ in the permutation pattern $\tau$. Let the ``valid blocks'' be those blocks whose ranks are in $[k/4, 3k/4]$.
    \item Sample $i_1 \sim [n]$ uniformly at random among the first halves of valid blocks. See which block it falls into and call it $j_1$.
    \item Let $r$ be the rank of $j_1$ in the permutation pattern $\tau$. Place $\frac{r-1}{k-1}$ at position $i_1$, $i_1 + 1, \dots, i_1 + \frac{s}{t} - 1$. 
    \item Find the block $j_2$ whose rank is $r + 1$. Sample a uniform random index $i_2$ in the first half of this block. Place $\frac{r}{k-1}$ at position $i_2$, $i_2 + 1, \dots, i_2 + \frac{s}{t} - 1$.
    \item Repeat this $t - 2$ more times, with ranks $r + 2$, $r + 3$, $\dots, r + t - 1$, respectively.
    \item For each of the remaining indices, independently sample values uniformly at random from $(0, 1)$ and place these values at these indices.
\end{enumerate}

Let $\mathcal{D}_{\textsc{no}}^{s, t}$ be the distribution over sequences induced by this construction. We illustrate the construction of sequences from $\mathcal{D}_{\textsc{no}}^{s, t}$ for the increasing pattern \Cref{fig:lowerbound-qc-visual-increasing} and general patterns \Cref{fig:lowerbound-qc-visual}.

\begin{figure}
    \centering
    \includegraphics[scale=0.14]{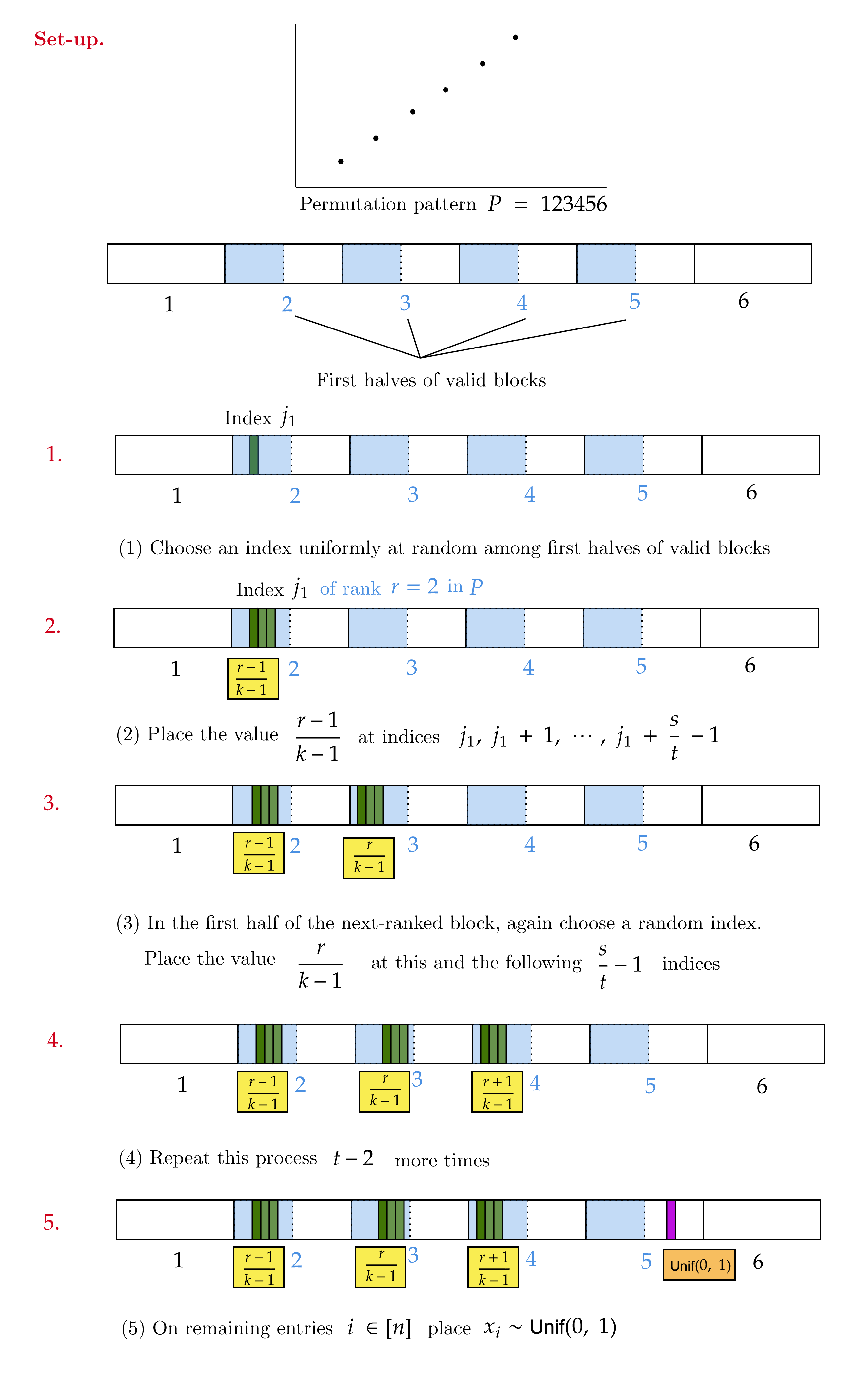}
    \caption{The process of generating a sequence according to the distribution $\mathcal{D}_{\textsc{no}}^{s, t}$ when the permutation pattern $\tau$ is the length-6 increasing permutation pattern. Each figure corresponds to a step of the construction described in \Cref{sec:quality-control-permutation-lb}.}
    \label{fig:lowerbound-qc-visual-increasing}
\end{figure}

\begin{figure}
    \centering
    \includegraphics[scale=0.14]{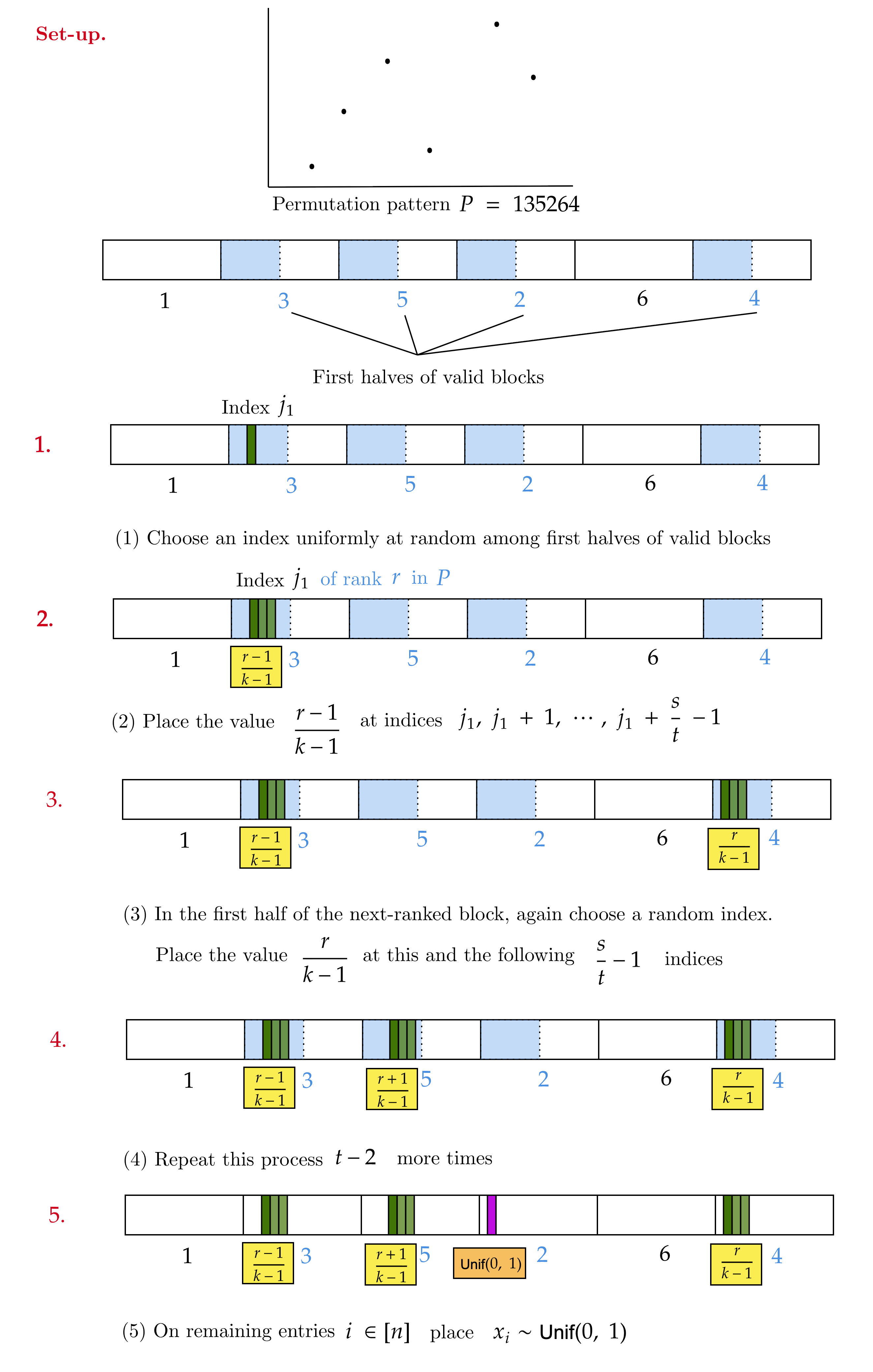}
    \caption{The process of generating a sequence according to the distribution $\mathcal{D}_{\textsc{no}}^{s, t}$. Each figure corresponds to a step of the construction described at the start of \Cref{sec:quality-control-permutation-lb}.}
    \label{fig:lowerbound-qc-visual}
\end{figure}

\paragraph{Analysis.}

In our analysis, we will be interested in the number of copies of $\tau$ in a sequence constructed according to the procedure above. As such, one of the properties we care about is captured by the following question: For $k - t$ values drawn uniformly at random from $(0, 1)$, what is the probability that exactly $r - 1$ values will fall \textit{below} the $t$ values planted, and the rest will fall \textit{above}? We will need such an event to hold to get a copy of $\tau$ in some positions, and will use this lemma later on in the analysis.

\begin{lemma}\label{lem:analysis-construction-1}
Let $k \in \mathbb{N}$ be sufficiently large and let $r \in [k/4, 3k/4]$. Let $X_1, X_2, \dots, X_{k - t}$ each be independently and identically distributed according to $\text{Unif}(0, 1)$. Let $p_{r-1}$ be the probability that exactly $r - 1$ of the $X_1, X_2, \dots, X_{k - t}$ have values less than $\frac{r - 1}{k - 1}$ on them, and exactly $k - t - r + 1$ of them have value greater than $\frac{r + t - 2}{k - 1}$ on them. Then $p_{r-1}$ is at least $\frac{c}{\sqrt{k}}$ for some constant $c$ that is independent of $k, t, r$.
\end{lemma}

\begin{proof}

    Let us begin by writing out $p_{r-1}$ explicitly:
    $$p_{r-1} = \mathbb{P}\left( \text{Binom}\left(k-t, \frac{r-1}{k-1}\right) = r - 1\right) = \binom{k - t}{r-1} \left( \frac{r - 1}{k-1} \right)^{r-1} \left(1- \frac{r + t - 2}{k-1} \right)^{k - t - r + 1}.$$
    Let $$A = \frac{r - 1}{k-1}, ~~~~ B = \left(1- \frac{r + t - 2}{k-1} \right).$$
    Let $$\delta = \frac{A}{A + B} = \frac{ \frac{r - 1}{k-1} }{\frac{r - 1}{k-1} + \left(1- \frac{r + t - 2}{k-1} \right)} = \frac{r - 1}{k-t}.$$
    Then, $$p_{r-1} = \left( A + B \right)^{k - t}  \cdot \binom{k - t}{r-1} \cdot \left( \frac{A}{A + B}\right)^{r - 1} \left( 1 - \frac{A}{A + B}\right)^{k - t - r + 1}.$$
    $$= \left( \frac{k - t}{k-1}\right)^{k - t}  \cdot \binom{k - t}{\delta \cdot (k - t)} \cdot \delta^{r - 1} ( 1 - \delta)^{k - t - r + 1}.$$

    Now, this expression is: $$=\left( \frac{k - t}{k-1}\right)^{k - t} \cdot \mathbb{P}\left(\text{Binom}\left( k-t, \delta \right) = \delta\cdot(k-t) \right).$$ First, $\left( \frac{k - t}{k-1}\right)^{k - t} \geq e^{-(t-1)}$,
    which is a constant.
    
    Second, for $\mathbb{P}\left(\text{Binom}\left( k-t, \delta \right) = \delta\cdot(k-t) \right)$, we are interested in the probability that a Binomial random variable achieves its expectation. By \Cref{prop:binomial}, since $\delta$ is a constant because $r \in [k/4, 3k/4]$, then for some constant $C$, 
    $$\mathbb{P}\left(\text{Binom}\left( k-t, \delta \right) = \delta\cdot(k-t) \right) \geq \frac{C}{\sqrt{k}}.$$
    Put together, we find that $p_{r-1}$ is at least $\frac{c}{\sqrt{k}}$ for some constant $c$ that is independent of $k, t, r$.
\end{proof}

Given \Cref{lem:analysis-construction-1}, we are now ready to analyze the \textit{expected} number of copies of the length-$k$ permutation pattern $\tau$ in a sequence constructed according to our procedure above.

\begin{lemma}\label{lem:expected-P}
    Consider any sequence generated according to the construction described above. The expected number of copies of the length-$k$ permutation pattern $\tau$ this sequence has is at least $3 \binom{n}{k}/k!$.
\end{lemma}

\begin{proof}
    Let $X$ be the random variable representing the number of copies of $\tau$ in a sequence from the distribution of the construction. For any set of indices $I$, let $X^{(I)}$ represent the number of copies of $\tau$ involving all indices in $I$.

    Consider all sets of indices $I$ that consist of one index in each of the $t$ planted runs (each of length $s/t$). Call this collection of sets of indices $\mathcal{I}$. Since all values within each planted run are equal, a copy of a permutation pattern can contain at most one index from each planted run. Thus, each copy counted by $X^{(I)}$ determines a unique set $I \in \mathcal{I}$. We thus have:
    \begin{equation}\label{eq:split-up-expectation}
       \mathbb{E}\left(X \right) \geq \sum_{I \in \mathcal{I}} \mathbb{E}\left(X^{(I)} \right). 
    \end{equation}

    We now compute $\mathbb{E}\left(X^{(I)} \right)$ for any $I \in \mathcal{I}$.
    $$\mathbb{E}\left(X^{(I)} \right) = \sum_{\substack{J \subseteq [n]\setminus I\\ |J|=k-t}} \mathbb{P}\left( J \cup I \text{ forms a copy of } \tau \right)$$

    Let $T_1, T_2, \dots, T_{t}$ be the $t$ blocks in which we have placed values in steps 1 through 5 of the construction. Consider any set $J$ of $k - t$ indices in $[n] \setminus \left( \bigcup_{i = 1}^{t} T_i\right)$. We now aim to understand the probability that a copy of $\tau$ exists on $J \cup I$. It suffices (as we aim for a lower bound) to consider copies of $\tau$ such that the indices of $I$ have ranks $r, r+1, \dots, r+t-1$ in the copy of $\tau$. We therefore need two properties to hold in order to obtain a copy of $\tau$:
    \begin{enumerate}
        \item Exactly $r - 1$ indices of $J$ have values less than $\frac{r - 1}{k - 1}$ on them, and exactly $k - t - r + 1$ of the indices have value greater than $\frac{r + t - 2}{k - 1}$ on them.
        \item The order of the indices of $J$ is consistent with the pattern $\tau$ after deleting the $t$ planted ranks.
    \end{enumerate}

    Since the elements of $I$ are already ordered consistently with a copy of the pattern $\tau$, put together these properties imply that a copy of $\tau$ exists on $J \cup I$.

    First, to address (1), by \Cref{lem:analysis-construction-1}, this probability is bounded below by $\frac{c}{\sqrt{k}}$ for some constant $c$. Second, conditioned on (1), every ordering of the values on the indices of $J$ is equally likely, and so the probability that the order of the indices of $J$ is consistent with the pattern $\tau$ is $\frac{1}{(k - t)!}$.

    Thus, we find that:
    $$\mathbb{P}\left( J \cup I \text{ forms a copy of } \tau \right) \geq \frac{c}{(k - t)! \sqrt{k}}.$$
    Let us now count the number of sets $J$ where $J \cup I $ can possibly form a copy of $\tau$. For this to occur, in each of the up to $b \leq t+1$ portions of the sequence between/to the sides of the blocks that the elements of $I$ fall in, we must have enough elements in $J$ for a permutation pattern where the indices of $I$ occupy the desired ranks to occur. If a portion of the sequence necessitates some $a$ elements, say, observe that its length will be at least $a n / k$ by construction. Additionally, the sum of these $a$'s will be $k - t$. For the $b \leq t+1$ nonempty portions of the sequence, let $a_1, a_2, \dots, a_b$ represent the number of elements we need to fall in each respective portion. Thus,
    $\sum_{j=1}^{b} a_j = k-t$.
    
    We then find:
    $$\mathbb{E}\left(X^{(I)} \right) \geq \frac{c}{(k - t)! \sqrt{k}} \cdot \prod_{j = 1}^{b}\binom{a_j n / k}{a_j}.$$
    For sufficiently large $n$ and $k$ and $t$ constant with respect to $n$ and $k$, for some constant $c'$ this is:
    $$\geq \frac{c'e^{k-t}}{k^{k - t + 1}} \cdot \prod_{j = 1}^{b}\left(\left(\frac{a_j n \cdot e}{a_j k}\right)^{a_j} \cdot \frac{1}{\sqrt{a_j}}\right) =  \frac{c'}{k^{k - t + 1}} \cdot \left(\frac{n \cdot e^2}{k} \right)^{k - t} \cdot \prod_{j = 1}^{b}\frac{1}{\sqrt{a_j}}.$$
    By the AM-GM inequality, since $\sum_{j=1}^{b} a_j = k - t$, this is:
    $$\geq \frac{c'}{k^{k - t + 1}} \cdot \left(\frac{n \cdot e^2}{k} \right)^{k - t} \cdot \left( \frac{b}{k-t}\right)^{b/2}.$$
    
    Thus, by \Cref{eq:split-up-expectation}, and since $|\mathcal{I}|=(s/t)^t$, we find:
    $$\mathbb{E}(X) \geq \left(\frac{s}{t}\right)^t \cdot \frac{c'}{k^{k - t + 1}} \cdot \left(\frac{n \cdot e^2}{k} \right)^{k - t} \cdot \left( \frac{b}{k-t}\right)^{b/2}.$$

    This is at least $3 \binom{n}{k}/k!$ for $s \geq \Omega\left(\frac{n}{k^{2 - \frac{b+2}{2t}}} \right)$. To see this, observe that, for sufficiently large $n$ and $k$, for some constant $d$ we have:
    $$3 \binom{n}{k}\cdot \frac{1}{k!} \leq d \cdot \left( \frac{n \cdot e^2}{k^2}\right)^k.$$

    We need:
    $$\left(\frac{s}{t}\right)^t \frac{c'}{k^{k - t + 1}} \cdot \left(\frac{n \cdot e^2}{k} \right)^{k - t} \cdot \left( \frac{b}{k-t}\right)^{b/2} \geq d \cdot \left( \frac{n \cdot e^2}{k^2}\right)^k.$$
    Equivalently, for some constant $d'$
    $$\left(\frac{s}{t}\right)^t \geq d' \cdot \frac{n^{t}}{b^{b/2} \cdot k^{2t - 1 - \frac{b}{2}}}.$$
    Since $t$ is a constant, this is satisfied for $s \geq \Omega\left(\frac{n}{k^{2 - \frac{b+2}{2t}}} \right)$.
    In the worst case, $b = t + 1$, and we get $$ s \geq \Omega\left(\frac{n}{k^{\frac{3}{2}-\frac{3}{2t}}}\right).$$ Thus, for every fixed $\delta>0$, taking $t$ to be a sufficiently large constant gives $$ s \geq \Omega\left(\frac{n}{k^{3/2-\delta}}\right).$$
\end{proof}

\begin{lemma}\label{lem:too-many-copies-whp}
    Consider any sequence generated according to the construction described above. With probability at least $9/10$, the number of copies of the length-$k$ permutation pattern $\tau$ this sequence has is at least $2 \binom{n}{k}/k!$.
\end{lemma}

\begin{proof}
    Let $X$ be the random variable representing the number of copies of $\tau$ in a sequence from the distribution of the construction, and let $\mathcal{I}$ and $X^{(I)}$ be as in the proof of \Cref{lem:expected-P}. Let $W = \sum_{I \in \mathcal{I}} X^{(I)}$ be the number of copies of $\tau$ that use exactly one index from each of the $t$ planted runs (occupying ranks $r, r+1, \dots, r+t-1$) together with $k - t$ background indices. Distinct $I \in \mathcal{I}$ use distinct run indices, so these copies are distinct, and thus $W \leq X$.

    Condition on the positions of the $t$ planted runs (equivalently, on the index $i_1$), and write $\mathbb{E}_{\mathrm{pos}}$ for the conditional expectation given these positions. Once the positions are fixed, the only remaining randomness is in the i.i.d. $\mathrm{Unif}(0,1)$ background values, and whether $I \cup J$ forms a copy of $\tau$ depends only on the values at the $k - t$ background indices $J$ (every index within a run plays the same rank in any copy of $\tau$). Thus, for fixed positions, $X^{(I)}$ is the same for every $I \in \mathcal{I}$, so
    $$W = |\mathcal{I}| \cdot V = \left(\frac{s}{t}\right)^t V,$$
    where $V$ is the number of background sets $J$ such that $J \cup I$ forms a copy of $\tau$ for a fixed $I \in \mathcal{I}$.

    The quantity $V$ counts the occurrences of a fixed length-$(k-t)$ pattern among the i.i.d. $\mathrm{Unif}(0,1)$ background values, where an occurrence is determined both by the order of values and the threshold condition that exactly $r-1$ of the background values fall below the planted band and the remaining fall above. By the same second-moment computation as in the proof of \Cref{prop:random-permutation-P-copies} (applied with error parameter $1/3$), for sufficiently large $n$, with probability at least $9/10$ over the background values,
    $$V \geq \frac{2}{3}\mathbb{E}_{\mathrm{pos}}[V], \qquad \text{and thus} \qquad W = \left(\frac{s}{t}\right)^t V \geq \frac{2}{3}\mathbb{E}_{\mathrm{pos}}[W].$$

    Finally, the lower bound of \Cref{lem:expected-P} holds pointwise over the run positions: for every fixed choice of positions, the intervals between and beside the planted runs have length at least $a_j n / k$, so $\mathbb{E}_{\mathrm{pos}}[W] = \sum_{I \in \mathcal{I}} \mathbb{E}_{\mathrm{pos}}[X^{(I)}] \geq 3 \binom{n}{k}/k!$. Therefore, with probability at least $9/10$,
    $$X \geq W \geq \frac{2}{3}\mathbb{E}_{\mathrm{pos}}[W] \geq 2 \binom{n}{k}/k!.$$
\end{proof}

Let $\mathcal{D}_{\textsc{yes}}$ be the distribution over length-$n$ sequences which draws the value at each index uniformly at random from $(0, 1)$. Let $\mathcal{D}_{\textsc{no}}^{s, t}$ be the distribution over sequences induced by the construction above. We now prove an upper bound on the probability that any deterministic algorithm making $q$ queries will be able to return different outcomes for inputs from $\mathcal{D}_{\textsc{yes}}$ versus $\mathcal{D}_{\textsc{no}}^{s, t}$.

\begin{lemma}\label{lem:probability-differ-deterministic}
    Let $A$ be a deterministic algorithm that makes $q$ queries to the values at indices of a length-$n$ sequence and outputs $0$ or $1$. Then,
    $$\left|\mathbb{P}_{x \sim \mathcal{D}_{\textsc{yes}}}\left[A(x) = 1 \right] - \mathbb{P}_{y \sim \mathcal{D}_{\textsc{no}}^{s, t}}\left[A(y) = 1 \right] \right| \leq \frac{8 qs}{n}.$$
\end{lemma}

\begin{proof}
    We first argue that for every set $Q \subseteq [n]$, $|Q| \leq q$, the probability of intersecting the set $S$ of the $s$ indices with planted values is at most $\frac{8 qs}{n}$. (That is, $S$ is the set of all indices whose value is set in Steps 1 through 5 of the construction.)

    To see this, fix any coordinate $j \in [n]$. We analyze the probability that $j \in S$, where the probability is taken over the random construction of $S$.

    Fix an index $j$, and let $B$ be the block containing $j$. Conditional on $B$ being one of the blocks containing planted values, $j \in S$ if the starting point of the planted run is in $\left\{j - \frac{s}{t} + 1, j - \frac{s}{t} + 2, \dots, j\right\}$. Since the starting point is chosen uniformly at random from the first half of the block, which has $n/(2k)$ elements, the probability that one of these $\frac{s}{t}$ points is sampled is at most:
    $$\frac{2k}{n} \cdot \frac{s}{t}.$$
    
    The initial block in Step 2 is chosen uniformly at random from blocks with ranks in $[k/4, 3k/4]$. The remaining blocks are specified given the initial block, and using a similar analysis as above, we find that the probability that $B$ is a block with any planted values is at most $\frac{4t}{k}$.
    
    Thus, for any fixed $j \in Q$, the probability that $j \in S$ is at most:
    $$\frac{2k}{n} \cdot \frac{4t}{k} \cdot \frac{s}{t} = \frac{8s}{n}.$$
    By a union bound, the probability that any $j \in Q$ is in $S$ is at most:
    $$\frac{8qs}{n}.$$
    
    We couple the distribution $\mathcal{D}_{\textsc{yes}} = (\text{Unif}(0, 1))^{\otimes n}$ and $\mathcal{D}_{\textsc{no}}^{s, t}$ as follows. First, draw a sequence $x$ according to $\mathcal{D}_{\textsc{yes}}$. Then, apply Steps 1 through 5 of the construction of $\mathcal{D}_{\textsc{no}}^{s, t}$, replacing some of the indices' values with planted values. Call the resulting sequence $r(x)$. We prove something stronger than the statement of \Cref{lem:probability-differ-deterministic}: For every sequence $x$, with high probability over the construction of $r(x)$, we have: 
    \begin{equation}\label{eq:indistinguishable}
        \mathbb{P}_{r(x)}[A(x) \neq A(r(x))] \leq \frac{8qs}{n}.
    \end{equation}

    Let $Q$ be the set of indices queried. As above, let $S$ be the set of planted values. Note that, when $Q \cap S = \emptyset$, $x$ and $r(x)$ look indistinguishable to the algorithm, since the query responses are the same on $x$ and $r(x)$ and the outputs are deterministic functions of prior queries and their values. We conclude that \Cref{eq:indistinguishable} holds because, as we proved above, the probability that $Q \cap S \neq \emptyset$ is at most $\frac{8qs}{n}$.
\end{proof}

\begin{proof}[Proof of \Cref{lem:any-pattern-lb-in-section-3}]
    Let $A$ be any algorithm that makes $q$ queries and accepts $x \sim (\text{Unif}(0, 1))^{\otimes n}$ with probability $1 - o(1)$ and rejects any $y$ with at least $2 \binom{n}{k}/k!$ copies of length-$k$ permutation pattern $\tau$ with probability at least $9/10$. By \Cref{prop:random-permutation-P-copies}, with probability $1 - o(1)$,  $x \sim (\text{Unif}(0, 1))^{\otimes n}$ has $(1 \pm \frac{1}{2})\binom{n}{k}/k!$ copies of $\tau$.

    Thus it must be the case that
    $$\mathbb{P}_{x \sim \mathcal{D}_{\textsc{yes}}, R}\left[A(x) = 1 \right] - \mathbb{P}_{y \sim \mathcal{D}_{\textsc{no}}^{s, t}, R}\left[A(y) = 1 \right] \geq 4/5$$
    where the probability is taken over both the randomness of the constructions of $x$ and $y$ and the random coin tosses made by the algorithm $A$.

    Thus, there must exist a fixed choice $R_0$ such that 
    $$\mathbb{E}_{x \sim \mathcal{D}_{\textsc{yes}}}\left[A'(x) = 1 \right] - \mathbb{P}_{y \sim \mathcal{D}_{\textsc{no}}^{s, t}}\left[A'(y) = 1 \right] \geq 4/5,$$
    where $A'$ is the deterministic counterpart of $A$ defined by fixing the randomness of $A$ to $R_0$. 

    Now, applying \Cref{lem:probability-differ-deterministic}, we find that:
    $$\frac{8qs}{n} \geq \frac{4}{5},$$
    i.e., that $q \geq n/(10 s)$.

    By definition of $s = \Theta\left(\frac{n}{k^{2 - \frac{b_{\tau, t} + 2}{2t}}} \right)$, we find that for any length-$k$ permutation pattern $\tau$, any algorithm for $(\{p^{\otimes n}\}_n, Q_\tau)$-quality control requires $\Omega\left(k^{2 - \frac{b_{\tau, t} + 2}{2t}}\right)$ queries, for any constant $t \geq 1$. This yields \Cref{lem:any-pattern-lb-in-section-3} for $p = Unif(0, 1)$.

    \paragraph{General non-atomic distributions:} The same lower bound applies to every non-atomic distribution $p$ over $\mathbb{R}$, by connecting $p$ to $\mathrm{Unif}(0,1)$ via its quantile function. Let $F$ be the cumulative distribution function of $p$, and let $F^{-1}(u) := \inf\{x : F(x) \geq u\}$ be its quantile function. When $U \sim \mathrm{Unif}(0,1)$, we have $F^{-1}(U) \sim p$. Furthermore, if $U_1, U_2, \dots, U_n$ are independent $\mathrm{Unif}(0,1)$ random variables, then $F^{-1}(U_1), F^{-1}(U_2), \dots, F^{-1}(U_n)$ are distinct with probability $1$, because $p$ is non-atomic. Since $F^{-1}$ is nondecreasing, it preserves the relative ordering of the elements. Thus, replacing $u$ entrywise in the construction with $F^{-1}(u)$ preserves all permutation pattern counts and extends the construction to general non-atomic $p$.
\end{proof}

\section{Patterns}\label{sec:patterns}

We now move to studying the count of \textit{subsequences} in sequences, which we refer to as \textit{patterns}. As opposed to permutation patterns, for patterns, we now care about \textit{values} instead of relative ordering.
The distribution that we perform quality control against is the product distribution over $n$ entries of some distribution $p$ over $[m]$, for $m$ constant. The most natural case is arguably $m = 2$ and $p = \text{Unif}(\{0, 1\})$, where the sequence generated is a Boolean sequence, and quality control asks whether the sequence looks like a random Boolean sequence with respect to the counts of small patterns. (And in this case, permutation patterns do not make sense, because everything is 0- or 1-valued, and so permutation patterns without ties are not well-defined.) 

Let the count of a pattern $P$ in length-$n$ sequence $f$ be $C_{P, n}(f)$. Let $p$ be a distribution over $[m]$ which assigns weight $p_i$ to each $i \in [m]$. Define the pattern count quality $R_{P, n}$ to be $$R_{P, n}(f) = \frac{C_{P, n}(f)}{\binom{n}{k} \cdot \prod_{j \in [k]} p_{P(j)}}.$$ For a distribution $p$, let $p^{\otimes n}$ be the product distribution over $n$ i.i.d. copies of $p$. 

Let $R_P = \{R_{P, n}\}_n$ and $\D = \{p^{\otimes n}\}_n$. We are interested in $(\D, R_P)$-quality control. 

Our main algorithmic result, stated as \Cref{thm:subsequences} at the start of \Cref{sec:pattern-ub} and proved at its end, shows that $(\D, R_P)$ is achievable with $\text{poly}(k/p_{\min})$ queries and time, where $p_{\min} = \min_{i \in [m]} p_i$. We also give a $\Omega(k^2/p_{\min})$ lower bound (\Cref{thm:qc-lb-sequences-general1}) for some patterns, proven in \Cref{sec:qc-lb-all}.

\subsection{Quality control algorithm}\label{sec:pattern-ub}
We now turn to quality control algorithms with the objective of accepting random sequences over $[m]^n$, for each $n \in \mathbb{N}$, and rejecting sequences whose count of a pattern $P$ of length $k$ does not match the expected count for product distributions with entries from a distribution $p$. We think of $k$ and $m$ as constant. (For $m$, the key example is Boolean sequences, for which $m = 2$. For $p$, the key example is the uniform distribution.)

We define the minimum weight of a distribution over a finite domain as follows.

\begin{definition}[Minimum weight]
     We say a probability distribution $p$ over $[m]$ has minimum weight $p_{\min}$ for $p_{\min} = \min_{i \in [m]} p_i$, where $p_i$ is the probability of $i$ under $p$.
\end{definition}

Recall that $R_{P, n}(\sigma)$ is the pattern count quality of $P$ in length-$n$ sequence $\sigma$, i.e., the number of times the pattern $P$ appears as a subsequence in $\sigma$, normalized by its expectation under the reference distribution. In this section, we prove the following.

\begin{theorem}\label{thm:subsequences}
    Consider a sequence of length $n$ with entries in $[m]$. Let $P$ be a length-$k$ sequence. Let $p$ be a distribution over $[m]$ with minimum weight $p_{\min}$, and define $\D = \{p^{\otimes n}\}_n$.
    Then $(\D, R_P)$-quality control is possible in $\widetilde{O}\left(\frac{k^5}{\ep^2 p_{\min}^2}\right)$ queries and $\widetilde{O}\left(\frac{k^7}{\ep^2 p_{\min}^2}\right)$ time.
\end{theorem}

For the key example of the uniform distribution over $[m]$ in each index, we obtain the following corollary.

\begin{corollary}
    Consider a sequence of length $n$ with entries in $[m]$. Let $P$ be a length-$k$ sequence. Let $p$ be the uniform distribution over $[m]$ and $\D = \{p^{\otimes n}\}_n$.
    Then $(\D, R_P)$-quality control is possible in $\widetilde{O}\left(\frac{m^2 k^5}{\ep^2}\right)$ queries and $\widetilde{O}\left(\frac{m^2 k^7}{\ep^2}\right)$ time.
\end{corollary}

We stress the lack of dependence on $n$ as well as the polynomial dependence on $m$ and $k$.

We start with an overview now before moving on to the elements of the proof. As in \Cref{sec:quality-control-permutation-ub}, we define a stronger target quality control problem, denoted $(\D, T_P)$-quality control for $T_P = \{T_{P, N}\}_N$. Our goal is to prove the analogues of \Cref{thm:jumbled-pattern} and \Cref{cor:qc-samples} for $(\D, T_P)$-quality control. The natural choice would be to let $T_P$ assure that pattern counts are correct for \textit{all} subpatterns $P'$ of $P$. While this approach is feasible, we do not pursue it here because it gives us a higher query and time complexity than we are finally able to achieve. Instead, we adopt a definition of $T_P$ carefully so as to make the algorithmic task of computing $T_P$ doable with dynamic programming, which is proven in \Cref{sec:pattern-dp}. At the same time, this choice is made so that we can still use the methodology of exponentially robust quasirandomness to bound the sample complexity. The exact definition is the following, which aggregates the counts of all subpatterns of length $\ell$ of $P$ into a single quantity and bounds the maximum deviation from the expectation over all $\ell \in [k]$.

\begin{definition}\label{def:T-P-qc}
    For each $N$ and each length-$N$ sequence $\sigma$, let
    $$T_{P, N}(\sigma) = 1 + \max_{\ell \in [k]}\left(\left|1 - \frac{\sum_{\substack{I \subseteq [k] \\ |I| = \ell}} C_{P|_I, N}(\sigma)}{\binom{N}{\ell} \cdot \sum_{\substack{I \subseteq [k] \\ |I| = \ell}} \prod_{i \in I} p_{P_i}}\right|\right),$$
    where subpatterns are indexed by the subsets $I \subseteq [k]$ and are therefore counted with multiplicity. Let $T_P = \{T_{P, N}\}_N$.
\end{definition}

$T_{P, N} \approx 1$ certifies that, for all $\ell \in [k]$, the sum of subpatterns of length $\ell$ of $P$ occurs with the right frequency in a given instance. This is a stronger guarantee than $R_{P, N} \approx 1$ since $T_{P, N} \approx 1$ implies $R_{P, N} \approx 1$, simply by looking at $\ell = k$.

To give a $(\D, T_P)$-quality control algorithm, it turns out that a dynamic programming algorithm allows us to compute $\sum_{\substack{P' \subseteq P \\ |P'| = \ell}}C_{P', N}$ simultaneously for all $\ell \in [k]$ on an instance $\sigma$, and thus to compute $T_{P, N}$ on the instance (see \Cref{sec:pattern-dp}). This dynamic programming algorithm uses the fact that a length-$L$ subpattern $P' \subseteq P$ that ends at some index $\sigma[i] = P[j]$ is (distinctly) formed by taking a length-$(L-1)$ subpattern ending at some $\sigma[i'] = P[j'], i' < i, j' < j$ and appending $\sigma[i] = P[j]$.  The dynamic programming approach yields a $(\D, T_P)$-quality control algorithm with query complexity $O(N)$ and time complexity $O(Nk^2)$.

We will then again use a notion of exponentially robust quasirandomness to show that a $\text{poly}(k/p_{\min})$-size sample preserves the aggregate counts of length-$\ell$ subpatterns, for all $\ell \in [k]$. This allows us to say that  $\text{poly}(k/p_{\min})$ samples suffice for  $(\D, T_P)$-quality control. Combining this with the time-efficient algorithm on the sample yields \Cref{thm:subsequences}. The proof of the sample size via exponentially robust quasirandomness will change for the case of patterns due to the new expected number of patterns of a certain length in a random instance. However, the argument itself will be analogous, as for permutation patterns we had already reasoned about bounded differences using aggregate length-$\ell$ pattern counts, for each $\ell \in [k]$.

As in \Cref{sec:quality-control-permutation-ub}, we begin with the time-efficient subroutine we will perform on the $\text{poly}(k/p_{\min})$-size sample, given in \Cref{thm:subsequences-dp}. We then use exponentially robust quasirandomness in \Cref{sec:pattern-exp-rob-qr-sequences} to argue that a $\text{poly}(k/p_{\min})$-size sample suffices to preserve pattern counts. We combine these ingredients in \Cref{sec:pattern-algo}.

\subsubsection{Analysis: dynamic programming}\label{sec:pattern-dp}

We analyze quality control for the counts of all patterns of a sequence $P$ as subsequences in the longer sequence $\sigma$.

\begin{theorem}\label{thm:subsequences-dp}
    Let $P$ be a length-$k$ sequence and let $p$ be any distribution over $[m]$ with minimum weight $p_{\min}$. Let $T_P = \{T_{P, N}\}_N$ and $\D = \{p^{\otimes N}\}_N$. Then, for $N = \Omega\left(\frac{k^3}{p_{\min} \ep^2}\log\left(2+\frac{k}{\ep p_{\min}}\right) \right)$, $(\D, T_P)$-quality control is possible in $O(N)$ queries and $O(N k^2)$ time.
\end{theorem}

We consider the algorithm given in \Cref{figure:dynamic-programming}, which is a dynamic programming algorithm. We will construct a three-dimensional table $M$ such that $M[i,j,L]$ will contain the number of pairs of length-$L$ indexed subsequences of $P$ and $\sigma$ that agree, such that the final selected positions are $j$ in $P$ and $i$ in $\sigma$.

\begin{figure}[h]
\centering
\fbox{
\parbox{\dimexpr0.95\linewidth-2\fboxsep-2\fboxrule\relax}{
\textbf{Algorithm Dynamic-Programming:} \\
On inputs: A length-$N$ sequence $\sigma$ in $[m]^N$, a length-$k$ pattern $P$, a probability distribution $p$ over $[m]$, and parameter $\ep \in (0, 1)$.
\\\\
Dynamic programming to compute aggregate subpattern counts:
\begin{enumerate}
    \item Initialize $(N+1) \times (k+1) \times (k+1)$ table $M$, indexed starting at $0$, to be all-zero. (The entry $M[i, j, L]$ will consist of the number of $P' \subseteq P$ of length $L$ in $\sigma[1, \dots, i]$ such that $\sigma[i] = P[j]$ and $P'$ ends at index $i$ of $\sigma$ and index $j$ of $P$.)
    \item Initialize $(N+1) \times (k+1) \times (k+1)$ table $S$, indexed starting at $0$, to be all-zero. (This will be a table of prefix sums.)
    \item For all $i$ from 1 to $N$, all $L$ from $1$ to $k$, and all $j$ from $1$ to $k$:
    \begin{enumerate}
        \item If $\sigma[i] = P[j]$:
        \begin{enumerate}
            \item If $L = 1$, set $M[i, j, L] = 1$.
            \item Else, set $M[i, j, L] = S[i-1, j-1, L-1]$.
        \end{enumerate}
        \item Else, set $M[i, j, L]  = 0$.
        \item Set $S[i, j, L] = M[i, j, L] + S[i-1, j, L] + S[i, j-1, L] - S[i - 1, j-1, L]$.
    \end{enumerate}
\end{enumerate}
\noindent Now verify that all the counts match that of sequences from $p^{\otimes N}$:
    \begin{enumerate}[resume]
    \item For all $L$ from $1$ to $k$: 
    \begin{enumerate}
        \item Let $\textsc{Total-Count-L}$ be $S[N, k, L]$.
        \item Let $p_{P_a}$ be the weight that $p$ assigns the symbol at the $a$th index of $P$.
        \item If $\textsc{Total-Count-L} \not \in (1 \pm \eps) \binom{N}{L} \cdot \sum_{I \subseteq [k], |I| = L} \prod_{a \in I} p_{P_a}$, then output \textsc{Reject}.
    \end{enumerate}
    \item Output \textsc{Accept}.
\end{enumerate}
}
}
\caption{This algorithm uses dynamic programming to count the number of occurrences of patterns of a sequence $P$ as subsequences in a longer sequence $\sigma$. The algorithm then accepts sequences whose count of subpatterns of $P$ match that of sequences from $p^{\otimes N}$, and rejects sequences otherwise.}
\label{figure:dynamic-programming}
\end{figure}

We prove the dynamic programming algorithm correctly computes aggregate counts of subpatterns of length $L$ of $P$ in $\sigma$, for every $L \in [k]$.

\begin{lemma}\label{lem:dp-correct}
    The algorithm in \Cref{figure:dynamic-programming}, Steps 1 through 3, results in a table $S[i, j, L]$ such that $S[N, k, L]$ equals the aggregate count of length-$L$ subpatterns of length-$k$ pattern $P$ in sequence $\sigma$.
\end{lemma}

\begin{proof}
    \noindent\textbf{Intuition.} At each step, \begin{equation}\label{eq:S-wrt-M}
    S[i, j, L] = \sum_{i' \leq i} \sum_{j' \leq j} M[i', j', L]
    \end{equation}
    where $M[i', j', L]$ holds the count (with multiplicity) of $P' \subseteq P$ of length $L$ in the first $i'$ entries of $\sigma$ such that $\sigma[i'] = P[j']$ and $P'$ ends at index $i'$ of $\sigma$ and pattern index $j'$. First, \Cref{eq:S-wrt-M} holds because of Step 3(c), where we update $S[i, j, L]$ according to inclusion-exclusion of smaller prefix sums and the remaining entry $M[i, j, L]$:
    $$S[i, j, L] = S[i-1, j, L] + S[i, j-1, L] + M[i, j, L] - S[i-1, j-1, L].$$
    
    To see why $M[i', j', L]$ holds the intended quantity, observe that in Step 3(a), if $\sigma[i] = P[j]$ and $L = 1$, then $M[i, j, L] = 1$ because there is one length-1 pattern $P' = P[j]$ at this index of $\sigma$. If $L > 1$, then $M[i, j, L]$ equals the number of patterns of length $L - 1$ in $\sigma[1, \dots, i-1]$ ending at some pattern index $j' < j$. Each of these can be extended to a length $L$ pattern with last entry $\sigma[i] = P[j]$. 
    Thus, we compute $M[i, j, L]$ properly as:
    $$M[i, j, L] = \mathbf{1}(\sigma[i] = P[j]) \cdot \left( \mathbf{1}(L > 1) \cdot \left( \sum_{i' < i} \sum_{j' < j} M[i', j', L-1]\right) + \mathbf{1}(L = 1)\right).$$
    
    Thus, we see that all quantities are updated correspondingly and, at the end, $S[N, k, L]$ stores the total count of length-$L$ subpatterns of $P$ in $\sigma$.

    We now give the formal proof of correctness.
    \\\\
    \noindent \textbf{Proof of the dynamic program's correctness.} We give an inductive proof. Let $C(i,j,L)$ be the number of pairs $(A,I)$ such that $A \subseteq [i]$, $I \subseteq [j]$, $|A|=|I|=L$, $\max A=i$, $\max I=j$, and $\sigma|_A=P|_I$. $C(i, j, L)$ is the quantity that $M[i, j, L]$ is intended to hold, and below we prove that $M$ indeed does hold $C$.

    For the base case, a length-1 occurrence ending at $\sigma[i]$ and $P[j]$ exists iff $\sigma[i] = P[j]$. $M[i, j, 1] = 1$ iff $\sigma[i] = P[j]$, and is 0 otherwise. Thus, $M[\cdot, \cdot, 1] = C[\cdot, \cdot, 1]$. We set $S$ so that $S[N, k, 1] = \sum_{i, j} M[i, j, 1]$, and thus $S[N, k, 1]$ properly counts length-1 subpattern occurrences.

    For the inductive hypothesis, assume $M[\cdot, \cdot, L-1] = C(\cdot, \cdot, L-1)$. For the inductive step, fix a pair $(i, j)$ such that $\sigma[i] = P[j]$. Every length-$L$ subpattern $P' \subseteq P$ ending with $\sigma[i] = P[j]$ is formed by taking a length-$(L-1)$ subpattern ending at some $\sigma[i'] = P[j']$ for $i' < i, j' < j$ and appending $\sigma[i] = P[j]$. Appending is valid since $\sigma[i] = P[j]$. Each extension corresponds to a distinct occurrence of a subpattern.
    
    Via this and the inductive hypothesis,
    $$C(i, j, L) = \sum_{i' \leq i-1} \sum_{j' \leq j - 1} C(i', j', L-1) = S[i-1, j-1, L-1].$$
    Thus, $M[i, j, L] = S[i-1, j-1, L-1]$ in Step 3(a) iff $\sigma[i] = P[j]$, and 0 otherwise, sets $M[i, j, L] = C(i, j, L)$. $S$ is set so that $S[N, k, L] = \sum_{i, j} M[i, j, L]$ via \Cref{eq:S-wrt-M} and thus $S[N, k, L]$ properly counts length-$L$ subpattern occurrences.

    We conclude that $M = C$ on all entries, and $$S[N, k, L] = \sum_{i \leq N} \sum_{j \leq k} C(i, j, L)$$
    which correctly counts the number of pairs $(A,I)$, where $A \subseteq [N]$, $I \subseteq [k]$, $|A|=|I|=L$, and $\sigma|_A=P|_I$ for all $L \in [k]$.
\end{proof}

We now prove \Cref{thm:subsequences-dp}.

\begin{proof}[Proof of \Cref{thm:subsequences-dp}]
    First, by \Cref{lem:dp-correct}, the dynamic programming algorithm correctly computes the aggregate counts of all subpatterns of length $L$ of the length-$k$ pattern $P$ on $\sigma$, for all $L \in [k]$.
    
    Given this, we prove completeness and soundness of the quality control algorithm. \\\\
    \textbf{Completeness.} Suppose that $\sigma \sim p^{\otimes N}$. By \Cref{prop:min-weight-sequence-aggregate-copies}, with probability at least $1-O\left(\frac{k^3}{Np_{\min}\eps^2}\right)$, simultaneously for every $a\in[k]$,
    $$\textsc{Count}_a(P,\sigma) \in (1 \pm \eps) \binom{N}{a} \cdot \sum_{\substack{I \subseteq [k] \\ |I| = a}} \prod_{i \in I} p_{P_i}.$$
    Since $N=\Omega\left(\frac{k^3}{p_{\min}\eps^2}\log\left(2+\frac{k}{\eps p_{\min}}\right)\right)$, the failure probability $O\left(\frac{k^3}{Np_{\min}\eps^2}\right)$ is $O\left(1/\log\left(2+\frac{k}{\eps p_{\min}}\right)\right)$ and is $o(1)$ as $k/(\eps p_{\min})\to\infty$. By \Cref{lem:dp-correct}, $S[N,k,a]$ is exactly $\textsc{Count}_a(P,\sigma)$. Thus, the algorithm accepts with probability $1-o(1)$, and completeness holds.
    \\\\
    \textbf{Soundness.} Suppose that $T_{P,N}(\sigma) \notin 1 \pm \eps$. By the definition of $T_{P,N}$, for some $a \in [k]$, the number of subpatterns of $P$ of length $a$ that appear as subsequences in $\sigma$ is not $\in (1 \pm \eps) \binom{N}{a} \cdot \sum_{I \subseteq [k], |I| = a} \prod_{i \in I} p_{P_i}$. By \Cref{lem:dp-correct}, $S[N,k,a]$ is exactly this aggregate count, so the algorithm \textsc{Dynamic-Programming} will detect this and reject. Thus, soundness holds.
    \\\\
    \textbf{Query complexity and runtime.} The algorithm queries all indices of the sequence $\sigma$, so the query complexity is $O(N)$. The algorithm uses $O(N k^2)$ arithmetic operations since there are three nested loops: two over $[k]$ and one over $[N]$. Thus, in the unit-cost arithmetic model, the runtime is $O(Nk^2)$.
\end{proof}

\subsubsection{Analysis: exponentially robust quasirandomness}\label{sec:pattern-exp-rob-qr-sequences}

We prove the following.

\begin{lemma}\label{lem:exp-qr-main-subsequences}
    Suppose $\sigma \in [m]^n$ is a length-$n$ sequence, $P$ is a length-$k$ sequence, $p$ is a distribution over $[m]$ with minimum weight $p_{\min}$, and $\ep \in (0, 1)$. Sample a set $S \subseteq [n]$ of $s = \widetilde{O}\left(\frac{k^5}{\ep^2 p_{\min}^2}\right)$ indices. Suppose there exists some $\ell \in [k]$ such that $\textsc{Count}_{\ell}(P, \sigma) \not\in (1 \pm \beta_\ell) \binom{n}{\ell} \sum_{\substack{I \subseteq [k] \\ |I| = \ell}} \prod_{i \in I} p_{P_i} $ (and consider the first such $\ell$). Then, with probability at least $1 - \exp(-\alpha_{\ell})$, $ \textsc{Count}_{\ell}(P, \sigma|_S) \not\in (1 \pm 3\beta_{\ell}/4) \binom{|S|}{\ell} \sum_{\substack{I \subseteq [k] \\ |I| = \ell}} \prod_{i \in I} p_{P_i}.$
\end{lemma}

We prove the following corollary from the lemma.

\begin{corollary}\label{cor:qc-samples-sequences}
    Let $P$ be a length-$k$ sequence over $[m]$, $\ep > 0$, and $p$ be a distribution over $[m]$ with minimum weight $p_{\min}$. Define $\D = \{p^{\otimes n}\}_n$. Then, $(\D, R_P)$-quality control is achievable in $\widetilde{O}\left(\frac{k^5}{\ep^2 p_{\min}^2}\right)$ queries.
\end{corollary}

We begin by proving \Cref{cor:qc-samples-sequences} from \Cref{lem:exp-qr-main-subsequences}.

\paragraph{The algorithm.} The quality control algorithm performs the following procedure. First, sample a subsequence on $s := \widetilde{O}\left(\frac{k^5}{\ep^2 p_{\min}^2}\right)$ indices. Then, for each $\ell \in [k]$, evaluate whether $\textsc{Count}_{\ell}(P, \sigma|_S) \not\in (1 \pm 3\beta_\ell/4) \binom{|S|}{\ell} \sum_{\substack{I \subseteq [k] \\ |I| = \ell}} \prod_{i \in I} p_{P_i}$ in $O(s^k)$ time.

\begin{proof}[Proof of \Cref{cor:qc-samples-sequences} from \Cref{lem:exp-qr-main-subsequences}]
    We argue completeness and soundness.

    \paragraph{Completeness.} We prove that a sequence $\sigma \sim p^{\otimes n}$ is accepted with probability $1 - o(1)$. Since $\sigma \sim p^{\otimes n}$ and $S \subseteq [n]$ imply $\sigma|_S \sim p^{\otimes |S|}$, by \Cref{prop:min-weight-sequence-aggregate-copies}, the probability that one of the tests fails is at most $O\left(\frac{k^3}{s p_{\min}\ep^2}\right)=O\left(\frac{p_{\min}}{k^2\log(2+k/(\ep p_{\min}))}\right)=o(1)$ as $k/(\ep p_{\min})\to\infty$. Thus, for all $\ell \in [k]$, $\textsc{Count}_{\ell}(P, \sigma|_S) \in (1 \pm 3\beta_\ell/4) \binom{|S|}{\ell} \sum_{\substack{I \subseteq [k] \\ |I| = \ell}} \prod_{i \in I} p_{P_i}$, and $\sigma$ is accepted.

    \paragraph{Soundness.} Suppose that $\sigma$ is a length-$n$ sequence such that the count of length-$k$ pattern $P$ in $\sigma$ is not in $(1 \pm \ep) \binom{n}{k} \prod_{i \in [k]} p_{P_i}$. Then, there is a minimum $\ell \in [k]$ such that $\textsc{Count}_{\ell}(P, \sigma) \not\in (1 \pm \beta_\ell) \binom{n}{\ell} \sum_{\substack{I \subseteq [k] \\ |I| = \ell}} \prod_{i \in I} p_{P_i}$. Then, by \Cref{lem:exp-qr-main-subsequences}, on a random subset $S$ of size $N \geq s := \widetilde{O}\left(\frac{k^5}{\ep^2 p_{\min}^2}\right)$, with probability at least $1-\exp(-\alpha_\ell)$, $\textsc{Count}_{\ell}(P, \sigma|_S) \not\in (1 \pm 3\beta_\ell/4) \binom{|S|}{\ell} \sum_{\substack{I \subseteq [k] \\ |I| = \ell}} \prod_{i \in I} p_{P_i}$. In this case, the quality control algorithm rejects.
\end{proof}

As in \Cref{sec:exp-rob-qr} for permutation patterns, we define a notion of \textit{exponentially robust quasirandomness} for the counts of subpatterns of a pattern $P$ in a longer sequence $\sigma$ of length $n$. The analysis will differ because the ranges of quantities (such as the expectation and bounded difference inequality values) will now differ as we handle pattern counts.

We begin by introducing some notation. 

\begin{definition}
    For a subsequence $\sigma$ over $[m]^n$ and a length-$k$ pattern $P$, and $\ell \in [k]$, define $\textsc{Count}_{\ell}(P, \sigma)$ be the count in $\sigma$ of all patterns of length $\ell$ that are subpatterns of $P$:

    $$\textsc{Count}_{\ell}(P, \sigma) = \sum_{\substack{I \subseteq [k]\\|I| = \ell}} C_{P|_I,n}(\sigma),$$
    where subpatterns are indexed by their position sets $I$ and are therefore counted with multiplicity.
\end{definition}

For a subset of indices $S \subseteq [n]$ and a length-$n$ sequence $\sigma$, let $\sigma|_S$ be the subsequence of $\sigma$ on the indices $S$. For a distribution $p$ over $[m]$, let $p_i$ be the weight that $p$ assigns $i \in [m]$.

Next, define exponentially robust quasirandomness as follows.

\begin{definition}[Exponentially robust quasirandomness of pattern count]\label{def:exp-rob-qr-sequence}
    For a length-$k$ pattern $P$,  any $\ell \in [k]$, and a distribution $p$ over $[m]$, a sequence $\sigma$ over $[m]^n$ is $(\alpha, \beta, s_0, \ell, p)$-exponentially robustly quasirandom with respect to $P$ if, for all $s \geq s_0$,
    $$\mathbb{P}_{S : |S| = s}\left(\textsc{Count}_{\ell}(P, \sigma|_S) \not \in (1 \pm \beta) \binom{s}{\ell}\sum_{\substack{I \subseteq [k] \\ |I| = \ell}} \prod_{i \in I} p_{P_i} \right) \leq \exp(- \alpha).$$
\end{definition}

That is, a sequence $\sigma$ is exponentially robustly quasirandom with respect to a sequence $P$ if, with all but exponentially small probability, the number of occurrences of subpatterns of $P$ in size-$s$ subsequences of $\sigma$ is approximately what is expected when $\sigma$ is a sequence drawn from $p^{\otimes n}$.

As in \Cref{sec:exp-rob-qr}, we prove a ``concentration lemma'' below. We first define the relevant parameters.

\begin{definition}\label{def:parameters-exprobustqr-sequences}
For $1 \leq \ell \leq k$, define
$$
    \beta_\ell := \frac{\ep}{4}\cdot 2^{(k-\ell)/k}.
$$
(Observe that $\beta_\ell=\Theta(\ep)$ for every $\ell\in[k]$, and in particular
$\beta_k \leq \ep$.)

For $1 \leq \ell \leq k$, define
$$
     \alpha_\ell := (k-\ell)\log 3 + k\log(1/p_{\min}) + \log(3k/\ep).
$$

Let
$$
    s = \widetilde{O}\left(\frac{k^5}{\ep^2 p_{\min}^2}\right).
$$
\end{definition}

\begin{lemma}[Concentration lemma]\label{lem:concentration-lemma-sequences}
    Suppose a sequence $\sigma \in [m]^n$ is $(\alpha_{\ell}, 2\beta_{\ell}, s, \ell, p)$-exponentially robustly quasirandom with respect to $P$. Then, for $\alpha_{\ell+1}, \beta_{\ell+1}, s$ as in \Cref{def:parameters-exprobustqr-sequences},

    $$\mathbb{P}_{S, |S| = s}\left(\left|\textsc{Count}_{\ell+1}(P, \sigma|_S)  - \mathbb{E}\left(\textsc{Count}_{\ell+1}(P, \sigma|_S) \right)\right|\geq  \frac{\beta_{\ell+1}}{4} \binom{s}{\ell+ 1} \sum_{\substack{I \subseteq [k] \\ |I| = \ell + 1}} \prod_{i \in I} p_{P_i} \right) \leq \exp(-\alpha_{\ell + 1}).$$
\end{lemma}

\begin{proof}[Proof of \Cref{lem:concentration-lemma-sequences}]
    We want to upper bound:
    \begin{equation}\label{eq:ub-3}
        \mathbb{P}_{S, |S| = s}\left(\left|\textsc{Count}_{\ell+1}(P, \sigma|_S)  - \mathbb{E}\left(\textsc{Count}_{\ell+1}(P, \sigma|_S) \right)\right|\geq  \frac{\beta_{\ell+1}}{4} \binom{s}{\ell+ 1} \sum_{\substack{I \subseteq [k] \\ |I| = \ell + 1}} \prod_{i \in I} p_{P_i} \right).
    \end{equation}

    We want to apply the inequality given in \Cref{thm:combes} to $\textsc{Count}_{\ell+1}(P, \sigma|_S)$, which utilizes high-probability bounded differences and requires some conditions on the parameters used.
    \\\\
    \textit{Step 1: Verifying the conditions of \Cref{thm:combes}.}
    This follows in exactly the same way as \Cref{sec:exp-rob-qr}, now with our new choice of parameters.
    \\\\
    \textit{Step 2: Bounded difference analysis.} Every length-$(\ell+1)$ subpattern occurrence that is added or eliminated when one sampled index is replaced contains that index. Deleting it produces a length-$\ell$ subpattern occurrence. Each length-$\ell$ subpattern of $P$ can be extended using at most $k-\ell$ additional positions of $P$, and we account for occurrences containing either the removed or the added sampled index. Therefore, in the high-probability case, the bounded difference is at most twice $k-\ell$ times the number of length-$\ell$ subpattern occurrences.

    By the exponentially robust quasirandomness assumption, in the ``high probability'' case, the number of length-$\ell$ subpatterns of $P$ is in $(1 \pm 2 \beta_{\ell}) \binom{s}{\ell} \sum_{\substack{I \subseteq [k] \\ |I| = \ell}} \prod_{i \in I} p_{P_i}$. The bounded difference is upper bounded by 
    $2(k-\ell) \binom{s}{\ell} \sum_{\substack{I \subseteq [k] \\ |I| = \ell}} \prod_{i \in I} p_{P_i}.$
    \\
    \textit{Applying the bounds.} We now apply \Cref{thm:combes} and the analysis above. We find that, for 
    $$\delta = \frac{\beta_{\ell+1}}{4} \binom{s}{\ell+ 1} \sum_{\substack{I \subseteq [k] \\ |I| = \ell + 1}} \prod_{i \in I} p_{P_i} \text{  and  } c = 2(k-\ell)\binom{s}{\ell} \sum_{\substack{I \subseteq [k] \\ |I| = \ell}} \prod_{i \in I} p_{P_i},$$ \Cref{eq:ub-3} is:
    $$\leq q + 2 \exp\left(-\frac{\delta^2}{8 s c^2} \right).$$
    The first term is upper-bounded by $\exp(-\alpha_\ell)=\frac{1}{3}\exp(-\alpha_{\ell+1})$ by definition. For the second term, we begin by observing that 
    $$\sum_{\substack{I \subseteq [k] \\ |I| = \ell + 1}} \prod_{i \in I} p_{P_i} \geq \frac{(k-\ell)p_{\min}}{\ell+1} \sum_{\substack{I \subseteq [k] \\ |I| = \ell}} \prod_{i \in I} p_{P_i}.$$ This holds because every $Q \subseteq P, |Q| = \ell$ can extend to a length-$(\ell + 1)$ subpattern in at least $k - \ell$ ways. Each way to extend contributes a factor of at least $p_{\min}$. Finally, each length-$(\ell + 1)$ subpattern is counted at most $\ell + 1$ times.

    Thus, $$\frac{\delta}{c} \geq \Omega\left( \beta_{\ell+1} \frac{(s-\ell)p_{\min}}{(\ell+1)^2} \right)$$ and $$\frac{\delta^2}{8sc^2} \geq \Omega\left( \beta_{\ell+1}^2 \frac{(s-\ell)^2p_{\min}^2}{s(\ell+1)^4} \right).$$ So \Cref{eq:ub-3} is at most
    $$\frac{1}{3}\exp(-\alpha_{\ell+1}) + 2 \exp\left(-\Omega\left( \beta_{\ell + 1}^2 \frac{(s-\ell)^2p_{\min}^2}{s(\ell+1)^4}\right) \right).$$

    For $s = \widetilde{\Omega}\left(\frac{k^5}{\ep^2 p_{\min}^2}\right)$, the second term is at most $\frac{2}{3}\exp(-\alpha_{\ell+1})$. Thus, \Cref{eq:ub-3} is at most $\exp(-\alpha_{\ell+1})$, and the concentration lemma follows.
\end{proof}

We also note the following fact.

\begin{lemma}\label{lem:ok-global-counts-exp-qr-sequences}
Suppose that a sequence $\sigma \in [m]^n$ satisfies the property that,
for some $\ell\le k$ and all $\ell'\leq \ell$,
$$ \textsc{Count}_{\ell'}(P,\sigma)
\in
(1\pm \beta_{\ell'})
\binom{n}{\ell'} \sum_{\substack{I \subseteq [k] \\ |I| = \ell'}} \prod_{i \in I} p_{P_i}. $$
Then $\sigma$ is
$(\alpha_\ell,2\beta_\ell,s,\ell,p)$-exponentially robustly
quasirandom with respect to $P$.
\end{lemma}

\begin{proof}
    We prove this statement via induction. For the base case, let $\ell=1$ and, for each $x\in[m]$, let $w_x$ be the number of occurrences of $x$ in $P$. For a uniformly random $s$-subset $S\subseteq[n]$, define $$f(S):=\textsc{Count}_1(P,\sigma|_S)=\sum_{i\in S}w_{\sigma_i}.$$
    Changing one sampled index changes $f$ by at most $k$, and
    $$\mathbb{E}[f(S)]=\frac{s}{n}\textsc{Count}_1(P,\sigma).$$
    Let $W_1:=\sum_{i\in[k]}p_{P_i}$. Since the global count is in $(1\pm\beta_1)nW_1$, the expectation is in $(1\pm\beta_1)sW_1$. Applying \Cref{thm:combes} with $\mathcal{Y}=\binom{[n]}{s}$, $q=0$, $c=k$, and $\delta=\beta_1sW_1$ gives
    $$\mathbb{P}\left(\left|f(S)-\mathbb{E}[f(S)]\right|\geq\beta_1sW_1\right)
    \leq 2\exp\left(-\frac{\beta_1^2sW_1^2}{8k^2}\right)
    \leq 2\exp\left(-\frac{\beta_1^2sp_{\min}^2}{8}\right),$$
    where the last inequality uses $W_1\geq kp_{\min}$. For the stated choice of $s$, this is at most $\exp(-\alpha_1)$. Therefore, with probability at least $1-\exp(-\alpha_1)$,
    $$f(S)\in(1\pm2\beta_1)sW_1,$$
    so $\sigma$ is $(\alpha_1,2\beta_1,s,1,p)$-exponentially robustly quasirandom with respect to $P$.

    For the inductive step, suppose that $\sigma$ is $(\alpha_{\ell-1},2\beta_{\ell-1},s,\ell-1,p)$-exponentially robustly quasirandom with respect to $P$. By \Cref{lem:concentration-lemma-sequences}, except with probability $\exp(-\alpha_\ell)$,
    $$\left|\textsc{Count}_{\ell}(P,\sigma|_S)-\mathbb{E}\left[\textsc{Count}_{\ell}(P,\sigma|_S)\right]\right|
    <\frac{\beta_\ell}{4}\binom{s}{\ell}\sum_{\substack{I \subseteq [k] \\ |I| = \ell}} \prod_{i \in I} p_{P_i}.$$
    Moreover,
    $$\mathbb{E}\left[\textsc{Count}_{\ell}(P,\sigma|_S)\right]
    =\textsc{Count}_{\ell}(P,\sigma)\frac{\binom{s}{\ell}}{\binom{n}{\ell}},$$
    which is in
    $$(1\pm\beta_\ell)\binom{s}{\ell}\sum_{\substack{I \subseteq [k] \\ |I| = \ell}} \prod_{i \in I} p_{P_i}$$
    by the hypothesis on the global count. Since $\beta_\ell+\beta_\ell/4\leq2\beta_\ell$, this proves that $\sigma$ is $(\alpha_\ell,2\beta_\ell,s,\ell,p)$-exponentially robustly quasirandom with respect to $P$.
\end{proof}

\begin{proof}[Proof of \Cref{lem:exp-qr-main-subsequences}]
    Let $\ell\in[k]$ be the first index for which the displayed deviation in the statement of the lemma occurs.

    If $\ell=1$, use the notation $f$ and $W_1$ from the preceding proof. Applying \Cref{thm:combes} with $\mathcal{Y}=\binom{[n]}{s}$, $q=0$, $c=k$, and $\delta=(\beta_1/4)sW_1$ gives
    $$\mathbb{P}\left(\left|f(S)-\mathbb{E}[f(S)]\right|\geq\frac{\beta_1}{4}sW_1\right)
    \leq2\exp\left(-\frac{\beta_1^2sW_1^2}{128k^2}\right)
    \leq2\exp\left(-\frac{\beta_1^2sp_{\min}^2}{128}\right)
    \leq\exp(-\alpha_1).$$
    Since $\mathbb{E}[f(S)]=(s/n)\textsc{Count}_1(P,\sigma)$ and the global count is not in $(1\pm\beta_1)nW_1$, it follows that, with probability at least $1-\exp(-\alpha_1)$,
    $$\textsc{Count}_1(P,\sigma|_S)\notin(1\pm3\beta_1/4)sW_1.$$

    Now suppose $\ell\geq2$. By the minimality of $\ell$ and \Cref{lem:ok-global-counts-exp-qr-sequences}, $\sigma$ is $(\alpha_{\ell-1},2\beta_{\ell-1},s,\ell-1,p)$-exponentially robustly quasirandom with respect to $P$. Therefore, by \Cref{lem:concentration-lemma-sequences}, with probability at least $1-\exp(-\alpha_\ell)$,
    $$\left|\textsc{Count}_{\ell}(P,\sigma|_S)-\mathbb{E}\left[\textsc{Count}_{\ell}(P,\sigma|_S)\right]\right|
    <\frac{\beta_\ell}{4}\binom{s}{\ell}\sum_{\substack{I \subseteq [k] \\ |I| = \ell}} \prod_{i \in I} p_{P_i}.$$
    Since
    $$\mathbb{E}\left[\textsc{Count}_{\ell}(P,\sigma|_S)\right]
    =\textsc{Count}_{\ell}(P,\sigma)\frac{\binom{s}{\ell}}{\binom{n}{\ell}}$$
    and the global count is not in the corresponding $(1\pm\beta_\ell)$ interval, the sampled count is not in the $(1\pm3\beta_\ell/4)$ interval appearing in the statement of the lemma.
\end{proof}

\subsubsection{The algorithm}\label{sec:pattern-algo}

Given the analysis of exponentially robust quasirandomness and counting subpatterns with dynamic programming, we now present the algorithm. The algorithm uses the principle of exponentially robust quasirandomness so that it only needs to query $N = \widetilde{O}\left(\frac{k^5}{\ep^2 p_{\min}^2}\right)$ indices. Then, dynamic programming handles the counting of subpatterns of the pattern we are interested in performing quality control of the counts for.

As discussed in \Cref{sec:pattern-dp}, the algorithm in fact solves the stronger $(\D,T_P)$-quality control problem. Recall that $T_{P,N}\approx 1$ certifies that, for every $\ell\in[k]$, the aggregate count of length-$\ell$ subpatterns of $P$ occurs with the right frequency. Since $T_P\approx 1$ implies $R_P\approx 1$ by considering $\ell=k$, this is a harder problem. We prove the following.

\begin{lemma}\label{lem:t-P-qc}
    Let $P$ be a length-$k$ sequence over $[m]$, $\ep > 0$, and $p$ be a distribution over $[m]$ with minimum weight $p_{\min}$. Define $\D = \{p^{\otimes n}\}_n$. Then, $(\D, T_P)$-quality control is achievable in $\widetilde{O}\left(\frac{k^5}{\ep^2 p_{\min}^2}\right)$ queries and $\widetilde{O}\left(\frac{k^7}{\ep^2 p_{\min}^2}\right)$ time.
\end{lemma}

We will then prove \Cref{thm:subsequences} as a corollary of \Cref{lem:t-P-qc}.

\begin{figure}[H]
\centering
\fbox{
\parbox{\dimexpr0.95\linewidth-2\fboxsep-2\fboxrule\relax}{
\textbf{Algorithm Pattern-Count:} \\
On inputs: A length-$n$ sequence $\sigma$ in $[m]^n$, a length-$k$ pattern $P$, parameter $\ep \in (0, 1)$, a probability distribution $p$ over $[m]$ with minimum probability $p_{\min}$, $n \geq N$
\begin{enumerate}
    \item Sample $N = \widetilde{O}\left(\frac{k^5}{\ep^2 p_{\min}^2}\right)$ indices; call the subsequence on these indices $\pi$.
    \item Run \textsc{Dynamic-Programming} on inputs: $\pi$, $P$, $p$, and $\ep/8$.
    \item Accept iff the call to \textsc{Dynamic-Programming} accepts; reject otherwise.
\end{enumerate}
}
}
\caption{Algorithm for quality control of pattern counts in a sequence in $[m]^n$, as compared to sequences from $p^{\otimes n}$.  Exponentially robust quasirandomness ensures that, upon sampling $N$ indices, the counts of the pattern $P$ among the indices sampled are proportional to the global count. Then, dynamic programming is utilized to count $P$ and the smaller subpatterns within $P$ in order to perform quality control.}
\label{figure:subsequence-alg}
\end{figure}

We now prove \Cref{lem:t-P-qc}.

\begin{proof}[Proof of \Cref{lem:t-P-qc}]
    \textbf{Completeness.} We first prove that a sequence $\sigma \sim p^{\otimes n}$ is accepted with probability $1 - o(1)$. First, if $\sigma \sim p^{\otimes n}$ and $S \subseteq [n]$, then $\sigma|_S$ is distributed as $p^{\otimes |S|}$. For our algorithm, $|S| = N$. By \Cref{prop:min-weight-sequence-aggregate-copies}, the probability that one of the tests fails is at most $O\left(\frac{k^3}{Np_{\min}\ep^2}\right)=O\left(\frac{p_{\min}}{k^2\log(2+k/(\ep p_{\min}))}\right)=o(1)$ as $k/(\ep p_{\min})\to\infty$. Thus, for each $\ell \in [k]$, the number of times subpatterns of $P$ appear as subsequences in $\sigma|_S$ is in $(1 \pm \ep/8) \binom{N}{\ell} \cdot \sum_{\substack{I \subseteq [k] \\ |I| = \ell}} \prod_{i \in I} p_{P_i}$. Thus, when \textsc{Pattern-Count} calls \textsc{Dynamic-Programming}, it accepts. Thus, random sequences over $[m]^n$ are accepted with probability $1 - o(1)$.
    \\\\ \textbf{Soundness.} Suppose that $T_{P,n}(\sigma) \notin 1 \pm \ep$. Then, there exists some $\ell \in [k]$ such that
    $$\textsc{Count}_{\ell}(P,\sigma) \notin (1\pm\ep)\binom{n}{\ell}\sum_{\substack{I \subseteq [k] \\ |I| = \ell}} \prod_{i \in I} p_{P_i}.$$
    Since $\beta_\ell\leq\ep$ for every $\ell\in[k]$, there is a minimum $\ell \in [k]$ such that $\textsc{Count}_{\ell}(P, \sigma) \not\in (1 \pm \beta_\ell) \binom{n}{\ell} \sum_{\substack{I \subseteq [k] \\ |I| = \ell}} \prod_{i \in I} p_{P_i}$. By \Cref{lem:exp-qr-main-subsequences}, on a random subset $S$ of size $N$, with probability at least $1-\exp(-\alpha_\ell)$,
    $$\textsc{Count}_{\ell}(P, \sigma|_S) \not\in (1 \pm 3\beta_{\ell}/4) \binom{N}{\ell} \sum_{\substack{I \subseteq [k] \\ |I| = \ell}} \prod_{i \in I} p_{P_i}.$$
    Since $\beta_\ell \geq \ep/4$, this implies that, with probability at least $1-\exp(-\alpha_\ell)$,
    $$\textsc{Count}_{\ell}(P, \sigma|_S) \not\in (1 \pm 3\ep/16) \binom{N}{\ell} \sum_{\substack{I \subseteq [k] \\ |I| = \ell}} \prod_{i \in I} p_{P_i}.$$
    Since $3\ep/16>\ep/8$, this will be detected by \textsc{Dynamic-Programming}, which is run with parameter $\ep/8$, and thus \textsc{Pattern-Count} will reject. Therefore, \textsc{Pattern-Count} solves $(\D,T_P)$-quality control.
\end{proof}

\begin{proof}[Proof of \Cref{thm:subsequences} from \Cref{lem:t-P-qc}]
    Suppose $A$ is a $(\D, T_{P})$-quality control algorithm. We argue that it is also a $(\D, R_{P})$-quality control algorithm. For completeness, $A$ accepts $x \sim \D$ with probability $1 - o(1)$. For soundness, by definition of $R_{P, n}$, $R_{P, n}(f) \not \in 1 \pm \ep$ implies $T_{P, n}(f) \not \in 1 \pm \ep$. Thus, $A$ rejects on $f$ with $R_{P, n}(f) \not \in 1 \pm \ep$. We conclude that $A$ is a $(\D, R_{P})$-quality control algorithm.
\end{proof}

\subsection{Quality control lower bound}\label{sec:qc-lb-all}

We turn to proving the lower bound component of \Cref{thm:any-subsequence-intro}, stated below.

\begin{theorem}\label{thm:qc-lb-sequences-general1}
    Let $k \in \mathbb{N}$, let $p$ be a distribution over $[m]$ with minimum probability $p_{\min}$, let $\ep \in (0,1)$, and define $\D = \{p^{\otimes n}\}_n$. Then, there exists a length-$k$ pattern $P$ such that $(\D,R_P)$-quality control requires $\Omega\left(\frac{k^2}{p_{\min}\ep^2}\right)$ queries. This lower bound holds for every input length $n$ satisfying $np_{\min}=\Omega\left(\frac{k^2}{\ep^2}\right)$.
\end{theorem}

We observe the following consequence for when $p = \mathrm{Unif}([m])$.

\begin{theorem}\label{thm:qc-lb-sequences-uniform}
    Let $\D = \{\mathrm{Unif}([m])^{\otimes n}\}_n$. There exists a length-$k$ pattern $P$ such that  $(\D, R_{P})$-quality control requires $\Omega\left(\frac{m k^2 }{\ep^2}\right)$ queries.
\end{theorem}

As in \Cref{sec:quality-control-permutation-lb}, we prove the lower bound for quality control via constructing distributions over sequences that are indistinguishable from each other with high probability unless $\Omega\left(\frac{k^2}{p_{\min} \ep^2} \right)$ queries are made. 

    \begin{proof}[Proof of \Cref{thm:qc-lb-sequences-general1}]
    Since we are performing $(\D, R_{P})$-quality control, in the ``accept'' case, for any $n \in \mathbb{N}$, the distribution must be $p^{\otimes n}$. In the ``reject'' case, we consider sequences distributed as $q^{\otimes n}$ for some other distribution $q$.

    Let $p_{\min}$ be the probability of the (or, a) minimum-probability element in $p$. Call the minimum probability element $v_{\min}$, and let $q_{\min}$ be the probability of $v_{\min}$ under $q$. We define distribution $q$ to be the distribution whose probability on $v_{\min}$ is $q_{\min}$ and whose conditional distribution on $[m]\setminus \{v_{\min}\}$ matches $p$'s.
    \\\\
    We couple sequences $x_{\textsc{yes}} \sim p^{\otimes n}$ and $x_N \sim q^{\otimes n}$ as follows. For each $i \in [n]$:
    \begin{enumerate}
        \item Sample $u \sim \mathrm{Unif}(0, 1)$ and an item $w \in [m]\setminus \{v_{\min}\}$ according to the marginal distribution of $p$ on $[m]\setminus \{v_{\min}\}$.
        \item In $x_{\textsc{yes}}$, place $v_{\min}$ in position $i$ if $u \leq p_{\min}$, and place $w$ in position $i$ otherwise.
        \item In $x_N$, place $v_{\min}$ in position $i$ if $u \leq q_{\min}$, and place $w$ in position $i$ otherwise.
    \end{enumerate}
    Note that $x_{\textsc{yes}}$ and $x_N$ agree except at positions where $p_{\min} < u \leq q_{\min}$, at which $x_N$ receives $v_{\min}$ while $x_{\textsc{yes}}$ receives $w$.

    The pattern $P$ that we consider $(\D, R_{P})$-quality control for is $P = v_{\min}^{\otimes k}$. We prove the following.

    First, what does $q_{\min}$ need to be? We answer this below.

    \begin{lemma}\label{lem:manycopies}
    In order for $x_N \sim q^{\otimes n}$ to have $(1 + \ep) \binom{n}{k} p_{\min}^k$ copies of $P$ with probability $\geq 9/10$, it suffices for $q_{\min} \geq (1 + 2\ep)^{1/k} p_{\min}$ when $n p_{\min}$ is $\Omega(k^2/\ep^2)$.
\end{lemma}
\begin{proof}[Proof of \Cref{lem:manycopies}]
    The number of occurrences of $v_{\min}$ is distributed as $Z \sim Binom(n, q_{\min})$. The number of copies of $P$ is $\binom{Z}{k}$.
    
    So, for the count of $P$ to be at least $(1 + \ep) \binom{n}{k} p_{\min}^k$, it suffices for $Z \geq (1 + \ep)^{1/k} n p_{\min} + k - 1$, because:
    $$\binom{Z}{k} \geq \frac{(Z - k + 1)^k}{k!} \geq (1 + \ep)\frac{n^k}{k!} p_{\min}^k \geq (1 + \ep) \binom{n}{k}p_{\min}^k.$$

    We apply Chernoff bounds to understand the probability that $Z$ is above this threshold. Chernoff  bounds tell us:
    $$\mathbb{P}\left(\mathrm{Binom}(n, q_{\min}) \geq (1 + \ep)^{1/k} n p_{\min} + k - 1 \right) \geq 9/10$$
    for $q_{\min} \geq (1 + 2\ep)^{1/k} p_{\min}$ and $n p_{\min} = \Omega(k^2/\ep^2)$. To see this, let's first observe that, by our setting of $q_{\min}$, for $X \sim \mathrm{Binom}(n, q_{\min})$, $\mathbb{E}[X] = n q_{\min} \geq (1 + 2\ep)^{1/k} np_{\min}$. We next find a $\delta \in (0, 1)$ such that
    $$(1 + \ep)^{1/k} n p_{\min} + k - 1 \leq (1 - \delta) (1 + 2\ep)^{1/k} np_{\min}.$$
    Equivalently, 
    $$\delta \leq \left(1 - \left( \frac{1 + \ep}{1 + 2\ep}\right)^{1/k} \right) - \frac{k - 1}{(1 + 2\ep)^{1/k}np_{\min}}.$$
    Using that $\exp(x) \approx 1 + x$ for small $x$, the first term is $1 - \left(\frac{1+\ep}{1+2\ep}\right)^{1/k} = \Theta(\ep/k)$ for large $k$ and $\ep \in (0,1)$. The second term is at most $\frac{k}{np_{\min}} = O(\ep^2/k)$ when $np_{\min} = \Omega(k^2/\ep^2)$, and is thus dominated by the first. Therefore, the RHS is at least $c\ep/k$ for some constant $c$, and so it suffices for $\delta \leq c \ep / k$.
    Thus, by Chernoff bounds,  $$\mathbb{P}\left(\mathrm{Binom}(n, q_{\min}) < (1 + \ep)^{1/k} n p_{\min} + k - 1 \right) \leq \exp(-\Omega(\ep^2 n p_{\min} / k^2)) \leq 1/10.$$
    Therefore, in order for $x_N \sim q^{\otimes n}$ to have $(1 + \ep) \binom{n}{k} p_{\min}^k$ copies of $P$ with probability $\geq 9/10$, it suffices for $q_{\min} \geq (1 + 2\ep)^{1/k} p_{\min}$ when $n p_{\min}$ is $\Omega(k^2/\ep^2)$.
\end{proof}

    We also note the following.

    \begin{lemma}\label{lem:distinguishing}
        Distinguishing if a sequence is constructed according to $x_{\textsc{yes}} \sim p^{\otimes n}$ or $x_N \sim q^{\otimes n}$ requires at least as many queries as distinguishing $(\text{Bern}(p_{\min}))^{\otimes n}$ from $(\text{Bern}(q_{\min}))^{\otimes n}$.
    \end{lemma}

    \begin{proof}[Proof of \Cref{lem:distinguishing}]
        Observe that the count of $v_{\min}$ is distributed as $(\text{Bern}(p_{\min}))^{\otimes n}$ in $x_{\textsc{yes}} \sim p^{\otimes n}$ and as $(\text{Bern}(q_{\min}))^{\otimes n}$ in $x_N \sim q^{\otimes n}$.

        Suppose we want to distinguish $(\text{Bern}(p_{\min}))^{\otimes n}$ from $(\text{Bern}(q_{\min}))^{\otimes n}$. Suppose we have samples from $\mathcal{D}$ which is either $(\text{Bern}(p_{\min}))^{\otimes n}$ or $(\text{Bern}(q_{\min}))^{\otimes n}$. One method to tell which distribution $\mathcal{D}$ is would be to construct a sequence $y$ which places $v_{\min}$ at all of the indices where $y \sim q$ is $1$. For the rest of the indices, sample according to the marginal distribution of $p$ on $[m] \setminus \{v_{\min}\}$. Observe that the sequence is either distributed according to $p^{\otimes n}$ or $q^{\otimes n}$, and thus we can run the supposed algorithm for distinguishing between sequences from the two distributions. This will tell us which distribution $\mathcal{D}$ is.

        Therefore, distinguishing if a sequence is constructed according to $x_{\textsc{yes}} \sim p^{\otimes n}$ or $x_N \sim q^{\otimes n}$ requires at least as many queries as distinguishing $(\text{Bern}(p_{\min}))^{\otimes n}$ from $(\text{Bern}(q_{\min}))^{\otimes n}$.
    \end{proof}

    We are now ready to proceed with the proof of \Cref{thm:qc-lb-sequences-general1}. Let $\delta := (q_{\min} - p_{\min})/p_{\min}$. Distinguishing $\text{Bern}(p_{\min})$ from $\text{Bern}(q_{\min})$ requires $q = \Omega\left( \frac{1}{KL(p_{\min} || q_{\min})}\right) = \Omega(\frac{1}{p_{\min} \delta^2})$ queries.

   In our case, $(q_{\min} - p_{\min})/p_{\min} = (1 + 2 \ep)^{1/k} - 1 \approx \frac{\ln(1 + 2\ep)}{k} = \Theta\left(\frac{\ep}{k}\right)$. So we require $q = \Omega\left(\frac{k^2}{\ep^2 p_{\min}}\right)$ queries.
\end{proof}

\section{Worst-case lower bounds}\label{sec:worst-case-lower-bounds}

We now move to proving $\exp(k)$ lower bounds for determining the count of permutation patterns and patterns in worst-case sequences. \Cref{sec:worst-case-permutation} proves the worst-case lower bound for permutation patterns, and \Cref{sec:worst-case-pattern} proves it for patterns.

Both lower bounds rely on a similar principle. Begin with an instance that has the right count of the (permutation) pattern to be accepted. Then, randomly perturb the instance in order to create many more copies of the (permutation) pattern in a way that is hard to detect. Interestingly, we will see that, in both cases, the instances that we start with look very far from random instances, which highlights the difference in the types of instances that quality control versus worst-case algorithms must accept.

\subsection{Worst-case lower bound for permutation patterns}\label{sec:worst-case-permutation}

We now prove that any algorithm for determining whether the count of a length-$k$ permutation pattern is in $(1 \pm \eps) \binom{n}{k} \cdot \frac{1}{k!}$ in a \textit{worst-case} length-$n$ sequence requires $\exp(k)$ queries.

\begin{theorem}\label{thm:worstcasepermutation}
    Let $\tau$ be the length-$k$ increasing permutation pattern (corresponding to $\tau(i) = i$). Any algorithm that distinguishes length-$n$ sequences with $\not \in (1 \pm \eps) \binom{n}{k} \cdot \frac{1}{k!}$ copies of $\tau$ from sequences with $(1 \pm \eps) \binom{n}{k} \cdot \frac{1}{k!}$ copies of $\tau$ must use $\Omega_{\ep}\left( \frac{k^{k+2}}{e^{2k}}\right)$ queries.
\end{theorem}

To prove this theorem, we construct two distributions over length-$n$ sequences, one with $\in (1 \pm \eps) \binom{n}{k} \cdot \frac{1}{k!}$ copies of the length-$k$ increasing pattern $\tau$ and the other with $\not \in (1 \pm \eps) \binom{n}{k} \cdot \frac{1}{k!}$ copies.\footnote{Let $\sigma$ be the permutation corresponding to the length-$n$ sequence. Since $\tau$ is the increasing pattern, a copy of $\tau$ is a subsequence $(i_1, i_2, \dots, i_k)$ such that $\sigma_{i_1} < \sigma_{i_2} < \dots < \sigma_{i_k}$.} These will be sequences over $[0, 1]^n$, but we will ensure the number of copies of the increasing patterns is in the corresponding ranges.

Before giving the two distributions, we define parameters $a$ and $b$ as follows:
$$a = \frac{1}{2} \left(\frac{n}{k-1} - \sqrt{\left( \frac{n}{k-1}\right)^2 - 4 \binom{n}{k} \frac{1}{k!}\left(\frac{k-1}{n} \right)^{k-2}} \right), ~~~~ b = \frac{8 (1 + 2 \ep)}{n} \binom{n}{k} \frac{1}{k!} \left( \frac{k-1}{n}\right)^{k - 1}.$$

\begin{figure}[h!]
    \centering
    \input{Visuals/lowerbound-worstcase-visual}
    \caption{Outline of the construction of the instances considered for the worst-case approximation lower bound for the count of length-$k$ increasing sequences. In both of the constructions, we first partition the length-$n$ sequence into $k-1$ blocks and construct decreasing sequences in each block (in increasing intervals). Lower some values in the first block to create some initial copies of the length-$k$ increasing permutation pattern. In Distribution 1, do nothing further. In Distribution 2, for the values in the first halves of each block (marked as blue), randomly choose elements to fall into a lower portion of the range. This causes a substantial enough increase in the counts of the increasing pattern.}
    \label{fig:lowerbound-worstcase-visual}
\end{figure}
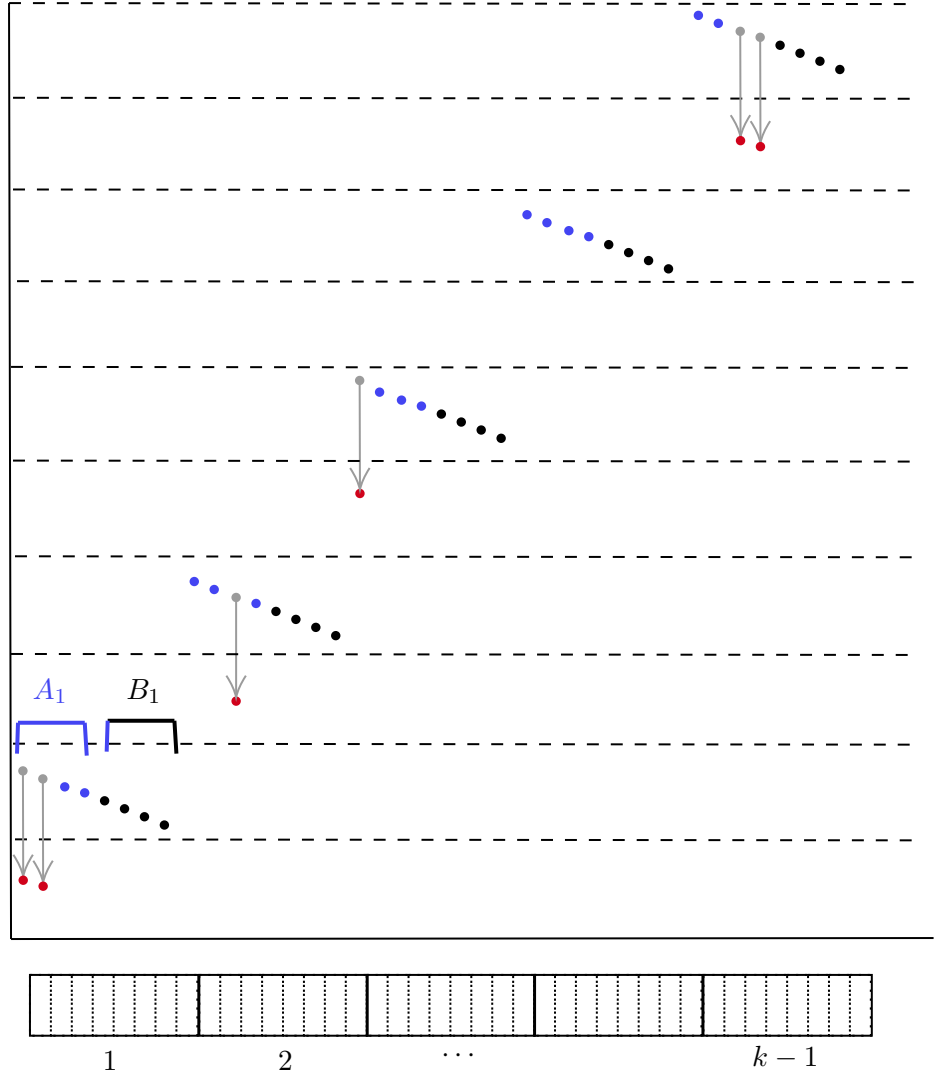

\paragraph{Distribution 1: with $(1 \pm \eps) \binom{n}{k} \cdot \frac{1}{k!}$ copies of $\tau$ with high probability.}

Sequences are drawn from Distribution 1 $(\mathcal{D}_1)$ as follows.

\begin{enumerate}
    \item Partition the length-$n$ sequence into $k-1$ consecutive blocks ($C_j$, $j \in [k-1]$), each of length
    $\ell = \frac{n}{k-1}$.
    \item Split each block $j \in [k-1]$ into two halves: $A_j$, the first $\ell/2$ positions,
    and $B_j$, the remaining $\ell/2$.
    \item Assign a value to every entry of every block $C_j$, $j \in [k-1]$. For the entry at
    position $i \in \{1, \dots, \ell\}$ within $C_j$, set its value to
    $$
        v_j(i) := \frac{2j-1}{2(k-1)} + \frac{1}{2(k-1)}\cdot\frac{\ell - i}{\ell},
    $$
    a decreasing run in $i$ inside the upper sub-band
    $\left[\frac{2j-1}{2(k-1)}, \frac{2j}{2(k-1)}\right)$.
    \item Lower the first $a$ entries of $A_1$ by one band: replace each such value $v$
    by $v - \frac{1}{2(k-1)}$, placing them in $\left[0, \frac{1}{2(k-1)}\right)$.
\end{enumerate}
Observe that, with $a = 0$, sequences from Distribution 1 have zero copies of $\tau$; the choice of $a$ is made so that the sequence has $(1 \pm \eps) \binom{n}{k} \cdot \frac{1}{k!}$ copies of $\tau$ with high probability. Also, observe that Distribution 1 is actually a deterministic construction (and only Distribution 2 will be a distribution over more than one sequence).

\paragraph{Distribution 2: with $> (1 + \eps) \binom{n}{k} \cdot \frac{1}{k!}$ copies of $\tau$ with high probability.} Sequences are drawn from Distribution 2 $(\mathcal{D}_2)$ as follows.

\begin{enumerate}
    \item Partition and split into blocks $C_j = A_j \cup B_j$ exactly as in Distribution 1.
    \item Assign $v_j(i)$ to every entry of every block $C_j$, $j \in [k-1]$, exactly as in
    Step 3 of Distribution 1.
    \item For each entry of every $A_j$, $j \in [k-1]$, independently with probability $b$
    lower it by one band: replace $v_j(i)$ by $v_j(i) - \frac{1}{2(k-1)}$,
    placing it in the lower sub-band $\left[\frac{2j-2}{2(k-1)}, \frac{2j-1}{2(k-1)}\right)$.
    \item Lower the first $a$ entries of $A_1$ by one band (as in Step 4 of
    Distribution 1).
\end{enumerate}
In Step 3, the chosen indices contribute to a decreasing sequence in the block $\left[ \frac{2j - 2}{2(k-1)}, \frac{2j - 1}{2(k-1)}\right)$. We refer to these indices as having ``dropped values.'' Indices that are not chosen in Step 3 contribute to a decreasing sequence in the block $\left[ \frac{2j - 1}{2(k-1)}, \frac{2j}{2(k-1)}\right)$ (except those modified in Step 4). The main idea is that, when indices are chosen to contribute to the lower increasing sequence, there are now at least \textit{two} choices of elements per block (one from $A_j$ and one from $B_j$) that can be used to create an occurrence of a  length-$k$ increasing sequence.

\paragraph{Characterizations of length-$k$ increasing sequences in the constructions.} Let $C_j = A_j \cup B_j$ be a block, for $j \in [k-1]$. In $x$ from $\mathcal{D}_1$, the occurrences of the length-$k$ increasing pattern $\tau$ have one element from the first $a$ entries of $C_1$, one element from the other entries of $C_1$, and one element in each $C_j$, $j \geq 2$. In $y \sim \mathcal{D}_2$, the same occurrences appear. But now there are also occurrences that take one element from each $C_j$, $j \geq 1$ except one $C_i$, from which it takes an element from $A_i$ whose value was dropped and an element from $B_i$. The dropping of values in $A_j$s allows for these extra occurrences. Since elements in the $A_j$s are dropped randomly, it will be difficult for any algorithm to distinguish between the $\mathcal{D}_1$ and $\mathcal{D}_2$ cases, as this requires finding one of the dropped values in the construction of $y \sim \mathcal{D}_2$.
\\\\
We prove the following lemmas regarding the number of copies of the increasing pattern $\tau$ in sequences drawn from each distribution.

\begin{lemma}\label{lem:d1}
    $x \sim \mathcal{D}_1$ has $(1 \pm \ep) \binom{n}{k} \cdot \frac{1}{k!}$ copies of the increasing pattern.
\end{lemma}

\begin{lemma}\label{lem:d2}
    With probability $0.99$, $x \sim \mathcal{D}_2$ has $> (1 + \ep) \binom{n}{k} \cdot \frac{1}{k!}$ copies of the increasing pattern.
\end{lemma}

\begin{proof}[Proof of \Cref{lem:d1}]
    The copies of $\tau$ involve: One of the first $a$ indices of $A_1$ (whose values were lowered in the block in order to create just enough copies of the length-$k$ increasing pattern), an index in the remaining indices of block $1$, and an index from each of the blocks $2, \dots, k-1$. The number of copies $N_\tau(x)$ is, therefore,
    $$N_\tau(x) = a \cdot \left(\frac{n}{k-1} - a \right) \cdot  \left(\frac{n}{k-1}\right)^{k-2}.$$
    This is in $(1 \pm \ep) \binom{n}{k} \cdot \frac{1}{k!}$ for our choice of $a$.
\end{proof}

\begin{proof}[Proof of \Cref{lem:d2}]
    Suppose that, over all indices in $A_j, j \in [k-1]$, $X$ of the indices were chosen in step 3 to have lower values.

    A copy of the increasing pattern $\tau$ can be formed with one of the lowered indices and one index from each remaining $B_{j'}$. The number of copies of $\tau$, $N_\tau(x)$, is, therefore:
    $$N_\tau(x) \geq X \cdot \frac{n}{2(k-1)} \cdot \left(\frac{n}{k-1} \right)^{k-2} = \frac{X}{2}\left(\frac{n}{k-1} \right)^{k-1}.$$
    This is $> (1 + \ep) \binom{n}{k} \cdot \frac{1}{k!}$ for $$X > 2(1 + \eps) \binom{n}{k} \frac{1}{k!} \left( \frac{k-1}{n}\right)^{k-1}.$$
    Now, by construction $X \sim \text{Binom}(n/2, b)$. Thus, 
    $$\mathbb{P}\left(X \leq 2(1 + \eps) \binom{n}{k} \frac{1}{k!} \left( \frac{k-1}{n}\right)^{k-1} \right) \leq \frac{1}{100}$$
    by the setting of $b$.
\end{proof}

We are now ready to prove \Cref{thm:worstcasepermutation}.

\begin{proof}[Proof of \Cref{thm:worstcasepermutation}]
    Distinguishing whether $x$ is drawn from $\mathcal{D}_1$ versus $\mathcal{D}_2$ reduces to distinguishing whether a certain event (a lowered value) occurs with probability $0$ ($\mathcal{D}_1$) or $b$ ($\mathcal{D}_2$). Thus, an algorithm with $q$ queries has the task of distinguishing $q$ copies of $\text{Bern}(0)$ from $q$ copies of $\text{Bern}(b)$.

    The number of queries needed to distinguish between these cases is:
    $$q = \Omega\left( \frac{1}{b}\right).$$ 
    
    In the setting of $b$ from the construction of Distribution 2, this gives $$q = \Omega_{\ep}\left( \frac{(k!)^2}{(k-1)^{k-1}}\right) = \Omega_{\ep}\left( \frac{k^{k+2}}{e^{2k}}\right).$$
\end{proof}

\subsection{Worst-case lower bound for patterns}\label{sec:worst-case-pattern}

We prove the following.

\begin{theorem}\label{thm:worst-case-subsequences}
    Let $m \geq k$ and let $P$ be a length-$k$ pattern with distinct symbols. Any algorithm that distinguishes whether any length-$n$ sequence $x$ has $\in (1 \pm \ep) \binom{n}{k} \frac{1}{m^k}$ copies of $P$ or $\not \in (1 \pm \ep) \binom{n}{k} \frac{1}{m^k}$ copies with high probability must use $\Omega\left( \ep m^k\right)$ queries when $m^k \ll n$.
\end{theorem}

We construct sequences which have $\in (1 \pm \ep) \binom{n}{k} \frac{1}{m^k}$ copies of a pattern or $0$, which is $\not \in (1 \pm \ep) \binom{n}{k} \frac{1}{m^k}$, copies and show that it is not possible to distinguish the two with high probability without $\Omega\left( \ep m^k\right)$ queries.

\begin{proof}[Proof of \Cref{thm:worst-case-subsequences}]

Let $P = P_1 P_2 \dots P_k$ be any length-$k$ sequence over $[m]^k$, such that all $P_i$ are distinct.
\\\\
We first construct sequences $S$ of length $n$ with zero copies of $P$:
\begin{enumerate}
    \item Break up $[n]$ into consecutive intervals of length $\frac{n}{2(k-1)}$.
    \item In the first two intervals, place the value $P_1$. In the second two intervals, place the value $P_2$, and so forth for all future intervals.
\end{enumerate}

Observe that there are no occurrences of $P_k$ and thus $P$ never appears in $S$ as constructed above.
\\\\
We now create a distribution over sequences $T$ of length $n$ with $\in (1 \pm \ep) \binom{n}{k} \frac{1}{m^k}$ copies of $P$:
\begin{enumerate}
    \item Start with $S$ as generated above.
    \item Divide the rightmost interval into $N/\ell$ consecutive sub-intervals, for $N = \frac{n}{2(k-1)}$ and $\ell = \frac{1}{1 + \ep/2} \binom{n}{k}\frac{1}{m^k} \frac{2 (k-1)^{k-1}}{n^{k-1}}$.
    \item Choose a sub-interval of the rightmost interval uniformly at random among the first $\ep N / \ell$ subintervals, and place $P_k$ at all indices in the subinterval chosen.
\end{enumerate}

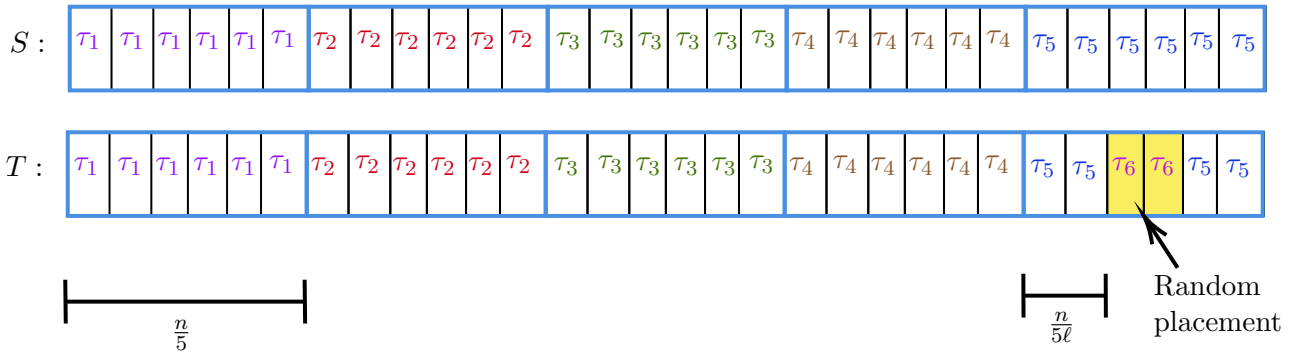
\begin{figure}[H]
    \centering
    \input{Visuals/subsequence-worst-case-lb-figure}
    \caption{Construction of sequences ($S$) with zero copies of $P = P_1 P_2 P_3 P_4 P_5 P_6$ and sequences ($T$) with $(1 \pm \ep) \binom{n}{k} \frac{1}{m^k}$ copies of $P$. In the latter case, the final element $P_6$ is placed in a randomly chosen block of $\ell$ indices, where $\ell$ is chosen so that the correct number of copies of $P$ appear with high probability. To distinguish between $S$ and $T$, an algorithm must find which random block the elements $P_6$ were planted in.}
    \label{fig:subsequence-worst-case-lb-figure}
\end{figure}

We first justify the choice of $\ell$. Suppose that, in the construction of $T$, in step 3, $P_k$ is placed in the $i$th block of the rightmost interval. Then, the number of occurrences of $P$ in $T$ -- denoted $C_{P}(T, i)$ is:
$$C_{P}(T, i) = \left(\frac{n}{k-1}\right)^{k-2} \left( \frac{n}{2(k-1)} + (i-1) \ell \right) \cdot \ell.$$

Since $i \sim \mathrm{Unif}([\frac{\ep N}{\ell}])$, the expected number of occurrences of $P$ is:
$$\mathbb{E}_i[C_{P}(T, i)] = \left(\frac{n}{k-1}\right)^{k-2} \left( \frac{n}{2(k-1)} + \frac{1}{2}\left(\frac{\ep n}{2 \ell (k-1)}-1\right) \ell \right) \cdot \ell.$$
For the choice of $$\ell = \frac{1}{1 + \ep/2} \binom{n}{k}\frac{1}{m^k} \frac{2 (k-1)^{k-1}}{n^{k-1}}.$$
(up to taking the floor), the expectation is $(1 - o(1)) \binom{n}{k} \cdot \frac{1}{m^k}$. Observe that, for constant $k$ and $m^k \ll n$, $\ell = \Theta(n/m^k)$.

Using that $i \sim \mathrm{Unif}([\frac{\ep N}{\ell}])$, we find that 
$$C_{P}(T, i) \in (1 \pm \ep) \binom{n}{k} \frac{1}{m^k}$$
for every possible $i \in [\frac{\ep N}{\ell}]$.

To distinguish between $S$ and a sequence drawn from the distribution over sequences $T$, the algorithm must identify which of the $N/\ell$ consecutive sub-intervals, for $N = \frac{n}{2(k-1)}$, has been assigned the value $P_k$. Since the choice of the sub-interval with value $P_k$ is uniform at random among the $N$ possible choices, the query complexity of any algorithm distinguishing between the two cases with high probability must be at least
$$\Omega\left( \frac{\ep N}{\ell}\right) = \Omega\left( \frac{n}{2(k-1)} \cdot \frac{\ep}{\ell}\right).$$
For constant $k$, this is:
$$\Omega\left( \ep m^k\right) \text{ for } m^k \ll n.$$
\end{proof}

\bibliographystyle{alpha}
\bibliography{bibliography}

\appendix

\section{Appendix}\label{sec:appendix}

\subsection{Proofs of preliminary propositions}\label{appendix:prelims}

\subsubsection{Concentration inequality}

We first state a concentration inequality of Pemantle and Peres for variables under a negative-dependence hypothesis.

\begin{proposition}[Corollary 5.2 of \cite{pemantle2014concentration}]\label{prop:mcdiarmid}
    Let $f: \binom{[n]}{t} \to \mathbb{R}$ be a 1-Lipschitz function with respect to the Hamming distance on the bases of a balanced matroid of rank $t$ on $n$ elements. Then
    $$\mathbb{P}\left(\left|f(A) - \mathbb{E}[f(A)] \right| \geq \delta\right) \leq 2 \exp\left( - \frac{\delta^2}{8t}\right).$$
\end{proposition}

We rely on a concentration inequality for functions that have a ``bounded difference'' property on a high-probability subset of the domain. We prove the following subset analogue of Proposition 2 of \cite{combes2024extension}, which is a high-probability bounded differences inequality.

\begin{proposition}[Subset analogue of Proposition 2 of \cite{combes2024extension}]\label{thm:combes-appendix}
    Let $f: \binom{[n]}{t} \to \mathbb{R}$ be a bounded, non-negative function, and $\mathcal{Y} \subseteq \binom{[n]}{t}$. Suppose the following ``bounded differences property'' holds for $f$ and $\mathcal{Y}$:

    For all $A, B \in \mathcal{Y}$,
    $$\left|f(A) - f (B)\right| \leq \frac{c}{2} |A \triangle B|$$
    for some constant $c$.

    Consider any $\delta > 0$. Let $S$ be a uniformly random element of $\binom{[n]}{t}$. Let $$q:= 1 - \mathbb{P}\left( S \in \mathcal{Y}\right),$$ and suppose $q \leq \frac{\delta}{2 \max(f)}.$

    Then, the following holds:
    $$\mathbb{P}\left( \left| f(S) - \mathbb{E}\left( f(S)\right)\right| \geq \delta\right) \leq q +  2 \exp\left(- \frac{\delta^2}{8 t c^2} \right).$$
\end{proposition}

\begin{proof}
    This proof is a modification of the proof of Lemma 2 in \cite{combes2024extension}, adapted to the setting of subsets (which requires relying on \cite{pemantle2014concentration} rather than McDiarmid's inequality \cite{mcdiarmid1989method}) and stated in terms of $\mathbb{E}(f(A))$ rather than $\mathbb{E}(f(A) | A \in \mathcal{Y})$.

    First, we assume that $f$ has the $c$-bounded difference property over subset $\mathcal{Y}$, meaning for all $A, B \in \mathcal{Y}$, $|f(A) - f(B)| \leq \frac{c}{2} |A \triangle B|$. This means that $f$ is 1-Lipschitz over $\mathcal{Y}$ according to distance metric $d_c(A, B) := \frac{c}{2} |A \triangle B|$.

    As in \cite{combes2024extension}, by \cite{mcshane1934extension}, the function 
    $$g(A) := \inf_{B \in \mathcal{Y}} \{f(B) + d_c(A, B) \}$$
    is a 1-Lipschitz function over $\binom{[n]}{t}$ such that $g(A) = f(A)$ on all $A \in \mathcal{Y}$. Observe that $g$, and thus $f$, has the $c$-bounded difference property over all of $\binom{[n]}{t}$.

    Next, observe that $|f(A) - g(A)| \leq \max(f) \cdot 1(A \not \in \mathcal{Y})$. Taking an expectation according to the uniform distribution over $\binom{[n]}{t}$, and applying Jensen's inequality, we find that
    $$\left|\mathbb{E}[f(S)] - \mathbb{E}[g(S)] \right| \leq \max(f) \cdot \mathbb{P}(S \not \in \mathcal{Y}) =: \max(f) \cdot q.$$
    By assumption that $q \leq \frac{\delta}{2 \max(f)}$, this implies that
    \begin{equation}\label{eq:combes-1}
      \left|\mathbb{E}[f(S)] - \mathbb{E}[g(S)] \right| \leq \frac{\delta}{2}.  
    \end{equation}

    We are interested in $\mathbb{P}\left( \left| f(S) - \mathbb{E}\left( f(S)\right)\right| \geq \delta\right)$ and now begin to analyze this.

    $$\mathbb{P}\left( \left| f(S) - \mathbb{E}\left( f(S)\right)\right| \geq \delta\right) \leq \mathbb{P}(S \not \in \mathcal{Y}) + \mathbb{P}\left( \left| f(S) - \mathbb{E}\left( f(S)\right)\right| \geq \delta ~\land~ S \in \mathcal{Y}\right)$$
    $$\leq q + \mathbb{P}\left( \left| g(S) - \mathbb{E}\left( f(S)\right)\right| \geq \delta\right).$$
    By \Cref{eq:combes-1}, this is:
    \begin{equation}\label{eq:combes-2}
        \leq q + \mathbb{P}\left( \left| g(S) - \mathbb{E}\left( g(S)\right)\right| \geq \delta/2\right).
    \end{equation}

    (We now diverge from the proof of \cite{combes2024extension} slightly.) Applying \Cref{prop:mcdiarmid} to $(2/c) g$, which is 1-Lipschitz with respect to the Hamming distance on $\binom{[n]}{t}$, \Cref{eq:combes-2} is upper bounded by:
    $$\leq q + 2 \exp\left(- \frac{\delta^2}{8t c^2} \right).$$

    Combining all the steps above, we have found that
    $$\mathbb{P}\left( \left| f(S) - \mathbb{E}\left( f(S)\right)\right| \geq \delta\right) \leq q +  2 \exp\left(- \frac{\delta^2}{8 t c^2} \right).$$
\end{proof}

\subsubsection{Random permutations and sequences}
We now move to proving the concentration of permutation patterns and patterns in random permutations and sequences.

\begin{proposition}[Restatement of \Cref{prop:random-permutation-P-copies}]
    Consider a random permutation $\sigma: [n] \to [n]$ and any length-$k$ permutation pattern $\tau$. If $n-k\geq 2k^3$, then $\sigma$ has $(1 \pm \ep) \binom{n}{k} \frac{1}{k!}$ copies of $\tau$ with probability at least $1-\frac{2k^3}{(n-k)\ep^2}$. In particular, for fixed $k$, $\tau$, and $\ep$, this probability is $1-o_n(1)$ as $n\to\infty$.
\end{proposition}

\begin{proof}
    We analyze the expectation and variance of the number of copies of $\tau$ in $\sigma$ and apply Chebyshev's inequality.

    First, we can write the number of copies of $\tau$ in $\sigma$ as the following random variable $X$:
    $$X = \sum_{I \subseteq [n], |I| = k} 1(I \equiv \tau),$$
    where $I \equiv \tau$ means that there is a copy of $\tau$ on the indices of $I$.

    Second, $\mathbb{E}[X] = \sum_{I \subseteq [n], |I| = k} \mathbb{P}(I \equiv \tau) = \binom{n}{k} \frac{1}{k!}$.

    Third,
    \begin{align*}
        Var(X)      
        &=\sum_{a=1}^k\sum_{\substack{I,J\subseteq[n]\\|I|=|J|=k\\|I\cap J|=a}}\operatorname{Cov}\left(1(I\equiv\tau),1(J\equiv\tau)\right)\\
        &\leq\sum_{a=1}^k\sum_{\substack{I,J\subseteq[n]\\|I|=|J|=k\\|I\cap J|=a}}\mathbb{P}(I\equiv\tau\land J\equiv\tau).
    \end{align*}
    To bound each term, condition on $I \equiv \tau$, an event of probability $\frac{1}{k!}$ that fixes the relative order of the values on $I$. Given this, the $k-a$ values on $J\setminus I$ must occur in one prescribed relative order for $J\equiv\tau$. This prescribed relative order occurs with probability $1/(k-a)!$, and the additional constraints involving the shared values can only decrease the probability. Thus, $\mathbb{P}(I \equiv \tau \land J \equiv \tau) \leq \frac{1}{k!(k-a)!}$, and since there are $\binom{n}{k}\binom{k}{a}\binom{n-k}{k-a}$ pairs with $|I \cap J| = a$,
    $$Var(X) \leq \sum_{a = 1}^k \binom{n}{k} \binom{k}{a}\binom{n-k}{k-a} \cdot \frac{1}{k!(k-a)!}.$$

    We now apply Chebyshev's inequality:
    $$\mathbb{P}\left(X \not \in (1 \pm \ep) \binom{n}{k} \frac{1}{k!} \right) \leq \frac{Var(X)}{\ep^2 \binom{n}{k}^2 \frac{1}{(k!)^2}} \leq \frac{1}{\ep^2 \binom{n}{k} } \sum_{a = 1}^k \binom{k}{a}\binom{n-k}{k-a} \frac{k!}{(k-a)!}.$$

    We now analyze the expression on the right-hand side. Observe that $\binom{k}{a}\leq k^a$, $\frac{k!}{(k-a)!}\leq k^a$, and $\frac{\binom{n-k}{k-a}}{\binom{n}{k}}\leq\left(\frac{k}{n-k}\right)^a$. Therefore,
    $$\mathbb{P}\left(X \not \in (1 \pm \ep) \binom{n}{k} \frac{1}{k!} \right)\leq\frac{1}{\ep^2}\sum_{a=1}^k\left(\frac{k^3}{n-k}\right)^a.$$
    If $n-k\geq2k^3$, then this is at most
    $$\frac{2k^3}{(n-k)\ep^2}.$$

    Thus, 
    $$\mathbb{P}\left(X \not \in (1 \pm \ep) \binom{n}{k} \frac{1}{k!} \right) \leq \frac{2k^3}{(n-k)\ep^2}=o_n(1)$$
    for fixed $k$, $\tau$, and $\ep$.

\end{proof}

We state a similar result for sequences.

\begin{proposition}[Restatement of \Cref{prop:min-weight-sequence-P-copies}]
    Let $p$ be a distribution over $[m]$ such that the minimum probability of any element is $p_{\min}$. Consider any length-$k$ sequence $P = P_1 P_2 \dots P_k$ over $[m]$, and let $p_{i}$ be the weight that $p$ assigns $P_i$ for all $i \in [k]$.
    
    Let $\sigma \sim p^{\otimes n}$. For $n = \Omega\left(\frac{k^3}{\ep^2 p_{\min}} \right)$, $\sigma$ has $(1 \pm \ep) \binom{n}{k} \prod_{i \in [k]} p_i$ copies of $P$ with probability at least $1 - \frac{1}{k}$. Furthermore, for fixed $k$, $P$, $p$, and $\ep$, this probability is $1-o_n(1)$ as $n \to \infty$.
\end{proposition}

\begin{proof}
    We analyze the expectation and variance of the number of copies of $P$ in $\sigma \sim p^{\otimes n}$.

    First, we can write the number of copies of $P$ in $\sigma$ as the following random variable $X$:
    $$X = \sum_{I \subseteq [n], |I| = k} 1(I \equiv P),$$
    where $I \equiv P$ means that there is a copy of $P$ on the indices of $I$.

    Second, $\mathbb{E}[X] = \sum_{I \subseteq [n], |I| = k} \mathbb{P}(I \equiv P) = \binom{n}{k} \prod_{i \in [k]} p_i$.

    Third,
    \begin{align*}
        Var(X)
        &= \sum_{a = 1}^k \sum_{\substack{I, J \subseteq [n] \\ |I| = |J| = k \\ |I \cap J| = a}} \operatorname{Cov}\left(1(I \equiv P),1(J \equiv P)\right)\\
        &\leq \sum_{a = 1}^k \sum_{\substack{I, J \subseteq [n] \\ |I| = |J| = k \\ |I \cap J| = a}} \mathbb{P}(I \equiv P \land J \equiv P).
    \end{align*}
    For a fixed pair with $|I \cap J| = a$, realizing both copies pins the value at each of the $2k - a$ distinct positions of $I \cup J$. Since $\sigma$ is i.i.d.\ these positions are independent, and each of the $a$ shared positions is constrained to a single value rather than two (or the probability is $0$ if the two copies disagree there). Thus $\mathbb{P}(I \equiv P \land J \equiv P) \leq \left( \prod_{i \in [k]} p_i\right)^2 \cdot \frac{1}{p_{\min}^a}$, and since there are $\binom{n}{k}\binom{k}{a}\binom{n-k}{k-a}$ pairs with $|I \cap J| = a$,
    $$Var(X) \leq \sum_{a = 1}^k \binom{n}{k} \binom{k}{a}\binom{n-k}{k-a} \cdot \left( \prod_{i \in [k]} p_i\right)^2 \cdot \frac{1}{p_{\min}^a}.$$

    We now apply Chebyshev's inequality:
    $$\mathbb{P}\left(X \not \in (1 \pm \ep) \binom{n}{k} \prod_{i \in [k]} p_i \right) \leq \frac{Var(X)}{\ep^2 \binom{n}{k}^2 \left(\prod_{i \in [k]} p_i\right)^2} \leq \frac{\sum_{a = 1}^k \binom{n}{k} \binom{k}{a}\binom{n-k}{k-a} \cdot \left( \prod_{i \in [k]} p_i\right)^2 \cdot \frac{1}{p_{\min}^a}}{\ep^2 \binom{n}{k}^2 \left(\prod_{i \in [k]} p_i\right)^2}$$
    \begin{equation}\label{eq:prelims}
        = \frac{\sum_{a = 1}^k \binom{k}{a}\binom{n-k}{k-a} \frac{1}{p_{\min}^a}}{\ep^2 \binom{n}{k}}.
    \end{equation}

    We analyze this term. First, $\binom{k}{a}\frac{\binom{n-k}{k-a}}{\binom{n}{k}} \leq k^a \left(\frac{k}{n-k} \right)^a = \left(\frac{k^2}{n-k}\right)^a$, using $\binom{k}{a} \leq k^a$ and $\frac{\binom{n-k}{k-a}}{\binom{n}{k}} \leq \frac{\binom{n}{k-a}}{\binom{n}{k}} \leq \left(\frac{k}{n-k}\right)^a$. Next, assume $n - k \geq \frac{2k^3}{\ep^2 p_{\min}}$. Then \Cref{eq:prelims} is at most:
    $$\leq \frac{1}{\ep^2} \sum_{a = 1}^k \left(\frac{k^2}{(n-k) p_{\min}} \right)^a \leq \frac{1}{\ep^2} \sum_{a = 1}^k \left(\frac{k^2}{\frac{2k^3}{\ep^2 p_{\min}} \cdot p_{\min}} \right)^a = \frac{1}{\ep^2} \sum_{a = 1}^k \left(\frac{\ep^2}{2k}\right)^a \leq \frac{1}{k}.$$

    Thus, we find that
    $$\mathbb{P}\left(X \not \in (1 \pm \ep) \binom{n}{k} \prod_{i \in [k]} p_i \right) \leq \frac{1}{k}.$$

    Furthermore, for fixed $k$, $P$, $p$, and $\ep$, the bound in \Cref{eq:prelims} tends to zero as $n \to \infty$. To see this, for all sufficiently large $n$, $\frac{k^2}{(n-k)p_{\min}} \leq \frac{1}{2}$, and therefore
    \begin{align*}
        \mathbb{P}\left(X \not \in (1 \pm \ep) \binom{n}{k} \prod_{i \in [k]} p_i \right)
        &\leq \frac{1}{\ep^2} \sum_{a=1}^k \left(\frac{k^2}{(n-k)p_{\min}}\right)^a\\
        &\leq \frac{2k^2}{\ep^2(n-k)p_{\min}}\\
        &=o_n(1).
    \end{align*}
    Thus, for fixed $k$, $P$, $p$, and $\ep$, $\sigma$ has $(1 \pm \ep)\binom{n}{k}\prod_{i \in [k]}p_i$ copies of $P$ with probability $1-o_n(1)$.
\end{proof}

\subsubsection{Aggregate subpattern counts}\label{appendix:aggregate-counts}

We first prove the concentration bound for aggregate counts of permutation patterns. The proof follows the same second-moment calculation as for \Cref{prop:random-permutation-P-copies}. Here we include the additional step needed to handle the aggregate count.

Recall the definition of $\textsc{Count}_{\ell}(\tau,\sigma)$: $$
\textsc{Count}_{\ell}(\tau,\sigma)
=
\sum_{\substack{Q \subseteq \tau\\ |Q|=\ell}}
\sum_{\substack{R \subseteq \sigma\\ |R|=\ell}}
1(R\equiv Q),
$$
where subpatterns $Q\subseteq\tau$ are counted with multiplicity and $R\equiv Q$ means that $R$ forms a copy of $Q$.

\begin{proposition}\label{prop:random-permutation-aggregate-copies}
    Let $\sigma$ be a uniformly random permutation over $[N]$, let $\tau$ be a length-$k$ permutation pattern, and fix $\ell \in [k]$. If $N-\ell \geq 2\ell^3$, then
    $$\mathbb{P}\left(\textsc{Count}_{\ell}(\tau,\sigma) \not \in (1\pm\ep)\binom{k}{\ell}\binom{N}{\ell}\frac{1}{\ell!}\right) \leq \frac{2\ell^3}{(N-\ell)\ep^2}.$$
    Thus if $N-k\geq2k^3$, then with probability at least $1-O\left(\frac{k^4}{N\ep^2}\right)$, $\textsc{Count}_{\ell}(\tau,\sigma) \in (1\pm\ep)\binom{k}{\ell}\binom{N}{\ell}\frac{1}{\ell!}$ holds simultaneously for every $\ell \in [k]$.
\end{proposition}

\begin{proof}
    Let $X_\ell:=\textsc{Count}_{\ell}(\tau,\sigma)$. By linearity of expectation,
    $$\mathbb{E}[X_\ell]=\binom{k}{\ell}\binom{N}{\ell}\frac{1}{\ell!}.$$
    Consider two sets $A,B\in\binom{[N]}{\ell}$ with $|A\cap B|=a$ and two sets $I,J\in\binom{[k]}{\ell}$. The probability that $\sigma|_A$ forms $\tau|_I$ is $1/\ell!$. Conditioned on this event, the $\ell-a$ entries of $B\setminus A$ must occur in one determined relative order for $\sigma|_B$ to form a copy of $\tau|_J$. This order occurs with probability $1/(\ell-a)!$, and the additional constraints involving the shared entries can only decrease the probability. Therefore,
    $$\mathbb{P}\left(\sigma|_A\equiv\tau|_I ~~ \land ~~ \sigma|_B\equiv\tau|_J\right)\leq\frac{1}{\ell!(\ell-a)!}.$$
    Since indicators on disjoint sets are independent, we get
    $$Var(X_\ell)\leq\binom{k}{\ell}^2\binom{N}{\ell}\sum_{a=1}^{\ell}\binom{\ell}{a}\binom{N-\ell}{\ell-a}\frac{1}{\ell!(\ell-a)!}.$$
    By Chebyshev's inequality,
    \begin{align*}
        \mathbb{P}\left(\left|X_\ell-\mathbb{E}[X_\ell]\right|>\ep\mathbb{E}[X_\ell]\right)
        &\leq\frac{1}{\ep^2}\sum_{a=1}^{\ell}\binom{\ell}{a}\frac{\binom{N-\ell}{\ell-a}}{\binom{N}{\ell}}\frac{\ell!}{(\ell-a)!}\\
        &\leq\frac{1}{\ep^2}\sum_{a=1}^{\ell}\left(\frac{\ell^3}{N-\ell}\right)^a\\
        &\leq\frac{2\ell^3}{(N-\ell)\ep^2}.
    \end{align*}
    The proposition then follows by a union bound over $\ell\in[k]$.
\end{proof}

We next prove the analogous aggregate concentration bound for patterns in sequences.

\begin{proposition}\label{prop:min-weight-sequence-aggregate-copies}
    Let $p$ be a distribution over $[m]$ with minimum weight $p_{\min}$, let $\sigma\sim p^{\otimes N}$, let $P$ be a length-$k$ pattern, and fix $\ell\in[k]$. Let
    $$W_\ell:=\sum_{\substack{I\subseteq[k]\\|I|=\ell}}\prod_{i\in I}p_{P_i}.$$
    If $N-\ell\geq 2\ell^2/p_{\min}$, then
    $$\mathbb{P}\left(\textsc{Count}_{\ell}(P,\sigma)\notin(1\pm\ep)\binom{N}{\ell}W_\ell\right)\leq\frac{2\ell^2}{(N-\ell)p_{\min}\ep^2}.$$
    Thus if $N-k\geq2k^2/p_{\min}$, then with probability at least $1-O\left(\frac{k^3}{Np_{\min}\ep^2}\right)$, $\textsc{Count}_{\ell}(P,\sigma)\in(1\pm\ep)\binom{N}{\ell}W_\ell$ holds simultaneously for every $\ell\in[k]$.
\end{proposition}

\begin{proof}
    Let $X_\ell:=\textsc{Count}_{\ell}(P,\sigma)$. By linearity of expectation,
    $$\mathbb{E}[X_\ell]=\binom{N}{\ell}W_\ell.$$
    Consider two sets $A,B\in\binom{[N]}{\ell}$ with $|A\cap B|=a$ and two sets $I,J\in\binom{[k]}{\ell}$. If the requirements $\sigma|_A=P|_I$ and $\sigma|_B=P|_J$ disagree on a shared position, their joint probability is zero. Otherwise, relative to the product of the two individual probabilities, the joint probability counts the probability of each shared symbol only once instead of twice. Since every symbol has probability at least $p_{\min}$,
    $$\mathbb{P}\left(\sigma|_A=P|_I ~~ \land ~~ \sigma|_B=P|_J\right)\leq\frac{1}{p_{\min}^a}\left(\prod_{i\in I}p_{P_i}\right)\left(\prod_{j\in J}p_{P_j}\right).$$
    Since indicators on disjoint sets are independent, we have:
    $$Var(X_\ell)\leq\binom{N}{\ell}W_\ell^2\sum_{a=1}^{\ell}\binom{\ell}{a}\binom{N-\ell}{\ell-a}\frac{1}{p_{\min}^a}.$$
    By Chebyshev's inequality,
    \begin{align*}
        \mathbb{P}\left(\left|X_\ell-\mathbb{E}[X_\ell]\right|>\ep\mathbb{E}[X_\ell]\right)
        &\leq\frac{1}{\ep^2}\sum_{a=1}^{\ell}\binom{\ell}{a}\frac{\binom{N-\ell}{\ell-a}}{\binom{N}{\ell}}\frac{1}{p_{\min}^a}\\
        &\leq\frac{1}{\ep^2}\sum_{a=1}^{\ell}\left(\frac{\ell^2}{(N-\ell)p_{\min}}\right)^a\\
        &\leq\frac{2\ell^2}{(N-\ell)p_{\min}\ep^2}.
    \end{align*}
    The proposition then follows by a union bound over $\ell\in[k]$.
\end{proof}

\end{document}

%% file: Visuals/lowerbound-worstcase-visual.tex
\tikzset{every picture/.style={line width=0.75pt}} 

\begin{tikzpicture}[x=0.75pt,y=0.75pt,yscale=-1,xscale=1]

\draw   (101,495.5) -- (185.5,495.5) -- (185.5,526) -- (101,526) -- cycle ;
\draw  [dash pattern={on 0.84pt off 2.51pt}] (101,495.5) -- (111.5,495.5) -- (111.5,526) -- (101,526) -- cycle ;
\draw  [dash pattern={on 0.84pt off 2.51pt}] (143,495.5) -- (153.5,495.5) -- (153.5,526) -- (143,526) -- cycle ;
\draw  [dash pattern={on 0.84pt off 2.51pt}] (132.5,495.5) -- (143,495.5) -- (143,526) -- (132.5,526) -- cycle ;
\draw  [dash pattern={on 0.84pt off 2.51pt}] (122,495.5) -- (132.5,495.5) -- (132.5,526) -- (122,526) -- cycle ;
\draw  [dash pattern={on 0.84pt off 2.51pt}] (111.5,495.5) -- (122,495.5) -- (122,526) -- (111.5,526) -- cycle ;
\draw  [dash pattern={on 0.84pt off 2.51pt}] (153.5,495.5) -- (164,495.5) -- (164,526) -- (153.5,526) -- cycle ;
\draw  [dash pattern={on 0.84pt off 2.51pt}] (164,495.5) -- (174.5,495.5) -- (174.5,526) -- (164,526) -- cycle ;
\draw  [dash pattern={on 0.84pt off 2.51pt}] (174.5,495.5) -- (185,495.5) -- (185,526) -- (174.5,526) -- cycle ;

\draw    (90.5,8) -- (91.5,477.5) ;
\draw    (91.5,477.5) -- (554.5,477) ;
\draw   (439,495.5) -- (523.5,495.5) -- (523.5,526) -- (439,526) -- cycle ;
\draw  [dash pattern={on 0.84pt off 2.51pt}] (439,495.5) -- (449.5,495.5) -- (449.5,526) -- (439,526) -- cycle ;
\draw  [dash pattern={on 0.84pt off 2.51pt}] (481,495.5) -- (491.5,495.5) -- (491.5,526) -- (481,526) -- cycle ;
\draw  [dash pattern={on 0.84pt off 2.51pt}] (470.5,495.5) -- (481,495.5) -- (481,526) -- (470.5,526) -- cycle ;
\draw  [dash pattern={on 0.84pt off 2.51pt}] (460,495.5) -- (470.5,495.5) -- (470.5,526) -- (460,526) -- cycle ;
\draw  [dash pattern={on 0.84pt off 2.51pt}] (449.5,495.5) -- (460,495.5) -- (460,526) -- (449.5,526) -- cycle ;
\draw  [dash pattern={on 0.84pt off 2.51pt}] (491.5,495.5) -- (502,495.5) -- (502,526) -- (491.5,526) -- cycle ;
\draw  [dash pattern={on 0.84pt off 2.51pt}] (502,495.5) -- (512.5,495.5) -- (512.5,526) -- (502,526) -- cycle ;
\draw  [dash pattern={on 0.84pt off 2.51pt}] (512.5,495.5) -- (523,495.5) -- (523,526) -- (512.5,526) -- cycle ;

\draw   (354,495.5) -- (438.5,495.5) -- (438.5,526) -- (354,526) -- cycle ;
\draw  [dash pattern={on 0.84pt off 2.51pt}] (354,495.5) -- (364.5,495.5) -- (364.5,526) -- (354,526) -- cycle ;
\draw  [dash pattern={on 0.84pt off 2.51pt}] (396,495.5) -- (406.5,495.5) -- (406.5,526) -- (396,526) -- cycle ;
\draw  [dash pattern={on 0.84pt off 2.51pt}] (385.5,495.5) -- (396,495.5) -- (396,526) -- (385.5,526) -- cycle ;
\draw  [dash pattern={on 0.84pt off 2.51pt}] (375,495.5) -- (385.5,495.5) -- (385.5,526) -- (375,526) -- cycle ;
\draw  [dash pattern={on 0.84pt off 2.51pt}] (364.5,495.5) -- (375,495.5) -- (375,526) -- (364.5,526) -- cycle ;
\draw  [dash pattern={on 0.84pt off 2.51pt}] (406.5,495.5) -- (417,495.5) -- (417,526) -- (406.5,526) -- cycle ;
\draw  [dash pattern={on 0.84pt off 2.51pt}] (417,495.5) -- (427.5,495.5) -- (427.5,526) -- (417,526) -- cycle ;
\draw  [dash pattern={on 0.84pt off 2.51pt}] (427.5,495.5) -- (438,495.5) -- (438,526) -- (427.5,526) -- cycle ;

\draw   (270,495.5) -- (354.5,495.5) -- (354.5,526) -- (270,526) -- cycle ;
\draw  [dash pattern={on 0.84pt off 2.51pt}] (270,495.5) -- (280.5,495.5) -- (280.5,526) -- (270,526) -- cycle ;
\draw  [dash pattern={on 0.84pt off 2.51pt}] (312,495.5) -- (322.5,495.5) -- (322.5,526) -- (312,526) -- cycle ;
\draw  [dash pattern={on 0.84pt off 2.51pt}] (301.5,495.5) -- (312,495.5) -- (312,526) -- (301.5,526) -- cycle ;
\draw  [dash pattern={on 0.84pt off 2.51pt}] (291,495.5) -- (301.5,495.5) -- (301.5,526) -- (291,526) -- cycle ;
\draw  [dash pattern={on 0.84pt off 2.51pt}] (280.5,495.5) -- (291,495.5) -- (291,526) -- (280.5,526) -- cycle ;
\draw  [dash pattern={on 0.84pt off 2.51pt}] (322.5,495.5) -- (333,495.5) -- (333,526) -- (322.5,526) -- cycle ;
\draw  [dash pattern={on 0.84pt off 2.51pt}] (333,495.5) -- (343.5,495.5) -- (343.5,526) -- (333,526) -- cycle ;
\draw  [dash pattern={on 0.84pt off 2.51pt}] (343.5,495.5) -- (354,495.5) -- (354,526) -- (343.5,526) -- cycle ;

\draw   (186,495.5) -- (270.5,495.5) -- (270.5,526) -- (186,526) -- cycle ;
\draw  [dash pattern={on 0.84pt off 2.51pt}] (186,495.5) -- (196.5,495.5) -- (196.5,526) -- (186,526) -- cycle ;
\draw  [dash pattern={on 0.84pt off 2.51pt}] (228,495.5) -- (238.5,495.5) -- (238.5,526) -- (228,526) -- cycle ;
\draw  [dash pattern={on 0.84pt off 2.51pt}] (217.5,495.5) -- (228,495.5) -- (228,526) -- (217.5,526) -- cycle ;
\draw  [dash pattern={on 0.84pt off 2.51pt}] (207,495.5) -- (217.5,495.5) -- (217.5,526) -- (207,526) -- cycle ;
\draw  [dash pattern={on 0.84pt off 2.51pt}] (196.5,495.5) -- (207,495.5) -- (207,526) -- (196.5,526) -- cycle ;
\draw  [dash pattern={on 0.84pt off 2.51pt}] (238.5,495.5) -- (249,495.5) -- (249,526) -- (238.5,526) -- cycle ;
\draw  [dash pattern={on 0.84pt off 2.51pt}] (249,495.5) -- (259.5,495.5) -- (259.5,526) -- (249,526) -- cycle ;
\draw  [dash pattern={on 0.84pt off 2.51pt}] (259.5,495.5) -- (270,495.5) -- (270,526) -- (259.5,526) -- cycle ;

\draw  [color={rgb, 255:red, 155; green, 155; blue, 155 }  ,draw opacity=1 ][fill={rgb, 255:red, 155; green, 155; blue, 155 }  ,fill opacity=1 ] (95.61,393.16) .. controls (95.61,392.14) and (96.44,391.32) .. (97.47,391.32) .. controls (98.49,391.32) and (99.32,392.14) .. (99.32,393.16) .. controls (99.32,394.18) and (98.49,395.01) .. (97.47,395.01) .. controls (96.44,395.01) and (95.61,394.18) .. (95.61,393.16) -- cycle ;
\draw  [color={rgb, 255:red, 155; green, 155; blue, 155 }  ,draw opacity=1 ][fill={rgb, 255:red, 155; green, 155; blue, 155 }  ,fill opacity=1 ] (105.61,397.16) .. controls (105.61,396.14) and (106.44,395.32) .. (107.47,395.32) .. controls (108.49,395.32) and (109.32,396.14) .. (109.32,397.16) .. controls (109.32,398.18) and (108.49,399.01) .. (107.47,399.01) .. controls (106.44,399.01) and (105.61,398.18) .. (105.61,397.16) -- cycle ;
\draw  [color={rgb, 255:red, 67; green, 68; blue, 242 }  ,draw opacity=1 ][fill={rgb, 255:red, 67; green, 68; blue, 242 }  ,fill opacity=1 ] (116.61,401.16) .. controls (116.61,400.14) and (117.44,399.32) .. (118.47,399.32) .. controls (119.49,399.32) and (120.32,400.14) .. (120.32,401.16) .. controls (120.32,402.18) and (119.49,403.01) .. (118.47,403.01) .. controls (117.44,403.01) and (116.61,402.18) .. (116.61,401.16) -- cycle ;
\draw  [color={rgb, 255:red, 67; green, 68; blue, 242 }  ,draw opacity=1 ][fill={rgb, 255:red, 67; green, 68; blue, 242 }  ,fill opacity=1 ] (126.61,404.16) .. controls (126.61,403.14) and (127.44,402.32) .. (128.47,402.32) .. controls (129.49,402.32) and (130.32,403.14) .. (130.32,404.16) .. controls (130.32,405.18) and (129.49,406.01) .. (128.47,406.01) .. controls (127.44,406.01) and (126.61,405.18) .. (126.61,404.16) -- cycle ;
\draw  [fill={rgb, 255:red, 0; green, 0; blue, 0 }  ,fill opacity=1 ] (136.61,408.16) .. controls (136.61,407.14) and (137.44,406.32) .. (138.47,406.32) .. controls (139.49,406.32) and (140.32,407.14) .. (140.32,408.16) .. controls (140.32,409.18) and (139.49,410.01) .. (138.47,410.01) .. controls (137.44,410.01) and (136.61,409.18) .. (136.61,408.16) -- cycle ;
\draw  [fill={rgb, 255:red, 0; green, 0; blue, 0 }  ,fill opacity=1 ] (146.61,412.16) .. controls (146.61,411.14) and (147.44,410.32) .. (148.47,410.32) .. controls (149.49,410.32) and (150.32,411.14) .. (150.32,412.16) .. controls (150.32,413.18) and (149.49,414.01) .. (148.47,414.01) .. controls (147.44,414.01) and (146.61,413.18) .. (146.61,412.16) -- cycle ;
\draw  [fill={rgb, 255:red, 0; green, 0; blue, 0 }  ,fill opacity=1 ] (156.61,416.16) .. controls (156.61,415.14) and (157.44,414.32) .. (158.47,414.32) .. controls (159.49,414.32) and (160.32,415.14) .. (160.32,416.16) .. controls (160.32,417.18) and (159.49,418.01) .. (158.47,418.01) .. controls (157.44,418.01) and (156.61,417.18) .. (156.61,416.16) -- cycle ;
\draw  [fill={rgb, 255:red, 0; green, 0; blue, 0 }  ,fill opacity=1 ] (166.61,420.32) .. controls (166.61,419.3) and (167.44,418.47) .. (168.47,418.47) .. controls (169.49,418.47) and (170.32,419.3) .. (170.32,420.32) .. controls (170.32,421.34) and (169.49,422.16) .. (168.47,422.16) .. controls (167.44,422.16) and (166.61,421.34) .. (166.61,420.32) -- cycle ;
\draw  [color={rgb, 255:red, 67; green, 68; blue, 242 }  ,draw opacity=1 ][fill={rgb, 255:red, 67; green, 68; blue, 242 }  ,fill opacity=1 ] (181.61,298.16) .. controls (181.61,297.14) and (182.44,296.32) .. (183.47,296.32) .. controls (184.49,296.32) and (185.32,297.14) .. (185.32,298.16) .. controls (185.32,299.18) and (184.49,300.01) .. (183.47,300.01) .. controls (182.44,300.01) and (181.61,299.18) .. (181.61,298.16) -- cycle ;
\draw  [color={rgb, 255:red, 67; green, 68; blue, 242 }  ,draw opacity=1 ][fill={rgb, 255:red, 67; green, 68; blue, 242 }  ,fill opacity=1 ] (191.61,302.16) .. controls (191.61,301.14) and (192.44,300.32) .. (193.47,300.32) .. controls (194.49,300.32) and (195.32,301.14) .. (195.32,302.16) .. controls (195.32,303.18) and (194.49,304.01) .. (193.47,304.01) .. controls (192.44,304.01) and (191.61,303.18) .. (191.61,302.16) -- cycle ;
\draw  [color={rgb, 255:red, 155; green, 155; blue, 155 }  ,draw opacity=1 ][fill={rgb, 255:red, 155; green, 155; blue, 155 }  ,fill opacity=1 ] (202.61,306.16) .. controls (202.61,305.14) and (203.44,304.32) .. (204.47,304.32) .. controls (205.49,304.32) and (206.32,305.14) .. (206.32,306.16) .. controls (206.32,307.18) and (205.49,308.01) .. (204.47,308.01) .. controls (203.44,308.01) and (202.61,307.18) .. (202.61,306.16) -- cycle ;
\draw  [color={rgb, 255:red, 67; green, 68; blue, 242 }  ,draw opacity=1 ][fill={rgb, 255:red, 67; green, 68; blue, 242 }  ,fill opacity=1 ] (212.61,309.16) .. controls (212.61,308.14) and (213.44,307.32) .. (214.47,307.32) .. controls (215.49,307.32) and (216.32,308.14) .. (216.32,309.16) .. controls (216.32,310.18) and (215.49,311.01) .. (214.47,311.01) .. controls (213.44,311.01) and (212.61,310.18) .. (212.61,309.16) -- cycle ;
\draw  [fill={rgb, 255:red, 0; green, 0; blue, 0 }  ,fill opacity=1 ] (222.61,313.16) .. controls (222.61,312.14) and (223.44,311.32) .. (224.47,311.32) .. controls (225.49,311.32) and (226.32,312.14) .. (226.32,313.16) .. controls (226.32,314.18) and (225.49,315.01) .. (224.47,315.01) .. controls (223.44,315.01) and (222.61,314.18) .. (222.61,313.16) -- cycle ;
\draw  [fill={rgb, 255:red, 0; green, 0; blue, 0 }  ,fill opacity=1 ] (232.61,317.16) .. controls (232.61,316.14) and (233.44,315.32) .. (234.47,315.32) .. controls (235.49,315.32) and (236.32,316.14) .. (236.32,317.16) .. controls (236.32,318.18) and (235.49,319.01) .. (234.47,319.01) .. controls (233.44,319.01) and (232.61,318.18) .. (232.61,317.16) -- cycle ;
\draw  [fill={rgb, 255:red, 0; green, 0; blue, 0 }  ,fill opacity=1 ] (242.61,321.16) .. controls (242.61,320.14) and (243.44,319.32) .. (244.47,319.32) .. controls (245.49,319.32) and (246.32,320.14) .. (246.32,321.16) .. controls (246.32,322.18) and (245.49,323.01) .. (244.47,323.01) .. controls (243.44,323.01) and (242.61,322.18) .. (242.61,321.16) -- cycle ;
\draw  [fill={rgb, 255:red, 0; green, 0; blue, 0 }  ,fill opacity=1 ] (252.61,325.32) .. controls (252.61,324.3) and (253.44,323.47) .. (254.47,323.47) .. controls (255.49,323.47) and (256.32,324.3) .. (256.32,325.32) .. controls (256.32,326.34) and (255.49,327.16) .. (254.47,327.16) .. controls (253.44,327.16) and (252.61,326.34) .. (252.61,325.32) -- cycle ;
\draw  [color={rgb, 255:red, 155; green, 155; blue, 155 }  ,draw opacity=1 ][fill={rgb, 255:red, 155; green, 155; blue, 155 }  ,fill opacity=1 ] (264.61,197.32) .. controls (264.61,196.3) and (265.44,195.47) .. (266.47,195.47) .. controls (267.49,195.47) and (268.32,196.3) .. (268.32,197.32) .. controls (268.32,198.34) and (267.49,199.16) .. (266.47,199.16) .. controls (265.44,199.16) and (264.61,198.34) .. (264.61,197.32) -- cycle ;
\draw  [color={rgb, 255:red, 67; green, 68; blue, 242 }  ,draw opacity=1 ][fill={rgb, 255:red, 67; green, 68; blue, 242 }  ,fill opacity=1 ] (274.61,203.16) .. controls (274.61,202.14) and (275.44,201.32) .. (276.47,201.32) .. controls (277.49,201.32) and (278.32,202.14) .. (278.32,203.16) .. controls (278.32,204.18) and (277.49,205.01) .. (276.47,205.01) .. controls (275.44,205.01) and (274.61,204.18) .. (274.61,203.16) -- cycle ;
\draw  [color={rgb, 255:red, 67; green, 68; blue, 242 }  ,draw opacity=1 ][fill={rgb, 255:red, 67; green, 68; blue, 242 }  ,fill opacity=1 ] (285.61,207.16) .. controls (285.61,206.14) and (286.44,205.32) .. (287.47,205.32) .. controls (288.49,205.32) and (289.32,206.14) .. (289.32,207.16) .. controls (289.32,208.18) and (288.49,209.01) .. (287.47,209.01) .. controls (286.44,209.01) and (285.61,208.18) .. (285.61,207.16) -- cycle ;
\draw  [color={rgb, 255:red, 67; green, 68; blue, 242 }  ,draw opacity=1 ][fill={rgb, 255:red, 67; green, 68; blue, 242 }  ,fill opacity=1 ] (295.61,210.16) .. controls (295.61,209.14) and (296.44,208.32) .. (297.47,208.32) .. controls (298.49,208.32) and (299.32,209.14) .. (299.32,210.16) .. controls (299.32,211.18) and (298.49,212.01) .. (297.47,212.01) .. controls (296.44,212.01) and (295.61,211.18) .. (295.61,210.16) -- cycle ;
\draw  [fill={rgb, 255:red, 0; green, 0; blue, 0 }  ,fill opacity=1 ] (305.61,214.16) .. controls (305.61,213.14) and (306.44,212.32) .. (307.47,212.32) .. controls (308.49,212.32) and (309.32,213.14) .. (309.32,214.16) .. controls (309.32,215.18) and (308.49,216.01) .. (307.47,216.01) .. controls (306.44,216.01) and (305.61,215.18) .. (305.61,214.16) -- cycle ;
\draw  [fill={rgb, 255:red, 0; green, 0; blue, 0 }  ,fill opacity=1 ] (315.61,218.16) .. controls (315.61,217.14) and (316.44,216.32) .. (317.47,216.32) .. controls (318.49,216.32) and (319.32,217.14) .. (319.32,218.16) .. controls (319.32,219.18) and (318.49,220.01) .. (317.47,220.01) .. controls (316.44,220.01) and (315.61,219.18) .. (315.61,218.16) -- cycle ;
\draw  [fill={rgb, 255:red, 0; green, 0; blue, 0 }  ,fill opacity=1 ] (325.61,222.16) .. controls (325.61,221.14) and (326.44,220.32) .. (327.47,220.32) .. controls (328.49,220.32) and (329.32,221.14) .. (329.32,222.16) .. controls (329.32,223.18) and (328.49,224.01) .. (327.47,224.01) .. controls (326.44,224.01) and (325.61,223.18) .. (325.61,222.16) -- cycle ;
\draw  [fill={rgb, 255:red, 0; green, 0; blue, 0 }  ,fill opacity=1 ] (335.61,226.32) .. controls (335.61,225.3) and (336.44,224.47) .. (337.47,224.47) .. controls (338.49,224.47) and (339.32,225.3) .. (339.32,226.32) .. controls (339.32,227.34) and (338.49,228.16) .. (337.47,228.16) .. controls (336.44,228.16) and (335.61,227.34) .. (335.61,226.32) -- cycle ;
\draw  [color={rgb, 255:red, 208; green, 2; blue, 27 }  ,draw opacity=1 ][fill={rgb, 255:red, 208; green, 2; blue, 27 }  ,fill opacity=1 ] (202.61,358.16) .. controls (202.61,357.14) and (203.44,356.32) .. (204.47,356.32) .. controls (205.49,356.32) and (206.32,357.14) .. (206.32,358.16) .. controls (206.32,359.18) and (205.49,360.01) .. (204.47,360.01) .. controls (203.44,360.01) and (202.61,359.18) .. (202.61,358.16) -- cycle ;
\draw  [color={rgb, 255:red, 208; green, 2; blue, 27 }  ,draw opacity=1 ][fill={rgb, 255:red, 208; green, 2; blue, 27 }  ,fill opacity=1 ] (264.76,254.01) .. controls (264.76,252.99) and (265.59,252.16) .. (266.61,252.16) .. controls (267.64,252.16) and (268.47,252.99) .. (268.47,254.01) .. controls (268.47,255.03) and (267.64,255.85) .. (266.61,255.85) .. controls (265.59,255.85) and (264.76,255.03) .. (264.76,254.01) -- cycle ;
\draw [color={rgb, 255:red, 155; green, 155; blue, 155 }  ,draw opacity=1 ]   (204.47,308.01) -- (204.47,354.32) ;
\draw [shift={(204.47,356.32)}, rotate = 270] [color={rgb, 255:red, 155; green, 155; blue, 155 }  ,draw opacity=1 ][line width=0.75]    (10.93,-4.9) .. controls (6.95,-2.3) and (3.31,-0.67) .. (0,0) .. controls (3.31,0.67) and (6.95,2.3) .. (10.93,4.9)   ;
\draw [color={rgb, 255:red, 155; green, 155; blue, 155 }  ,draw opacity=1 ]   (266.47,199.16) -- (266.61,250.16) ;
\draw [shift={(266.61,252.16)}, rotate = 269.84] [color={rgb, 255:red, 155; green, 155; blue, 155 }  ,draw opacity=1 ][line width=0.75]    (10.93,-4.9) .. controls (6.95,-2.3) and (3.31,-0.67) .. (0,0) .. controls (3.31,0.67) and (6.95,2.3) .. (10.93,4.9)   ;
\draw  [dash pattern={on 4.5pt off 4.5pt}]  (92.5,55.5) -- (547.5,56) ;
\draw  [color={rgb, 255:red, 67; green, 68; blue, 242 }  ,draw opacity=1 ][fill={rgb, 255:red, 67; green, 68; blue, 242 }  ,fill opacity=1 ] (348.61,114.16) .. controls (348.61,113.14) and (349.44,112.32) .. (350.47,112.32) .. controls (351.49,112.32) and (352.32,113.14) .. (352.32,114.16) .. controls (352.32,115.18) and (351.49,116.01) .. (350.47,116.01) .. controls (349.44,116.01) and (348.61,115.18) .. (348.61,114.16) -- cycle ;
\draw  [color={rgb, 255:red, 67; green, 68; blue, 242 }  ,draw opacity=1 ][fill={rgb, 255:red, 67; green, 68; blue, 242 }  ,fill opacity=1 ] (358.61,118.16) .. controls (358.61,117.14) and (359.44,116.32) .. (360.47,116.32) .. controls (361.49,116.32) and (362.32,117.14) .. (362.32,118.16) .. controls (362.32,119.18) and (361.49,120.01) .. (360.47,120.01) .. controls (359.44,120.01) and (358.61,119.18) .. (358.61,118.16) -- cycle ;
\draw  [color={rgb, 255:red, 67; green, 68; blue, 242 }  ,draw opacity=1 ][fill={rgb, 255:red, 67; green, 68; blue, 242 }  ,fill opacity=1 ] (369.61,122.16) .. controls (369.61,121.14) and (370.44,120.32) .. (371.47,120.32) .. controls (372.49,120.32) and (373.32,121.14) .. (373.32,122.16) .. controls (373.32,123.18) and (372.49,124.01) .. (371.47,124.01) .. controls (370.44,124.01) and (369.61,123.18) .. (369.61,122.16) -- cycle ;
\draw  [color={rgb, 255:red, 67; green, 68; blue, 242 }  ,draw opacity=1 ][fill={rgb, 255:red, 67; green, 68; blue, 242 }  ,fill opacity=1 ] (379.61,125.16) .. controls (379.61,124.14) and (380.44,123.32) .. (381.47,123.32) .. controls (382.49,123.32) and (383.32,124.14) .. (383.32,125.16) .. controls (383.32,126.18) and (382.49,127.01) .. (381.47,127.01) .. controls (380.44,127.01) and (379.61,126.18) .. (379.61,125.16) -- cycle ;
\draw  [fill={rgb, 255:red, 0; green, 0; blue, 0 }  ,fill opacity=1 ] (389.61,129.16) .. controls (389.61,128.14) and (390.44,127.32) .. (391.47,127.32) .. controls (392.49,127.32) and (393.32,128.14) .. (393.32,129.16) .. controls (393.32,130.18) and (392.49,131.01) .. (391.47,131.01) .. controls (390.44,131.01) and (389.61,130.18) .. (389.61,129.16) -- cycle ;
\draw  [fill={rgb, 255:red, 0; green, 0; blue, 0 }  ,fill opacity=1 ] (399.61,133.16) .. controls (399.61,132.14) and (400.44,131.32) .. (401.47,131.32) .. controls (402.49,131.32) and (403.32,132.14) .. (403.32,133.16) .. controls (403.32,134.18) and (402.49,135.01) .. (401.47,135.01) .. controls (400.44,135.01) and (399.61,134.18) .. (399.61,133.16) -- cycle ;
\draw  [fill={rgb, 255:red, 0; green, 0; blue, 0 }  ,fill opacity=1 ] (409.61,137.16) .. controls (409.61,136.14) and (410.44,135.32) .. (411.47,135.32) .. controls (412.49,135.32) and (413.32,136.14) .. (413.32,137.16) .. controls (413.32,138.18) and (412.49,139.01) .. (411.47,139.01) .. controls (410.44,139.01) and (409.61,138.18) .. (409.61,137.16) -- cycle ;
\draw  [fill={rgb, 255:red, 0; green, 0; blue, 0 }  ,fill opacity=1 ] (419.61,141.32) .. controls (419.61,140.3) and (420.44,139.47) .. (421.47,139.47) .. controls (422.49,139.47) and (423.32,140.3) .. (423.32,141.32) .. controls (423.32,142.34) and (422.49,143.16) .. (421.47,143.16) .. controls (420.44,143.16) and (419.61,142.34) .. (419.61,141.32) -- cycle ;

\draw  [color={rgb, 255:red, 67; green, 68; blue, 242 }  ,draw opacity=1 ][fill={rgb, 255:red, 67; green, 68; blue, 242 }  ,fill opacity=1 ] (434.61,14.16) .. controls (434.61,13.14) and (435.44,12.32) .. (436.47,12.32) .. controls (437.49,12.32) and (438.32,13.14) .. (438.32,14.16) .. controls (438.32,15.18) and (437.49,16.01) .. (436.47,16.01) .. controls (435.44,16.01) and (434.61,15.18) .. (434.61,14.16) -- cycle ;
\draw  [color={rgb, 255:red, 67; green, 68; blue, 242 }  ,draw opacity=1 ][fill={rgb, 255:red, 67; green, 68; blue, 242 }  ,fill opacity=1 ] (444.61,18.16) .. controls (444.61,17.14) and (445.44,16.32) .. (446.47,16.32) .. controls (447.49,16.32) and (448.32,17.14) .. (448.32,18.16) .. controls (448.32,19.18) and (447.49,20.01) .. (446.47,20.01) .. controls (445.44,20.01) and (444.61,19.18) .. (444.61,18.16) -- cycle ;
\draw  [color={rgb, 255:red, 155; green, 155; blue, 155 }  ,draw opacity=1 ][fill={rgb, 255:red, 155; green, 155; blue, 155 }  ,fill opacity=1 ] (455.61,22.16) .. controls (455.61,21.14) and (456.44,20.32) .. (457.47,20.32) .. controls (458.49,20.32) and (459.32,21.14) .. (459.32,22.16) .. controls (459.32,23.18) and (458.49,24.01) .. (457.47,24.01) .. controls (456.44,24.01) and (455.61,23.18) .. (455.61,22.16) -- cycle ;
\draw  [color={rgb, 255:red, 155; green, 155; blue, 155 }  ,draw opacity=1 ][fill={rgb, 255:red, 155; green, 155; blue, 155 }  ,fill opacity=1 ] (465.61,25.16) .. controls (465.61,24.14) and (466.44,23.32) .. (467.47,23.32) .. controls (468.49,23.32) and (469.32,24.14) .. (469.32,25.16) .. controls (469.32,26.18) and (468.49,27.01) .. (467.47,27.01) .. controls (466.44,27.01) and (465.61,26.18) .. (465.61,25.16) -- cycle ;
\draw  [fill={rgb, 255:red, 0; green, 0; blue, 0 }  ,fill opacity=1 ] (475.61,29.16) .. controls (475.61,28.14) and (476.44,27.32) .. (477.47,27.32) .. controls (478.49,27.32) and (479.32,28.14) .. (479.32,29.16) .. controls (479.32,30.18) and (478.49,31.01) .. (477.47,31.01) .. controls (476.44,31.01) and (475.61,30.18) .. (475.61,29.16) -- cycle ;
\draw  [fill={rgb, 255:red, 0; green, 0; blue, 0 }  ,fill opacity=1 ] (485.61,33.16) .. controls (485.61,32.14) and (486.44,31.32) .. (487.47,31.32) .. controls (488.49,31.32) and (489.32,32.14) .. (489.32,33.16) .. controls (489.32,34.18) and (488.49,35.01) .. (487.47,35.01) .. controls (486.44,35.01) and (485.61,34.18) .. (485.61,33.16) -- cycle ;
\draw  [fill={rgb, 255:red, 0; green, 0; blue, 0 }  ,fill opacity=1 ] (495.61,37.16) .. controls (495.61,36.14) and (496.44,35.32) .. (497.47,35.32) .. controls (498.49,35.32) and (499.32,36.14) .. (499.32,37.16) .. controls (499.32,38.18) and (498.49,39.01) .. (497.47,39.01) .. controls (496.44,39.01) and (495.61,38.18) .. (495.61,37.16) -- cycle ;
\draw  [fill={rgb, 255:red, 0; green, 0; blue, 0 }  ,fill opacity=1 ] (505.61,41.32) .. controls (505.61,40.3) and (506.44,39.47) .. (507.47,39.47) .. controls (508.49,39.47) and (509.32,40.3) .. (509.32,41.32) .. controls (509.32,42.34) and (508.49,43.16) .. (507.47,43.16) .. controls (506.44,43.16) and (505.61,42.34) .. (505.61,41.32) -- cycle ;
\draw [color={rgb, 255:red, 155; green, 155; blue, 155 }  ,draw opacity=1 ]   (457.47,22.16) -- (457.61,73.16) ;
\draw [shift={(457.61,75.16)}, rotate = 269.84] [color={rgb, 255:red, 155; green, 155; blue, 155 }  ,draw opacity=1 ][line width=0.75]    (10.93,-4.9) .. controls (6.95,-2.3) and (3.31,-0.67) .. (0,0) .. controls (3.31,0.67) and (6.95,2.3) .. (10.93,4.9)   ;
\draw [color={rgb, 255:red, 155; green, 155; blue, 155 }  ,draw opacity=1 ]   (467.47,25.16) -- (467.61,76.16) ;
\draw [shift={(467.61,78.16)}, rotate = 269.84] [color={rgb, 255:red, 155; green, 155; blue, 155 }  ,draw opacity=1 ][line width=0.75]    (10.93,-4.9) .. controls (6.95,-2.3) and (3.31,-0.67) .. (0,0) .. controls (3.31,0.67) and (6.95,2.3) .. (10.93,4.9)   ;
\draw  [color={rgb, 255:red, 208; green, 2; blue, 27 }  ,draw opacity=1 ][fill={rgb, 255:red, 208; green, 2; blue, 27 }  ,fill opacity=1 ] (455.76,77.01) .. controls (455.76,75.99) and (456.59,75.16) .. (457.61,75.16) .. controls (458.64,75.16) and (459.47,75.99) .. (459.47,77.01) .. controls (459.47,78.03) and (458.64,78.85) .. (457.61,78.85) .. controls (456.59,78.85) and (455.76,78.03) .. (455.76,77.01) -- cycle ;
\draw  [color={rgb, 255:red, 208; green, 2; blue, 27 }  ,draw opacity=1 ][fill={rgb, 255:red, 208; green, 2; blue, 27 }  ,fill opacity=1 ] (465.76,80.01) .. controls (465.76,78.99) and (466.59,78.16) .. (467.61,78.16) .. controls (468.64,78.16) and (469.47,78.99) .. (469.47,80.01) .. controls (469.47,81.03) and (468.64,81.85) .. (467.61,81.85) .. controls (466.59,81.85) and (465.76,81.03) .. (465.76,80.01) -- cycle ;
\draw  [dash pattern={on 4.5pt off 4.5pt}]  (92.5,101.5) -- (547.5,102) ;
\draw  [dash pattern={on 4.5pt off 4.5pt}]  (91.5,190.5) -- (546.5,191) ;
\draw  [dash pattern={on 4.5pt off 4.5pt}]  (92.5,237.5) -- (547.5,238) ;
\draw  [dash pattern={on 4.5pt off 4.5pt}]  (93.5,285.5) -- (548.5,286) ;
\draw  [dash pattern={on 4.5pt off 4.5pt}]  (91.5,334.5) -- (546.5,335) ;
\draw  [dash pattern={on 4.5pt off 4.5pt}]  (92.5,379.5) -- (547.5,380) ;
\draw  [dash pattern={on 4.5pt off 4.5pt}]  (93.5,427.5) -- (548.5,428) ;
\draw  [dash pattern={on 4.5pt off 4.5pt}]  (90.5,8) -- (545.5,8.5) ;
\draw  [dash pattern={on 4.5pt off 4.5pt}]  (94.5,147.5) -- (549.5,148) ;
\draw [color={rgb, 255:red, 155; green, 155; blue, 155 }  ,draw opacity=1 ]   (97.47,393.16) -- (97.61,444.16) ;
\draw [shift={(97.61,446.16)}, rotate = 269.84] [color={rgb, 255:red, 155; green, 155; blue, 155 }  ,draw opacity=1 ][line width=0.75]    (10.93,-4.9) .. controls (6.95,-2.3) and (3.31,-0.67) .. (0,0) .. controls (3.31,0.67) and (6.95,2.3) .. (10.93,4.9)   ;
\draw [color={rgb, 255:red, 155; green, 155; blue, 155 }  ,draw opacity=1 ]   (107.47,396.16) -- (107.61,447.16) ;
\draw [shift={(107.61,449.16)}, rotate = 269.84] [color={rgb, 255:red, 155; green, 155; blue, 155 }  ,draw opacity=1 ][line width=0.75]    (10.93,-4.9) .. controls (6.95,-2.3) and (3.31,-0.67) .. (0,0) .. controls (3.31,0.67) and (6.95,2.3) .. (10.93,4.9)   ;
\draw  [color={rgb, 255:red, 208; green, 2; blue, 27 }  ,draw opacity=1 ][fill={rgb, 255:red, 208; green, 2; blue, 27 }  ,fill opacity=1 ] (95.76,448.01) .. controls (95.76,446.99) and (96.59,446.16) .. (97.61,446.16) .. controls (98.64,446.16) and (99.47,446.99) .. (99.47,448.01) .. controls (99.47,449.03) and (98.64,449.85) .. (97.61,449.85) .. controls (96.59,449.85) and (95.76,449.03) .. (95.76,448.01) -- cycle ;
\draw  [color={rgb, 255:red, 208; green, 2; blue, 27 }  ,draw opacity=1 ][fill={rgb, 255:red, 208; green, 2; blue, 27 }  ,fill opacity=1 ] (105.76,451.01) .. controls (105.76,449.99) and (106.59,449.16) .. (107.61,449.16) .. controls (108.64,449.16) and (109.47,449.99) .. (109.47,451.01) .. controls (109.47,452.03) and (108.64,452.85) .. (107.61,452.85) .. controls (106.59,452.85) and (105.76,452.03) .. (105.76,451.01) -- cycle ;
\draw [color={rgb, 255:red, 67; green, 68; blue, 242 }  ,draw opacity=1 ][line width=1.5]    (95,369) -- (128.5,369) ;
\draw [color={rgb, 255:red, 67; green, 68; blue, 242 }  ,draw opacity=1 ][line width=1.5]    (95,369) -- (94.5,384.5) ;
\draw [color={rgb, 255:red, 67; green, 68; blue, 242 }  ,draw opacity=1 ][line width=1.5]    (128.5,369) -- (129.5,385.5) ;
\draw [color={rgb, 255:red, 0; green, 0; blue, 0 }  ,draw opacity=1 ][line width=1.5]    (140,368) -- (173.5,368) ;
\draw [color={rgb, 255:red, 67; green, 68; blue, 242 }  ,draw opacity=1 ][line width=1.5]    (140,368) -- (139.5,383.5) ;
\draw [color={rgb, 255:red, 0; green, 0; blue, 0 }  ,draw opacity=1 ][line width=1.5]    (173.5,368) -- (174.5,384.5) ;

\draw (136,532.4) node [anchor=north west][inner sep=0.75pt]    {$1$};
\draw (224,532.4) node [anchor=north west][inner sep=0.75pt]    {$2$};
\draw (462,529.4) node [anchor=north west][inner sep=0.75pt]    {$k-1$};
\draw (306,531.4) node [anchor=north west][inner sep=0.75pt]    {$\cdots $};
\draw (101,345.4) node [anchor=north west][inner sep=0.75pt]    {$\textcolor[rgb]{0.26,0.27,0.95}{A_{1}}$};
\draw (148,345.4) node [anchor=north west][inner sep=0.75pt]    {$B_{1}$};

\end{tikzpicture}

%% file: Visuals/subsequence-worst-case-lb-figure.tex
\tikzset{every picture/.style={line width=0.75pt}} 

\begin{tikzpicture}[x=0.75pt,y=0.75pt,yscale=-1,xscale=1]

\draw  [draw opacity=0][fill={rgb, 255:red, 248; green, 231; blue, 28 }  ,fill opacity=0.7 ] (560.2,76) -- (598.21,76) -- (598.21,118) -- (560.2,118) -- cycle ;
\draw  [line width=0.75]  (59.47,76) -- (179.41,76) -- (179.41,118) -- (59.47,118) -- cycle ;
\draw  [line width=0.75]  (179.41,76) -- (299.35,76) -- (299.35,118) -- (179.41,118) -- cycle ;
\draw  [line width=0.75]  (299.35,76) -- (419.29,76) -- (419.29,118) -- (299.35,118) -- cycle ;
\draw  [line width=0.75]  (419.29,76) -- (539.22,76) -- (539.22,118) -- (419.29,118) -- cycle ;
\draw  [line width=0.75]  (539.22,76) -- (638.52,76) -- (638.52,118) -- (539.22,118) -- cycle ;
\draw  [line width=0.75]  (80.45,76) -- (200.38,76) -- (200.38,118) -- (80.45,118) -- cycle ;
\draw  [line width=0.75]  (200.38,76) -- (320.32,76) -- (320.32,118) -- (200.38,118) -- cycle ;
\draw  [line width=0.75]  (320.32,76) -- (440.26,76) -- (440.26,118) -- (320.32,118) -- cycle ;
\draw  [line width=0.75]  (440.26,76) -- (560.2,76) -- (560.2,118) -- (440.26,118) -- cycle ;
\draw  [line width=0.75]  (560.2,76) -- (638.52,76) -- (638.52,118) -- (560.2,118) -- cycle ;
\draw  [line width=0.75]  (98.8,76) -- (218.73,76) -- (218.73,118) -- (98.8,118) -- cycle ;
\draw  [line width=0.75]  (218.73,76) -- (338.67,76) -- (338.67,118) -- (218.73,118) -- cycle ;
\draw  [line width=0.75]  (338.67,76) -- (458.61,76) -- (458.61,118) -- (338.67,118) -- cycle ;
\draw  [line width=0.75]  (458.61,76) -- (578.55,76) -- (578.55,118) -- (458.61,118) -- cycle ;
\draw  [line width=0.75]  (578.55,76) -- (638.5,76) -- (638.5,118) -- (578.55,118) -- cycle ;
\draw  [line width=0.75]  (118.46,76) -- (238.4,76) -- (238.4,118) -- (118.46,118) -- cycle ;
\draw  [line width=0.75]  (238.4,76) -- (358.33,76) -- (358.33,118) -- (238.4,118) -- cycle ;
\draw  [line width=0.75]  (358.33,76) -- (478.27,76) -- (478.27,118) -- (358.33,118) -- cycle ;
\draw  [line width=0.75]  (478.27,76) -- (598.21,76) -- (598.21,118) -- (478.27,118) -- cycle ;
\draw  [line width=0.75]  (598.21,76) -- (638.5,76) -- (638.5,118) -- (598.21,118) -- cycle ;
\draw  [line width=0.75]  (135.5,76) -- (255.44,76) -- (255.44,118) -- (135.5,118) -- cycle ;
\draw  [line width=0.75]  (255.44,76) -- (375.37,76) -- (375.37,118) -- (255.44,118) -- cycle ;
\draw  [line width=0.75]  (375.37,76) -- (495.31,76) -- (495.31,118) -- (375.37,118) -- cycle ;
\draw  [line width=0.75]  (495.31,76) -- (615.25,76) -- (615.25,118) -- (495.31,118) -- cycle ;
\draw  [line width=0.75]  (615.25,76) -- (638.5,76) -- (638.5,118) -- (615.25,118) -- cycle ;
\draw  [color={rgb, 255:red, 74; green, 144; blue, 226 }  ,draw opacity=1 ][line width=1.5]  (38.5,76) -- (158.44,76) -- (158.44,118) -- (38.5,118) -- cycle ;
\draw  [color={rgb, 255:red, 74; green, 144; blue, 226 }  ,draw opacity=1 ][line width=1.5]  (158.44,76) -- (278.38,76) -- (278.38,118) -- (158.44,118) -- cycle ;
\draw  [color={rgb, 255:red, 74; green, 144; blue, 226 }  ,draw opacity=1 ][line width=1.5]  (278.38,76) -- (398.31,76) -- (398.31,118) -- (278.38,118) -- cycle ;
\draw  [color={rgb, 255:red, 74; green, 144; blue, 226 }  ,draw opacity=1 ][line width=1.5]  (398.31,76) -- (518.25,76) -- (518.25,118) -- (398.31,118) -- cycle ;
\draw  [color={rgb, 255:red, 74; green, 144; blue, 226 }  ,draw opacity=1 ][line width=1.5]  (518.25,76) -- (638.19,76) -- (638.19,118) -- (518.25,118) -- cycle ;
\draw [line width=1.5]    (37,161) -- (157.5,161) ;
\draw [line width=1.5]    (37.25,150) -- (37.5,172) ;
\draw [line width=1.5]    (157.38,150) -- (157.63,172) ;
\draw [line width=1.5]    (518,158) -- (559.46,158) ;
\draw [line width=1.5]    (518.09,147) -- (518.17,169) ;
\draw [line width=1.5]    (559.41,147) -- (559.5,169) ;
\draw [line width=1.5]    (595.5,144) -- (580.19,120.51) ;
\draw [shift={(578.55,118)}, rotate = 56.9] [color={rgb, 255:red, 0; green, 0; blue, 0 }  ][line width=1.5]    (14.21,-4.28) .. controls (9.04,-1.82) and (4.3,-0.39) .. (0,0) .. controls (4.3,0.39) and (9.04,1.82) .. (14.21,4.28)   ;
\draw  [line width=0.75]  (60.47,13) -- (180.41,13) -- (180.41,55) -- (60.47,55) -- cycle ;
\draw  [line width=0.75]  (180.41,13) -- (300.35,13) -- (300.35,55) -- (180.41,55) -- cycle ;
\draw  [line width=0.75]  (300.35,13) -- (420.29,13) -- (420.29,55) -- (300.35,55) -- cycle ;
\draw  [line width=0.75]  (420.29,13) -- (540.22,13) -- (540.22,55) -- (420.29,55) -- cycle ;
\draw  [line width=0.75]  (540.22,13) -- (639.52,13) -- (639.52,55) -- (540.22,55) -- cycle ;
\draw  [line width=0.75]  (81.45,13) -- (201.38,13) -- (201.38,55) -- (81.45,55) -- cycle ;
\draw  [line width=0.75]  (201.38,13) -- (321.32,13) -- (321.32,55) -- (201.38,55) -- cycle ;
\draw  [line width=0.75]  (321.32,13) -- (441.26,13) -- (441.26,55) -- (321.32,55) -- cycle ;
\draw  [line width=0.75]  (441.26,13) -- (561.2,13) -- (561.2,55) -- (441.26,55) -- cycle ;
\draw  [line width=0.75]  (561.2,13) -- (639.52,13) -- (639.52,55) -- (561.2,55) -- cycle ;
\draw  [line width=0.75]  (99.8,13) -- (219.73,13) -- (219.73,55) -- (99.8,55) -- cycle ;
\draw  [line width=0.75]  (219.73,13) -- (339.67,13) -- (339.67,55) -- (219.73,55) -- cycle ;
\draw  [line width=0.75]  (339.67,13) -- (459.61,13) -- (459.61,55) -- (339.67,55) -- cycle ;
\draw  [line width=0.75]  (459.61,13) -- (579.55,13) -- (579.55,55) -- (459.61,55) -- cycle ;
\draw  [line width=0.75]  (579.55,13) -- (639.19,13) -- (639.19,55) -- (579.55,55) -- cycle ;
\draw  [line width=0.75]  (119.46,13) -- (239.4,13) -- (239.4,55) -- (119.46,55) -- cycle ;
\draw  [line width=0.75]  (239.4,13) -- (359.33,13) -- (359.33,55) -- (239.4,55) -- cycle ;
\draw  [line width=0.75]  (359.33,13) -- (479.27,13) -- (479.27,55) -- (359.33,55) -- cycle ;
\draw  [line width=0.75]  (479.27,13) -- (599.21,13) -- (599.21,55) -- (479.27,55) -- cycle ;
\draw  [line width=0.75]  (599.21,13) -- (639.5,13) -- (639.5,55) -- (599.21,55) -- cycle ;
\draw  [line width=0.75]  (136.5,13) -- (256.44,13) -- (256.44,55) -- (136.5,55) -- cycle ;
\draw  [line width=0.75]  (256.44,13) -- (376.37,13) -- (376.37,55) -- (256.44,55) -- cycle ;
\draw  [line width=0.75]  (376.37,13) -- (496.31,13) -- (496.31,55) -- (376.37,55) -- cycle ;
\draw  [line width=0.75]  (496.31,13) -- (616.25,13) -- (616.25,55) -- (496.31,55) -- cycle ;
\draw  [line width=0.75]  (616.25,13) -- (639.52,13) -- (639.52,55) -- (616.25,55) -- cycle ;
\draw  [color={rgb, 255:red, 74; green, 144; blue, 226 }  ,draw opacity=1 ][line width=1.5]  (39.5,13) -- (159.44,13) -- (159.44,55) -- (39.5,55) -- cycle ;
\draw  [color={rgb, 255:red, 74; green, 144; blue, 226 }  ,draw opacity=1 ][line width=1.5]  (159.44,13) -- (279.38,13) -- (279.38,55) -- (159.44,55) -- cycle ;
\draw  [color={rgb, 255:red, 74; green, 144; blue, 226 }  ,draw opacity=1 ][line width=1.5]  (279.38,13) -- (399.31,13) -- (399.31,55) -- (279.38,55) -- cycle ;
\draw  [color={rgb, 255:red, 74; green, 144; blue, 226 }  ,draw opacity=1 ][line width=1.5]  (399.31,13) -- (519.25,13) -- (519.25,55) -- (399.31,55) -- cycle ;
\draw  [color={rgb, 255:red, 74; green, 144; blue, 226 }  ,draw opacity=1 ][line width=1.5]  (519.25,13) -- (639.19,13) -- (639.19,55) -- (519.25,55) -- cycle ;

\draw (40,87.4) node [anchor=north west][inner sep=0.75pt]  [color={rgb, 255:red, 144; green, 19; blue, 254 }  ,opacity=1 ]  {$\tau _{1}$};
\draw (62,87.4) node [anchor=north west][inner sep=0.75pt]  [color={rgb, 255:red, 144; green, 19; blue, 254 }  ,opacity=1 ]  {$\tau _{1}$};
\draw (81,87.4) node [anchor=north west][inner sep=0.75pt]  [color={rgb, 255:red, 144; green, 19; blue, 254 }  ,opacity=1 ]  {$\tau _{1}$};
\draw (100,87.4) node [anchor=north west][inner sep=0.75pt]  [color={rgb, 255:red, 144; green, 19; blue, 254 }  ,opacity=1 ]  {$\tau _{1}$};
\draw (119,87.4) node [anchor=north west][inner sep=0.75pt]  [color={rgb, 255:red, 144; green, 19; blue, 254 }  ,opacity=1 ]  {$\tau _{1}$};
\draw (138,86.4) node [anchor=north west][inner sep=0.75pt]  [color={rgb, 255:red, 144; green, 19; blue, 254 }  ,opacity=1 ]  {$\tau _{1}$};
\draw (159,87.4) node [anchor=north west][inner sep=0.75pt]  [color={rgb, 255:red, 208; green, 2; blue, 27 }  ,opacity=1 ]  {$\tau _{2}$};
\draw (181,86.4) node [anchor=north west][inner sep=0.75pt]  [color={rgb, 255:red, 208; green, 2; blue, 27 }  ,opacity=1 ]  {$\tau _{2}$};
\draw (200,87.4) node [anchor=north west][inner sep=0.75pt]  [color={rgb, 255:red, 208; green, 2; blue, 27 }  ,opacity=1 ]  {$\tau _{2}$};
\draw (219,87.4) node [anchor=north west][inner sep=0.75pt]  [color={rgb, 255:red, 208; green, 2; blue, 27 }  ,opacity=1 ]  {$\tau _{2}$};
\draw (238,87.4) node [anchor=north west][inner sep=0.75pt]  [color={rgb, 255:red, 208; green, 2; blue, 27 }  ,opacity=1 ]  {$\tau _{2}$};
\draw (257,86.4) node [anchor=north west][inner sep=0.75pt]  [color={rgb, 255:red, 208; green, 2; blue, 27 }  ,opacity=1 ]  {$\tau _{2}$};
\draw (281,87.4) node [anchor=north west][inner sep=0.75pt]  [color={rgb, 255:red, 65; green, 117; blue, 5 }  ,opacity=1 ]  {$\tau _{3}$};
\draw (303,86.4) node [anchor=north west][inner sep=0.75pt]  [color={rgb, 255:red, 65; green, 117; blue, 5 }  ,opacity=1 ]  {$\tau _{3}$};
\draw (322,87.4) node [anchor=north west][inner sep=0.75pt]  [color={rgb, 255:red, 65; green, 117; blue, 5 }  ,opacity=1 ]  {$\tau _{3}$};
\draw (341,87.4) node [anchor=north west][inner sep=0.75pt]  [color={rgb, 255:red, 65; green, 117; blue, 5 }  ,opacity=1 ]  {$\tau _{3}$};
\draw (360,87.4) node [anchor=north west][inner sep=0.75pt]  [color={rgb, 255:red, 65; green, 117; blue, 5 }  ,opacity=1 ]  {$\tau _{3}$};
\draw (379,86.4) node [anchor=north west][inner sep=0.75pt]  [color={rgb, 255:red, 65; green, 117; blue, 5 }  ,opacity=1 ]  {$\tau _{3}$};
\draw (399,87.4) node [anchor=north west][inner sep=0.75pt]  [color={rgb, 255:red, 139; green, 87; blue, 42 }  ,opacity=1 ]  {$\tau _{4}$};
\draw (421,86.4) node [anchor=north west][inner sep=0.75pt]  [color={rgb, 255:red, 139; green, 87; blue, 42 }  ,opacity=1 ]  {$\tau _{4}$};
\draw (440,87.4) node [anchor=north west][inner sep=0.75pt]  [color={rgb, 255:red, 139; green, 87; blue, 42 }  ,opacity=1 ]  {$\tau _{4}$};
\draw (459,87.4) node [anchor=north west][inner sep=0.75pt]  [color={rgb, 255:red, 139; green, 87; blue, 42 }  ,opacity=1 ]  {$\tau _{4}$};
\draw (478,87.4) node [anchor=north west][inner sep=0.75pt]  [color={rgb, 255:red, 139; green, 87; blue, 42 }  ,opacity=1 ]  {$\tau _{4}$};
\draw (497,86.4) node [anchor=north west][inner sep=0.75pt]  [color={rgb, 255:red, 139; green, 87; blue, 42 }  ,opacity=1 ]  {$\tau _{4}$};
\draw (561,87.4) node [anchor=north west][inner sep=0.75pt]  [color={rgb, 255:red, 189; green, 16; blue, 224 }  ,opacity=1 ]  {$\tau _{6}$};
\draw (580,87.4) node [anchor=north west][inner sep=0.75pt]  [color={rgb, 255:red, 189; green, 16; blue, 224 }  ,opacity=1 ]  {$\tau _{6}$};
\draw (89,168.4) node [anchor=north west][inner sep=0.75pt]    {$\frac{n}{5}$};
\draw (520,88.4) node [anchor=north west][inner sep=0.75pt]  [color={rgb, 255:red, 16; green, 52; blue, 224 }  ,opacity=1 ]  {$\tau _{5}$};
\draw (541,89.4) node [anchor=north west][inner sep=0.75pt]  [color={rgb, 255:red, 16; green, 52; blue, 224 }  ,opacity=1 ]  {$\tau _{5}$};
\draw (599,87.4) node [anchor=north west][inner sep=0.75pt]  [color={rgb, 255:red, 16; green, 52; blue, 224 }  ,opacity=1 ]  {$\tau _{5}$};
\draw (618,88.4) node [anchor=north west][inner sep=0.75pt]  [color={rgb, 255:red, 16; green, 52; blue, 224 }  ,opacity=1 ]  {$\tau _{5}$};
\draw (529,164.4) node [anchor=north west][inner sep=0.75pt]    {$\frac{n}{5\ell }$};
\draw (582,146) node [anchor=north west][inner sep=0.75pt]   [align=left] {Random\\placement };
\draw (41,24.4) node [anchor=north west][inner sep=0.75pt]  [color={rgb, 255:red, 144; green, 19; blue, 254 }  ,opacity=1 ]  {$\tau _{1}$};
\draw (63,24.4) node [anchor=north west][inner sep=0.75pt]  [color={rgb, 255:red, 144; green, 19; blue, 254 }  ,opacity=1 ]  {$\tau _{1}$};
\draw (82,24.4) node [anchor=north west][inner sep=0.75pt]  [color={rgb, 255:red, 144; green, 19; blue, 254 }  ,opacity=1 ]  {$\tau _{1}$};
\draw (101,24.4) node [anchor=north west][inner sep=0.75pt]  [color={rgb, 255:red, 144; green, 19; blue, 254 }  ,opacity=1 ]  {$\tau _{1}$};
\draw (120,24.4) node [anchor=north west][inner sep=0.75pt]  [color={rgb, 255:red, 144; green, 19; blue, 254 }  ,opacity=1 ]  {$\tau _{1}$};
\draw (139,23.4) node [anchor=north west][inner sep=0.75pt]  [color={rgb, 255:red, 144; green, 19; blue, 254 }  ,opacity=1 ]  {$\tau _{1}$};
\draw (160,24.4) node [anchor=north west][inner sep=0.75pt]  [color={rgb, 255:red, 208; green, 2; blue, 27 }  ,opacity=1 ]  {$\tau _{2}$};
\draw (182,23.4) node [anchor=north west][inner sep=0.75pt]  [color={rgb, 255:red, 208; green, 2; blue, 27 }  ,opacity=1 ]  {$\tau _{2}$};
\draw (201,24.4) node [anchor=north west][inner sep=0.75pt]  [color={rgb, 255:red, 208; green, 2; blue, 27 }  ,opacity=1 ]  {$\tau _{2}$};
\draw (220,24.4) node [anchor=north west][inner sep=0.75pt]  [color={rgb, 255:red, 208; green, 2; blue, 27 }  ,opacity=1 ]  {$\tau _{2}$};
\draw (239,24.4) node [anchor=north west][inner sep=0.75pt]  [color={rgb, 255:red, 208; green, 2; blue, 27 }  ,opacity=1 ]  {$\tau _{2}$};
\draw (258,23.4) node [anchor=north west][inner sep=0.75pt]  [color={rgb, 255:red, 208; green, 2; blue, 27 }  ,opacity=1 ]  {$\tau _{2}$};
\draw (282,24.4) node [anchor=north west][inner sep=0.75pt]  [color={rgb, 255:red, 65; green, 117; blue, 5 }  ,opacity=1 ]  {$\tau _{3}$};
\draw (304,23.4) node [anchor=north west][inner sep=0.75pt]  [color={rgb, 255:red, 65; green, 117; blue, 5 }  ,opacity=1 ]  {$\tau _{3}$};
\draw (323,24.4) node [anchor=north west][inner sep=0.75pt]  [color={rgb, 255:red, 65; green, 117; blue, 5 }  ,opacity=1 ]  {$\tau _{3}$};
\draw (342,24.4) node [anchor=north west][inner sep=0.75pt]  [color={rgb, 255:red, 65; green, 117; blue, 5 }  ,opacity=1 ]  {$\tau _{3}$};
\draw (361,24.4) node [anchor=north west][inner sep=0.75pt]  [color={rgb, 255:red, 65; green, 117; blue, 5 }  ,opacity=1 ]  {$\tau _{3}$};
\draw (380,23.4) node [anchor=north west][inner sep=0.75pt]  [color={rgb, 255:red, 65; green, 117; blue, 5 }  ,opacity=1 ]  {$\tau _{3}$};
\draw (400,24.4) node [anchor=north west][inner sep=0.75pt]  [color={rgb, 255:red, 139; green, 87; blue, 42 }  ,opacity=1 ]  {$\tau _{4}$};
\draw (422,23.4) node [anchor=north west][inner sep=0.75pt]  [color={rgb, 255:red, 139; green, 87; blue, 42 }  ,opacity=1 ]  {$\tau _{4}$};
\draw (441,24.4) node [anchor=north west][inner sep=0.75pt]  [color={rgb, 255:red, 139; green, 87; blue, 42 }  ,opacity=1 ]  {$\tau _{4}$};
\draw (460,24.4) node [anchor=north west][inner sep=0.75pt]  [color={rgb, 255:red, 139; green, 87; blue, 42 }  ,opacity=1 ]  {$\tau _{4}$};
\draw (479,24.4) node [anchor=north west][inner sep=0.75pt]  [color={rgb, 255:red, 139; green, 87; blue, 42 }  ,opacity=1 ]  {$\tau _{4}$};
\draw (498,23.4) node [anchor=north west][inner sep=0.75pt]  [color={rgb, 255:red, 139; green, 87; blue, 42 }  ,opacity=1 ]  {$\tau _{4}$};
\draw (521,25.4) node [anchor=north west][inner sep=0.75pt]  [color={rgb, 255:red, 16; green, 52; blue, 224 }  ,opacity=1 ]  {$\tau _{5}$};
\draw (542,26.4) node [anchor=north west][inner sep=0.75pt]  [color={rgb, 255:red, 16; green, 52; blue, 224 }  ,opacity=1 ]  {$\tau _{5}$};
\draw (600,24.4) node [anchor=north west][inner sep=0.75pt]  [color={rgb, 255:red, 16; green, 52; blue, 224 }  ,opacity=1 ]  {$\tau _{5}$};
\draw (621,25.4) node [anchor=north west][inner sep=0.75pt]  [color={rgb, 255:red, 16; green, 52; blue, 224 }  ,opacity=1 ]  {$\tau _{5}$};
\draw (563,26.4) node [anchor=north west][inner sep=0.75pt]  [color={rgb, 255:red, 16; green, 52; blue, 224 }  ,opacity=1 ]  {$\tau _{5}$};
\draw (582,26.4) node [anchor=north west][inner sep=0.75pt]  [color={rgb, 255:red, 16; green, 52; blue, 224 }  ,opacity=1 ]  {$\tau _{5}$};
\draw (8,24.4) node [anchor=north west][inner sep=0.75pt]    {$S:$};
\draw (6,86.4) node [anchor=north west][inner sep=0.75pt]    {$T:$};

\end{tikzpicture}